\documentclass[a4paper]{article}

\usepackage{thmtools}
\usepackage{amsmath}
\usepackage{amsthm}
\usepackage{amssymb}
\usepackage{adjustbox}
\usepackage{authblk}

\usepackage{tabu}

\newtheorem{theorem}{Theorem}[section]
\newtheorem{lemma}[theorem]{Lemma}
\newtheorem{definition}[theorem]{Definition}
\newtheorem{corollary}[theorem]{Corollary}

\newtheorem{conjecture}[theorem]{Conjecture}

\usepackage{dcolumn}
\newcolumntype{d}[1]{D{.}{.}{#1}}

\usepackage{array}
\newcolumntype{C}[1]{>{\centering\let\newline\\\arraybackslash\hspace{0pt}}m{#1}}

\usepackage[lined, algoruled, linesnumbered]{algorithm2e}

\newenvironment{mylist}[1]{
\setbox1=\hbox{#1}
\begin{list}{}{
\setlength{\labelwidth}{\wd1}
\setlength{\leftmargin}{\wd1}
\addtolength{\leftmargin}{0em}
\addtolength{\leftmargin}{\labelsep}
\setlength{\rightmargin}{1em}}}{\end{list}}

\newcommand{\litem}[1]{\item[#1\hfill]}

\newcommand{\newLNT}{hyperloop nesting}
\newcommand{\NewLNT}{Hyperloop nesting}
\newcommand{\cvertex}{canonical}

\newcommand{\proofend}{$\Box$\\}

\title{Incremental Strong Connectivity and $2$-Connectivity in Directed Graphs
}

\author[1]{Loukas Georgiadis}
\author[2]{Giuseppe F.~Italiano}
\author[2]{Nikos Parotsidis}
\affil[1]{University of Ioannina, Greece. 
}
\affil[2]{University of Rome Tor Vergata, Italy. 
}

\date{}

\begin{document}

\maketitle

\begin{abstract}
In this paper, we present new incremental algorithms for maintaining data
structures that represent all connectivity cuts of size one in directed graphs
(digraphs), and the strongly connected components that result by the removal of each of those cuts.
We give a conditional lower bound that provides evidence that our algorithms may be tight up to a sub-polynomial factors.
As an additional result, with our approach we can also maintain dynamically
the $2$-vertex-connected components of a digraph during any sequence of edge insertions in a total of $O(mn)$ time.
This matches the bounds for the incremental maintenance of the $2$-edge-connected components of a digraph.
\end{abstract}

\section{Introduction}
\label{sec:inrtoduction}
A dynamic graph algorithm aims at updating efficiently the solution
of a problem after each update, faster than recomputing it from scratch. A dynamic graph problem is said to be \emph{fully dynamic} if the update operations
include both insertions and deletions of edges, and it
is said to be \emph{incremental} (resp., \emph{decremental}) if only
insertions (resp., deletions) are allowed.
In this paper, we present new incremental algorithms for some basic connectivity problems on directed graphs (digraphs), which were recently considered in the literature~\cite{2C:GIP:arXiv}.
Before defining the problems and stating our bounds, we need
some definitions.

Let $G=(V,E)$ be a digraph. $G$ is \emph{strongly connected} if there is a directed path from each vertex to every other vertex.
The \emph{strongly connected components} (in short \emph{SCCs}) of $G$ are its maximal strongly connected subgraphs.
Two vertices $u,v \in V$  are \emph{strongly connected} if they belong to the same SCC
of $G$.
An edge (resp., a vertex) of $G$ is a \emph{strong bridge} (resp., a \emph{strong articulation point}) if its removal increases the number of SCCs in the remaining graph. See Figure \ref{figure:SCC-example}.
Given two vertices $u$ and $v$, we say that an edge (resp., a vertex) of $G$ is a \emph{separating edge} (resp., a \emph{separating vertex}) for $u$ and $v$ if its removal leaves $u$ and $v$ in different SCCs.
Let $G$ be strongly connected: $G$ is \emph{$2$-edge-connected} (resp., \emph{$2$-vertex-connected}) if it has no strong bridges (resp., no strong articulation points).
Two vertices $u, v\in V$ are said to be \emph{$2$-edge-connected} (resp., \emph{$2$-vertex-connected}), denoted by $u \leftrightarrow_{\mathrm{2e}} v$ (resp., $u \leftrightarrow_{\mathrm{2v}} v$), if there are two edge-disjoint (resp., internally vertex-disjoint) directed paths from $u$ to $v$  and two edge-disjoint (resp., internally vertex-disjoint) directed paths from $v$ to $u$. (Note that a path from $u$ to $v$ and a path from $v$ to $u$ need not be edge-disjoint or internally vertex-disjoint).
A \emph{$2$-edge-connected component}  (resp., a \emph{$2$-vertex-connected component}) of a digraph $G=(V,E)$ is defined as a maximal subset $B \subseteq V$ such that $u \leftrightarrow_{\mathrm{2e}} v$ (resp., $u \leftrightarrow_{\mathrm{2v}} v$) for all $u, v \in B$.
Given a digraph $G$, we denote by
$G\setminus e$ (resp., $G\setminus v$) be the digraph obtained after deleting edge $e$ (resp., vertex $v$) from $G$.

\begin{figure}[t!]
\begin{center}
\vspace{-.9cm}
\includegraphics[trim={0 0 0 2cm}, clip=true, width=1.0\textwidth]{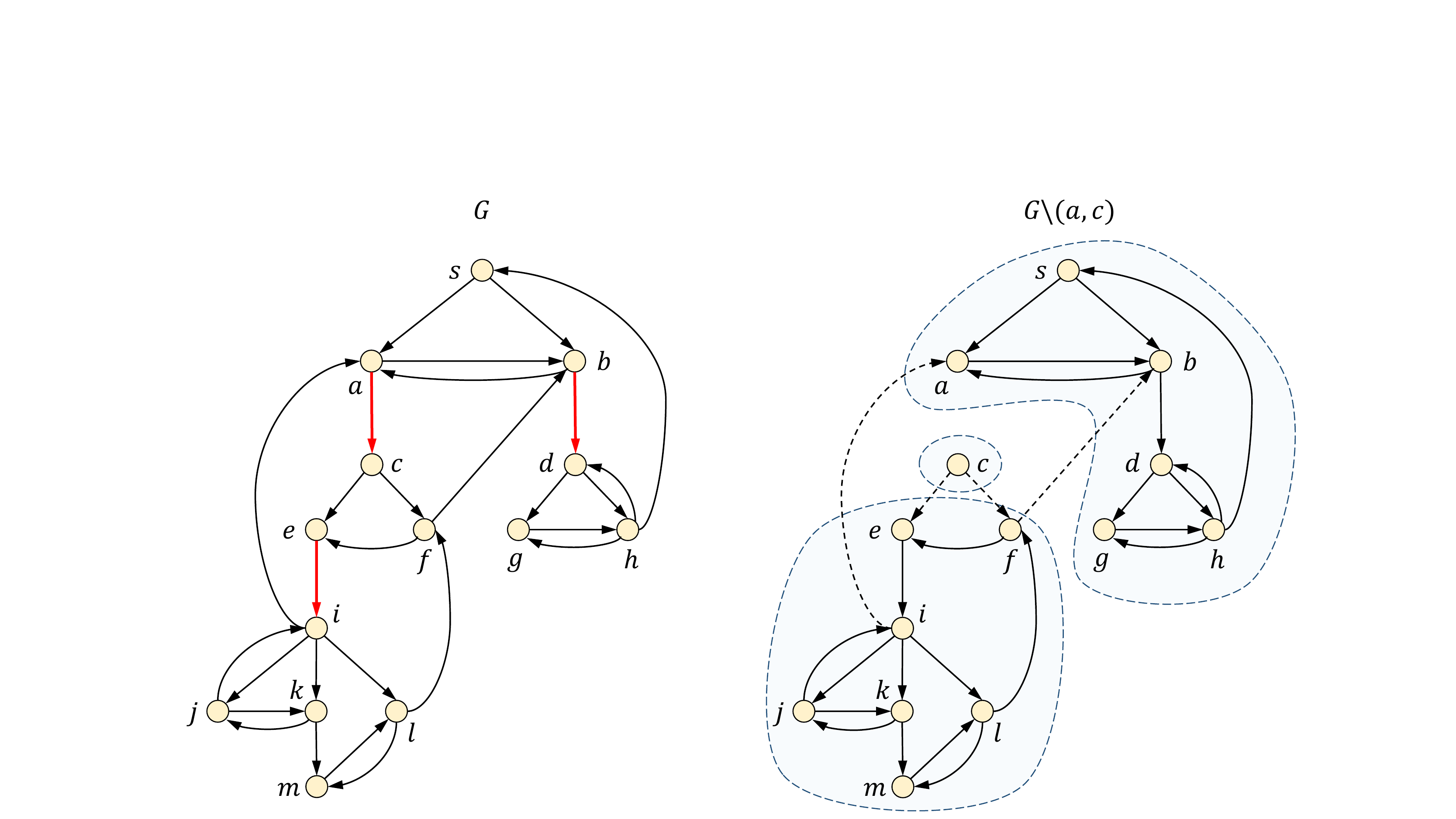}	
\caption{A strongly connected digraph $G$ with strong bridges shown in red (better viewed in color), and the SCCs of $G \setminus e$ after the deletion of
 the strong bridge $e=(a,c)$.}
\label{figure:SCC-example}
\end{center}
\vspace{-.7cm}
\end{figure}

Let $G=(V,E)$ be a strongly connected graph. In very recent work~\cite{2C:GIP:arXiv}, we presented an $O(n)$-space data structure that, after a linear-time preprocessing, 
is able to answer in asymptotically optimal (worst-case) time all the following queries on a static digraph:
\begin{itemize}
	\item Report in $O(1)$ time the total number of SCCs in $G\setminus e$ (resp., $G\setminus v$), for any query edge $e$  (resp., vertex $v$) in $G$.
	\item Report in $O(1)$ time the size of the largest and of the smallest SCCs in $G\setminus e$ (resp., $G\setminus v$), for any query edge $e$  (resp., vertex $v$) in $G$.
	\item Report in $O(n)$ time all the SCCs of $G\setminus e$ (resp., $G\setminus v$), for any query edge $e$ (resp., vertex $v$).
	\item Test in $O(1)$ time whether two query vertices $u$ and $v$ are strongly connected in $G\setminus e$ (resp., $G\setminus v$), for any query edge $e$ (resp., vertex $v$).
	\item For any two query vertices $u$ and $v$ that are strongly connected in $G$,
	report all edges $e$ (resp., vertices $v$) such that $u$ and $v$ are not strongly connected in $G\setminus e$ (resp., $G\setminus v$)
	in time $O(k+1)$, where $k$ is the number of separating edges (resp., separating vertices).
\end{itemize}

As pointed out in~\cite{2C:GIP:arXiv,Paudel2017Computing}, this data structure is motivated by applications in many areas, including
computational biology~\cite{Gunawardena12,MihalakUY15}
social network analysis~\cite{Kempe:2003,Ventresca2015},
network resilience~\cite{Shen:2013} and network
immunization~\cite{Aspnes:2006,Immune:PhysRevLett,Kuhlman:2010}.

A dynamic version of the aforementioned data structure can be used to monitor the critical components (i.e., edges and vertices) whose removal disrupts the underlying graphs, in graphs that change over time.
An ideal scenario is to design efficient algorithms in the fully dynamic setting.
However, we show that no data structure that can answer any of the queries that we consider in sublinear time in the number of edges, can be maintained faster than recomputing the data structure from scratch unless a widely believed conjecture is proved wrong.
There are real-word dynamic networks where edge deletion occur rarely, in which case the incremental setting finds applications.
Such networks include, for instance, communication networks, road networks, the power grid.

\smallskip
\noindent{\bf Our Results.}
We show a conditional lower bound for the fully dynamic version of this problem. More specifically, 
let $G = (V,E)$ be a digraph with $n$ vertices that undergoes $m$ edge
updates from an initially empty graph.  
We prove that any fully dynamic algorithm that can answer any of the queries considered here requires either $\Omega(m^{1-o(1)})$ amortized update time, or $\Omega(m^{1-o(1)})$ query time, 
unless the Strong Exponential Time Hypothesis~\cite{impagliazzo1999complexity,impagliazzo1998problems} is false.

Motivated by this hardness result, we focus on the incremental version of this problem. 
We present an incremental version of the data structure introduced in~\cite{2C:GIP:arXiv}, which can be maintained throughout a sequence of edge insertions. In particular,
we show how to maintain a digraph $G$ undergoing edge insertions in a total of $O(mn)$ time, where $n$ is the number of vertices and $m$ the number of edges after all insertions, so that all the queries we consider can be answered in asymptotically optimal (worst-case) time after each insertion.
As an additional result, with our approach we can also
maintain
the $2$-vertex-connected components of a digraph during any sequence of edge insertions in a total of
$O(mn)$ time.
After every insertion we can test whether two query vertices are $2$-vertex-connected and, whenever the answer is negative, produce a separating vertex (or an edge) for the two query vertices.
This matches the bounds for the
incremental maintenance of the $2$-edge-connected components of a digraph~\cite{GIN16:ICALP}.

Before our work,
no algorithm for all those problems  was faster than recomputing the solution from scratch after each edge insertion, which yields a total of
$O(m^2)$. Our algorithms improve substantially over those bounds.
In addition, we show a
conditional lower bound for the total update time of an incremental
data structure that can answer queries of the form ``are $ u $ and $ v $ strongly connected in $G \setminus e$'', where $u,v\in V$, $e\in E$.
In particular, we prove that the existence of a data structure
that supports the aforementioned queries with total update time $O((mn)^{1-\epsilon})$ (for some constant $\epsilon > 0$).
Therefore, a polynomial improvement of our bound leads to a breakthrough.

\smallskip
\noindent{\bf Related Work.}
Many efficient algorithms for several dynamic graph problems have been proposed in the literature, including dynamic connectivity~\cite{HK99,HLT01,Nanongkai2017Dynamic,PT07,Thorup2000}, minimum spanning trees~\cite{EGIN97,F85,HK01,HLT01,Nanongkai2017Dynamic}, edge/vertex connectivity~\cite{EGIN97,HLT01} on undirected graphs, and transitive closure~\cite{DI08,HK95,King99,Sankowski2004Dynamic} and shortest paths~\cite{Abraham2017Fully,DI04,King99,Thorup04} on digraphs.
Dynamic problems on digraphs are known to be harder than on undirected graphs and most of the dynamic algorithms on undirected graphs have polylog update bounds, while dynamic algorithms on digraphs have higher polynomial update bounds.  The hardness of dynamic algorithms on digraphs has been recently supported also by conditional lower bounds~\cite{AW14}.

In \cite{decdom17}, the decremental version of the data structure considered in this paper is presented. The total time and space required to maintain decrementally the data structure is $O(mn \log{n})$ and $O(n^2 \log n)$, respectively: here $m$ is the number of edges in the initial graph.
We remark that our incremental algorithms are substantially different from decremental algorithms of~\cite{decdom17}, and indeed the techniques that we use here are substantially different from~\cite{decdom17}.
More specifically, the main approach of~\cite{decdom17}
is to maintain the SCCs in $G\setminus v$ for each $v\in V$, by carefully combining $n$ appropriate instances of the decremental SCCs algorithm from~\cite{scc-decomposition}.
This allows to maintain
decrementally the dominator tree
in $O(mn \log{n})$ total time and $O(n^2 \log n)$ space.
On the contrary,
in the incremental setting it is already known how to maintain dominator trees, and the main challenge is
to maintain efficiently information about nesting loops throughout edge insertions. This allows us to achieve better bounds than in the decremental setting~\cite{decdom17}: namely, $O(mn)$ total time and $O(m+n)$ space.

In \cite{GIN16:ICALP} we presented an incremental algorithm that maintains the $2$-edge-connected components of a directed graph with $n$ vertices through any sequence of edge insertions in a total of $O(mn)$ time, where $m$ is the number of edges after all insertions. After each insertion, we can test in constant time if two query vertices $v$ and $w$ are $2$-edge-connected, and if not we can produce in constant time a ``witness'' of this property, by exhibiting
an edge that is contained in all paths from $v$ to $w$ or in all paths from $w$ to $v$.

\smallskip
\noindent{\bf Our Technical Contributions.}
Our fist contribution is to dynamize the recent data structure in~\cite{2C:GIP:arXiv}, which hinges on two main building blocks: dominator trees and loop nesting trees (which are reviewed in Section~\ref{sec:definitions}).
While it is known how to maintain efficiently dominator trees in the incremental setting~\cite{dyndom:2012},
the incremental maintenance of loop nesting trees is a challenging task.
Indeed, loop nesting trees are heavily based on depth-first search, and   maintaining efficiently a dfs tree of a digraph under edge insertions has been an elusive goal: no efficient solutions are known up to date, and incremental algorithms are available only in the restricted case of DAGs~\cite{FGN97}.
To overcome these inherent difficulties, we manage to define a new notion of strongly connected subgraphs of a digraph, which is still relevant for our problem and is independent of
depth first search.
This new notion is based on some specific nesting loops, which define a laminar family. One of the technical contributions of this paper is to show how to maintain efficiently this family of nesting loops during edge insertions. We believe that this result might be of independent interest, and perhaps it might shed further light to the incremental dfs problem on general digraphs.

Our second contribution, the incremental maintenance of the $2$-vertex-connected components of a digraph, completes the picture on incremental $2$-connectivity on digraphs by complementing the recent $2$-edge connectivity results of \cite{GIN16:ICALP}.
We remark that $2$-vertex connectivity in digraphs is much more difficult than $2$-edge connectivity, since it is plagued with several degenerate special cases, which are not only more tedious but also more cumbersome to deal with. For instance, $2$-edge-connected components partition the vertices of a digraph, while $2$-vertex-connected components do not.
Furthermore, two vertices $v$ and $w$ are $2$-edge-connected if and only if the removal of any edge leaves $v$ and $w$ in the same SCC. Unfortunately, this property no longer holds for $2$-vertex connectivity, as for instance two mutually adjacent vertices are always left in the same strongly connected component by the removal of any other vertex, but they are not necessarily $2$-vertex-connected.

\section{Dominator trees, loop nesting trees and auxiliary components}
\label{sec:definitions}

In this section we review the two main ingredients used by the recent framework in~\cite{2C:GIP:arXiv}: dominator trees and loop nesting trees. As already mentioned in the introduction, one of the main technical difficulties behind our approach is that the incremental maintenance of loop nesting trees seems an elusive goal. We then review the notion of auxiliary components, which is used in Section \ref{sec:canonical-loop-nesting} to overcome this difficulty.
We remark that both dominator trees and auxiliary components can be maintained efficiently during edge insertions \cite{dyndom:2012,GIN16:ICALP}.

Throughout, we assume that the reader is familiar with standard graph terminology, as contained
for instance in~\cite{clrs}.
Given a rooted tree, we denote by $T(v)$ the
set of descendants of $v$ in $T$.
Given a digraph $G=(V,E)$, and a set of vertices $S \subseteq V$, we denote by $G[S]$ the subgraph induced by $S$.
Moreover, we use $V(S)$ and $E(S)$ to refer to the vertices of $S$ and to the edges adjacent to $S$, respectively.
The \emph{reverse digraph} of $G$, denoted by $G^R=(V, E^R)$, is obtained
by reversing the direction of all edges.
A \emph{flow graph} $F$ is a directed graph (digraph) with a distinguished start vertex $s \in V(F)$, where all vertices in $V(F)$ are reachable from $s$ in $F$.
We denote by $G_s$ the subgraphs of $G$ induced by the vertices that are reachable from $s$; that is, $G_s$ is a flow graph with start vertex $s$.
Respectively, we denote by $G^R_s$ the subgraphs of $G^R$ induced by the vertices that are reachable from $s$.
If $G$ is strongly connected, all vertices are reachable from $s$ and reach $s$, so we can view both $G$ and $G^R$ as flow graphs with start vertex $s$.

\smallskip
\noindent{\bf Dominator trees.}
A vertex $v$ is a \emph{dominator} of a vertex $w$ ($v$ \emph{dominates} $w$) if every path from $s$ to $w$ contains $v$.
%
The dominator relation in $G$ can be represented by a tree rooted at $s$, the \emph{dominator tree} $D$, such that $v$ dominates $w$ if and only if $v$ is an ancestor of $w$ in $D$.
See Figure \ref{figure:dominators-example}.
We denote by $dom(w)$  the set of vertices that dominate $w$. Also, we let $d(w)$ denote the parent of a vertex $w$ in $D$.
Similarly, we can define the dominator relation in the flow graph $G_s^R$, and let $D^R$ denote the dominator tree of $G_s^R$, and $d^R(v)$ the parent of $v$ in $D^R$.
The dominator tree of a flow graph can be computed in linear time, see, e.g.,~\cite{domin:ahlt,dominators:bgkrtw}.
%
%
%
An edge $(u,v)$ is a \emph{bridge} of a flow graph $G_s$ if all paths from $s$ to $v$ include $(u,v)$.\footnote{Throughout the paper, to avoid confusion we use consistently the term  \emph{bridge} to refer to a bridge of a flow graph and the term \emph{strong bridge} to refer to a strong bridge in the original graph.}
Let $s$ be an arbitrary start vertex of $G$.
As shown in \cite{Italiano2012}, an edge $e=(u,v)$ is strong bridge of $G$ if and only if it is either a bridge of $G_s$ or a bridge of $G_s^R$.
As a consequence,
all the strong bridges of
$G$ can be obtained from the bridges of the flow graphs $G_s$ and $G_s^R$, and thus there can be at most $2(n-1)$ strong bridges overall.
\begin{figure}[t!]
	\begin{center}
		\vspace{-.9cm}
		\includegraphics[trim={13cm 3cm 0cm 3cm}, clip=true, width=0.7\textwidth]{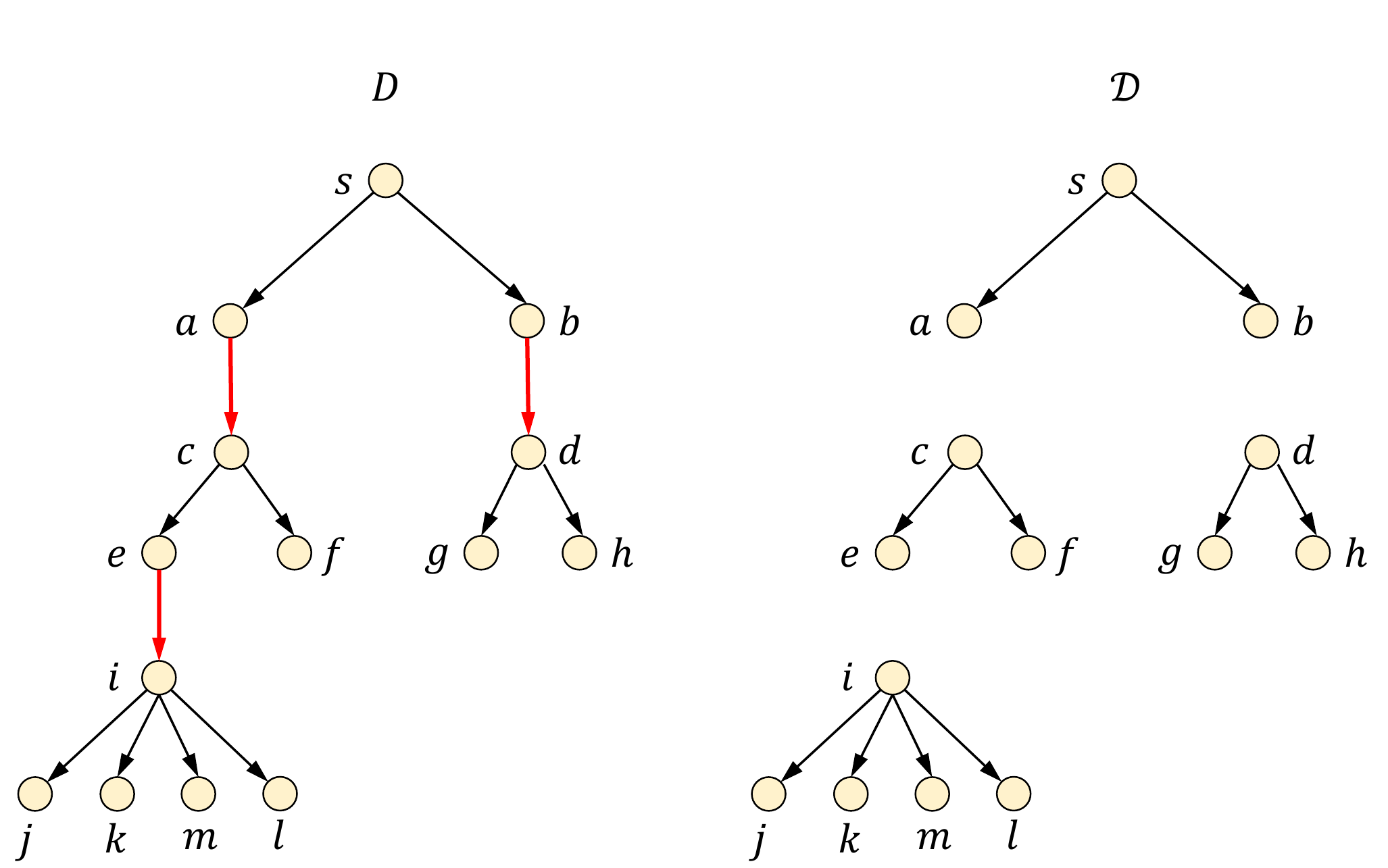}	
		\caption{The dominator tree $D$ (on the left) of the digraph of Figure \ref{figure:SCC-example} with start vertex $s$, and its bridge decomposition $\mathcal{D}$ (on the right).
		}
		\label{figure:dominators-example}
	\end{center}
\end{figure}
After deleting from the dominator trees $D$ and $D^R$ respectively the bridges of $G_s$ and $G_s^R$, we obtain the \emph{bridge decomposition}  of $D$ and $D^R$ into forests $\mathcal{D}$ and $\mathcal{D}^R$.
Throughout the paper, we denote by $D_u$ (resp., $D_u^R$) the tree in $\mathcal{D}$ (resp., $\mathcal{D}^R$) containing vertex $u$, and by $r_u$ (resp., $r^R_u$) the root of $D_u$ (resp., $D_u^R$).
The following lemma from \cite{2ECB} holds for a flow graph $G_s$ of a strongly connected digraph $G$ (and hence also for the flow graph $G_s^R$ of $G^R$).

\begin{lemma}
\label{lemma:partition-paths} \emph{(\cite{2ECB})}
Let $G$ be a strongly connected digraph and let $(u,v)$ be a strong bridge of $G$. Also, let $D$ 
be the dominator tree of the
flow graph $G_s$,
for an arbitrary start vertex $s$.
Suppose $u=d(v)$. Let $w$ be any vertex that is not a descendant of $v$ in $D$. Then there is path from $w$ to $v$ in $G$ that does not contain any proper descendant of $v$ in $D$. Moreover, all simple paths in $G$ from $w$ to any descendant of $v$ in $D$ must contain the edge $(d(v),v)$.
\end{lemma}

\vspace{-.07cm}

\noindent{\bf Loop nesting forests.}
Let $G$ be a digraph, and $G_s$ the flow graph with an arbitrary start vertex $s$.
A \emph{loop nesting forest} represents a hierarchy of strongly connected subgraphs of $G_s$~\cite{st:t}, 
defined with respect to a dfs tree $T$ of $G_s$, rooted at $s$, as follows.
For any vertex $u$, $\mathit{loop}(u)$ is the set of all descendants $x$ of $u$ in $T$ such that there is a path from $x$ to $u$ in $G$ containing only descendants of $u$ in $T$.
Any two vertices in $\mathit{loop}(u)$ reach each other.
Therefore, $\mathit{loop}(u)$ induces a strongly connected subgraph of $G$; it is the unique maximal set of descendants of $u$ in $T$ that does so.
The $\mathit{loop}(u)$ sets form a laminar family of subsets of $V$:
for any two vertices $u$ and $v$, $\mathit{loop}(u)$ and $\mathit{loop}(v)$ are either disjoint or nested.
The
\emph{loop nesting forest} $H$ of $G_s$, with respect to $T$, is the forest in which the parent of any vertex $v$, denoted by $h(v)$, is the nearest proper ancestor $u$ of $v$ in $T$ such that $v \in \mathit{loop}(u)$ if there is such a vertex $u$, and null otherwise.
Then $\emph{loop}(u)$ is the set of all descendants of vertex $u$ in $H$, which we also denote as $H(u)$ (the subtree of $H$ rooted at vertex $u$).
A loop nesting forest can be computed in linear time~\cite{dominators:bgkrtw,st:t}.
When $G$ is strongly connected,
each vertex is contained in a loop, and $H$ is a tree, rooted at $s$.
Therefore, we refer to $H$ as the \emph{loop nesting tree} of $G_s$ (see Figure~\ref{figure:loop-nesting-example}).

\begin{figure}[t!]
\begin{center}
\vspace{-.9cm}
\includegraphics[trim={1cm 5cm 1cm 0cm}, clip=true, width=\textwidth]{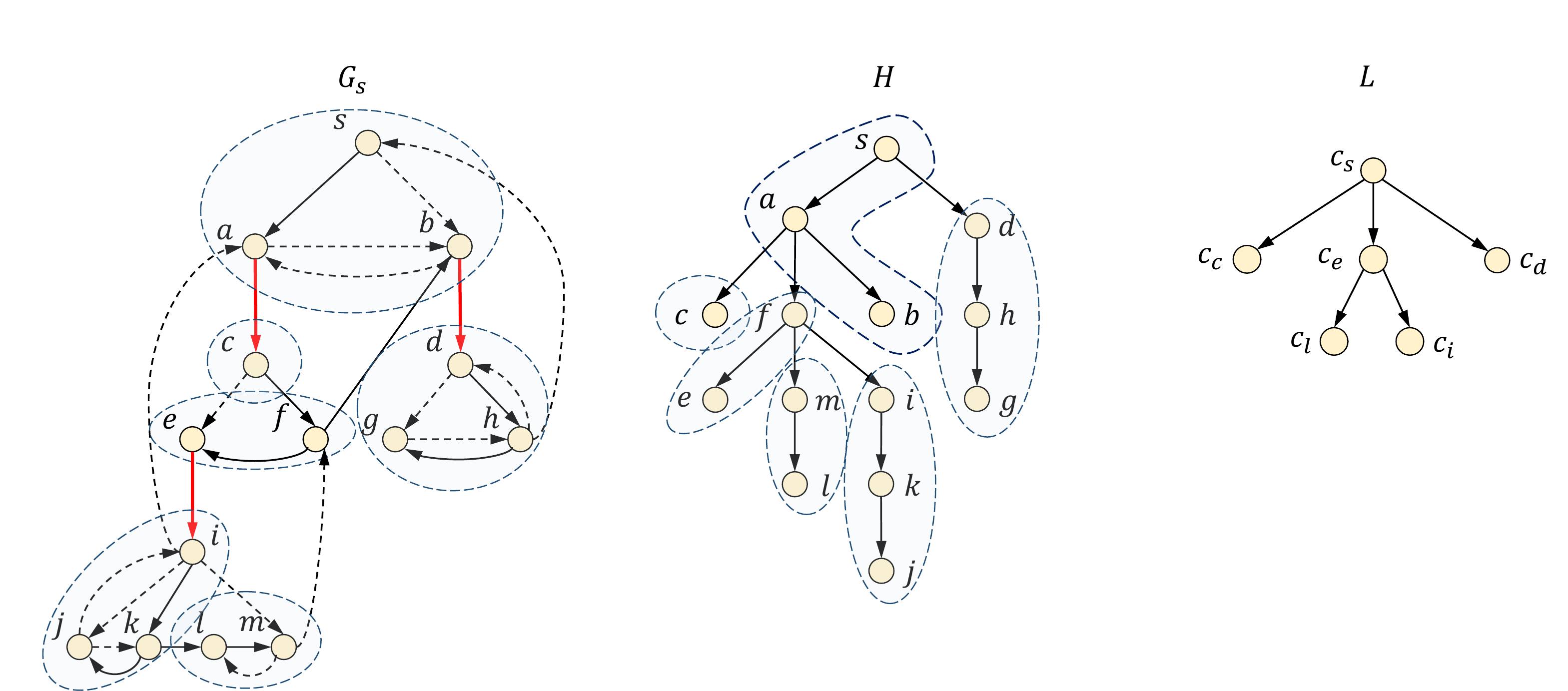}	
\caption{The flow graph $G_s$ of the graph of Figure \ref{figure:SCC-example} with solid dfs edges (left); the loop nesting tree $H$ of $G_s$ generated by the dfs traversal on the left (middle); the \newLNT{} tree $L$ of $G_s$ (right). The bridges of $G_s$ are shown red. The grouped vertices in both $G_s$ and $L$ represent the auxiliary components of $G_s$.}
\label{figure:loop-nesting-example}
\end{center}
\vspace{-.7cm}
\end{figure}

\smallskip
\noindent{\bf Auxiliary components.}
Let $G_s$ be a flow graph and $D$ and $\mathcal{D}$ be the dominator tree and the bridge decomposition of $G_s$, respectively.
Let $e=(u,v)$ be a bridge of the flow graph $G_s$.
We say that an SCC $C$ in $G[D(v)]$ is an \emph{$e$-dominated component} of $G$.
We also say that $C \subseteq V$ is a \emph{bridge-dominated component} if it is an $e$-dominated component for some bridge $e$:
bridge-dominated components form a laminar family \cite{GIN16:ICALP}.
An \emph{auxiliary component} of $G_s$ is a maximal subset of vertices $C \cap D_v$ such that $C$ is a subset of a $(d(r_v),r_v)$-dominated component.
Each auxiliary component $C$ is represented by an arbitrarily chosen vertex $u\in C$, which we call the \emph{\cvertex{} vertex} of $C$.
For each vertex $v\in C$, we refer to the \cvertex{} vertex of $C$ by $c_v$.
That is, if $u$ is the \cvertex{} vertex of an auxiliary component then $c_u=u$.
Following the bridge decomposition $\mathcal{D}$ of the dominator tree $D$ of $G_s$, the auxiliary components are defined with respect to the start vertex $s$.

\section{\NewLNT{} forest}
\label{sec:canonical-loop-nesting}
In this section, we introduce the new notion of \emph{\newLNT{} forest}, which,
differently from loop nesting forest, can be maintained efficiently during edge insertions, as we show in Section \ref{sec:canonical-loop-update}.
Given  a \cvertex{} vertex $v\not = c_s$, we define the \emph{hyperloop of $v$}, and denote it by $hloop(v)$, as the set of \cvertex{} vertices that are in the same $(d(r_v),r_v)$-dominated component as $v$.
As a special case, all \cvertex{} vertices that are strongly connected to $s$ are in the hyperloop $hloop(c_s)$.
It can be shown that
hyperloops form a laminar family of subsets of $V$, with respect to the start vertex $s$: for any two \cvertex{} vertices $u$ and $v$, $hloop(u)$ and $hloop(v)$ are either disjoint or nested (i.e., one contains the other).
This property allows us to define the \emph{\newLNT{} forest} $L$ of $G_s$ as follows.
The parent $\ell(v)$ of a \cvertex{} vertex $v$ in $L$ is the (unique) \cvertex{} vertex $u$, $u\notin D(r_v)$, with the largest depth in $D$, such that $v \in hloop(u)$.
If there is no vertex $u\notin D(r_v)$, such that $v \in hloop(u)$, then  $\ell(v) = \emptyset$; notice that in this case $v$ is not strongly connected to $s$ as well.
See Figure \ref{figure:loop-nesting-example}.
Then, $hloop(u)$ is the set of all descendants of a \cvertex{} vertex $u$ in $L$, which we also denote as $L(u)$ (the subtree of $L$ rooted at vertex $u$).
Similarly to the loop nesting forest,
the \newLNT{} forest of a strongly connected digraph is a tree.

\begin{figure}[t!]
	\begin{center}
		\includegraphics[trim={0cm 0cm 0cm 2cm}, clip=true, width=1.0\textwidth]{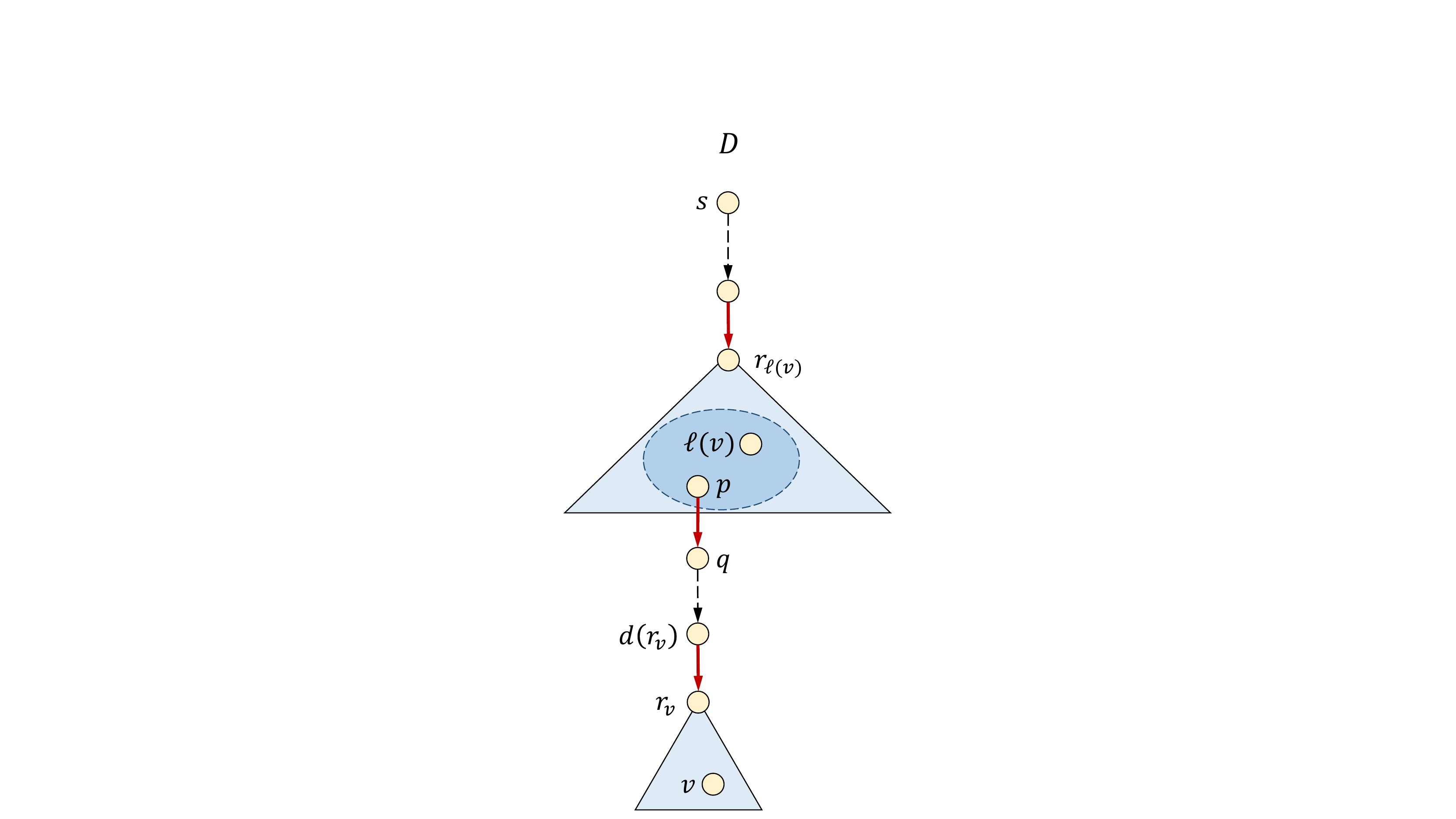}	
		\caption{A representation of the relation between a \cvertex{} vertex $v$ and its parent in $L$ with respect to the bridge decomposition of $D$.}
		\label{figure:canonical-loop-nesting-tree-parent-relation}
	\end{center}
	\vspace{-.7cm}
\end{figure}

We begin the study of the hyperloop nesting forest
by showing that it is unique, and thus, depends solely on the structure of the graph.
(We consider a fixed choice of the canonical vertices of the auxiliary components.)

\begin{lemma}
The hyperloop nesting forest of a flowgraph $G_s$ is unique.
\end{lemma}
\begin{proof}
Let $v$ be a canonical vertex of $G_s$.
By the definition of the hyperloop nesting forest $L$, the parent $\ell(v)$ of $v$ in $L$ is the canonical vertex $u \not \in D(r_v)$ with maximum depth in $D$, such that $u$ and $v$ are in the same $(d(r_{u}),r_{u})$-dominated component. Then, the fact that both $u$ and $v$ are descendants of $r_u$ implies that $u$ is unique and so the lemma follows.
\proofend
\end{proof}

Given a vertex $u$ in a flow graph $G_s$, we define its level, denoted by $level(u)$, to be the number of bridges $(v,w)$ of $G_s$ such that $w$ is an ancestor of $u$ in $D$.
In other words, the level of $u$ equals the number of strong bridges that appear in all paths from $s$ to $u$ in $G_s$.
As a result, all vertices in the same tree of the bridge decomposition have the same level.
In the next lemma we show that each \cvertex{} vertex has at most one ancestor in $L$ at each level.

\begin{lemma}
	\label{lemma:unique-level}
	Let $G_s$ a flow graph and let $u$ be a \cvertex{} vertex of $G_s$.
	All ancestors of $u$ in the \newLNT{} forest have unique level.
\end{lemma}
\begin{proof}
	Let $w$ and $v$ be two distinct ancestors of $u$ in $L$ such that $level(w) = level(v)$.
	By the definition of the \newLNT{} forest both $r_w$ and $r_v$ are ancestors of $u$ in $D$.
	Then, $r_w = r_v$.
	Assume by contradiction that $u$ is strongly connected with both $w$ and $v$ in $G[D(r_w)]$.
	Then $w$ and $v$ are strongly connected in $G[D(r_w)]$.
	By the definition of the auxiliary components, $w$ and $v$ are in the same auxiliary component, and thus $c_w = c_v$.
	A contradiction to the fact that both $w$ and $v$ are \cvertex{} vertices.
	The lemma follows. \proofend
\end{proof}

The following lemma characterizes the relationship between the loop nesting forest $H$ and the \newLNT{} forest $L$ of a flow graph $G_s$.
More specifically,  it
shows that $L$ can be obtained from $H$ by contracting all the vertices of each auxiliary component into their \cvertex{} vertex.
This yields immediately a linear-time algorithm to compute the \newLNT{} forest of a flow graph $G_s$: we first compute a loop nesting forest $H$ of $G_s$  \cite{dominators:bgkrtw,st:t} and then contract each vertex $v$ to $c_v$ in $H$.
In the following we denote by $h_v$ the unique ancestor of $v$ in $H$, such that $h_v \in D_v$ and $h(h_v) \notin D_v$.
If $v\in D_s$, it follows that $h_v = s$.
We can compute $h_v$, for all $v\in V$ in $O(n)$ time \cite{2C:GIP:arXiv}.

\begin{lemma}
	\label{lemma:loops-nestings-relation}
	For every vertex $v$, the following hold:
	\begin{itemize}
		\item The canonical vertices of $v$ and $h_v$ are the same.
		\item The canonical vertex of the parent of $c_v$ in $L$ is the canonical vertex of the parent of $h_v$ in $H$, including the case where $h(h_v) = \emptyset$.
	\end{itemize}
\end{lemma}
\begin{proof}
	Let $T$ be the dfs traversal that generated $H$.
	In the case where $v\in D_s$, it trivially follows that $c_{h_v} = c_v = c_s$.
	Now assume that $v\notin D_s$.
	Note that $r_v$ is an ancestor of $h_v$ in $T$ since all paths from $s$ to $h_v$ go through $r_v$ by Lemma~\ref{lemma:partition-paths}.
	Therefore, also all the descendants of $h_v$ in $H$ are descendant of $r_v$ in $T$.
	The fact that $h_v$ is an ancestor of $v$ in $H$ it means that $v$ is a descendant of $h_v$ in $T$ and $v$ has a path $P$ to $h_v$ using only descendants of $h_v$ in $T$.
	The path $P$ cannot contain vertices that are not in $D(r_v)$ since otherwise by Lemma~\ref{lemma:partition-paths} they contain $r_v$, which contradicts the fact that all vertices on $P$ contain descendants of $h_v$ in $T$ (recall that $r_v$ is an ancestor of $h_v$ in $T$).
	It follows that $h_v$ has a path to $v$ in $G[D(r_v)]$ and $v$ has a path to $h_v$ in $G[D(r_v)]$.
	Thus, $c_v = c_{h_v}$ by the definition of the auxiliary components.
	Exactly the same argument can be applied on the reverse graph to show that $c^R_{h^R_v} = c^R_v$.
	
	Consider now $h(h_v)$.
	In the case where $v\in D_s$, it trivially follows that $h(c_{h_v}) = h(c_v) = \emptyset$.
	Therefore in the following we assume that $v\notin D_s$.
	First, we deal with the case where $c_{h(h_v)} = \emptyset$.
	Since $v\notin D_s$, the only case that this can happen is when $h(h_v) = \emptyset$.
	That means, there is no ancestor $w$ of $h_v$ in $T$, such that $h_v$ has a path to $w$ using only vertices in $T(w)$.
	We show that this means, there is also no vertex $z\notin D(r_v)$, such that $v$ and $z$ are strongly connected in $G$.
	Assume, for the sake of contradiction that there is such vertex $z$, and let $C$ be the SCC containing both $h_v$, and therefore also $v$, and $z$.
	Moreover, let vertex $z'\in C, z' \not=h_v$ be the vertex that is visited first by the dfs that generated $T$.
	Then all vertices in $C$ are descendants of $z'$ in $T$ since they are all reachable from $z'$ and they were not visited by the dfs before $z'$.
	Hence, $h_v$ is a descendant of $z'$ in $T$.
	A contradiction to the fact that $h(h_v) = \emptyset$.
	Therefore, there is no vertex $z\notin D(r_v)$, such that $v$ and $z$ are strongly connected in $G \supset G[D(r_z)]$.
	Thus, $\ell(c_v) = \emptyset$.
	Now we consider the case where $c_{h_v} = c_v \not= \emptyset$.
	The fact that $r_{h(h_v)}$ is an ancestor of $h(h_v)$ in $D$ implies that $r_{h(h_v)}$ is an ancestor of $h(h_v)$ also in $T$.
	By the definition of $H$ it follows that $h_v$ is a descendant of $h(h_v)$ in $T$, and $h_v$ has a path $P$ to $h(h_v)$ using only descendants of $h(h_v)$ in $T$.
	The path $P$ cannot contain vertices that are not in $D(r_{h(h_v)})$ since otherwise, by Lemma~\ref{lemma:partition-paths}, they contain $r_{h(h_v)}$, which contradicts the fact that all vertices on $P$ contain descendants of $h(h_v)$ in $T$ (recall that $r_{h(h_v)}$ is an ancestor of $h(h_v)$ in $T$).
	Therefore $h(h_v)$ and $h_v$ are strongly connected in $G[r_{h(h_v)}]$, and by definition so do $c_{h(h_v)}$ and $c_{h_v}$.
	Hence, $c_{h(h_v)}$ is an ancestor of $c_{h_v}$ in $L$.
	Now we show that there is no other \cvertex{} vertex $w\not=h(h_v)$ such that $level(c_w)>level(h(h_v))$ and $c_w$ is an ancestor of $c_{h_v}$ in $L$.
	This implies that $\ell(c_v) = c_{h(h_v)}$.
	Assume, for the sake of contradiction, that there is such a vertex $w$.
	Then $h_v$ and $c_w$ are in the same SCC $C$ in $G[D(r_{c_w})]$.
	Let $z \in C$ be the first among the vertices in $C$ visited by the dfs that generated $T$.
	Then all vertices in $C$ are descendants of $z$ in $T$ since they are all reachable from $z$ and they were not visited by the dfs before $z$.
	Thus, $h_v$ is a descendant of $z$ in $H$.
	Now we show that this implies that $z$ is also a descendant of $h(h_v)$ in $H$.
	By the definition of the \newLNT{} forest $r_z$ and $r_{h(h_v)}$ are ancestors of $h_v$ in $D$.
	Furthermore, since $level(z) > level(h(h_v))$ it holds that $r_{h(h_v)}$ is an ancestor of $r_z$ in $D$.
	By the fact that $h(h_v)$ and $z$ are ancestors of $h_v$ in $T$ it holds that $h(h_v)$ and $z$ have ancestor-descendant relation, and hence, $z$ is a descendant of $h(h_v)$ in $T$ (as $r_z$ and $r_{h(h_v)}$ are ancestors of $h_v$ in $D$ and $level(z) > level(h(h_v))$).
	Since $h_v$ has a path to $h(h_v)$ using only descendants of $h(h_v)$ in $T$, it follows that $z$ also has a path to $h(h_v)$ using only descendants of $h(h_v)$ in $T$.
	Thus, $z$ is a descendant of $h(h_v)$ in $H$.
	In summary, $z$ if an ancestor of $h_v$ and a descendant of $h(h_v)$ in $H$; a contradiction to the definition of $h(h_v) \not=z$.
	This concludes the lemma.
	\proofend
\end{proof}

\section{Updating the dominator tree after an edge insertion }
\label{sec:incrementa-dominator-and-auxiliary-components}

In this section, we briefly review the algorithm from \cite{dyndom:2012} that updates the dominator tree of a flow graph $G_s$ after an edge insertion. Let $G_s$ be a flow graph with start vertex $s$.
Let $(x,y)$ be the edge to be inserted.
Let $D$ be the dominator tree of $G_s$ before the insertion; we let $D'$ be the the dominator tree of $G'_s$.
In general, for any function $f$ on $V$, we let $f'$ be the function after the update.
We say that vertex $v$ is \emph{$D$-affected} by the update if $d(v)$ (its parent in $D$) changes, i.e., $d'(v) \not= d(v)$.
We let $\mathit{nca_D}(x,y)$ denote the nearest common ancestor of $x$ and $y$ in the dominator tree $D$.

\begin{lemma}
\label{lemma:insert-affected} \emph{(\cite{irdom:rr94})}
If $v$ is $D$-affected, then it becomes a child of $\mathit{nca_D}(x,y)$ in $D'$, i.e., $d'(v)=\mathit{nca_D}(x,y)$.
\end{lemma}

We say that a vertex is \emph{$D$-scanned}
if it is a descendant of a $D$-affected vertex after an edge insertion.
Note that every $D$-affected vertex is also $D$-scanned since each vertex is a descendant of itself in $D$.
There are two key ideas behind the incremental dominators algorithm. First, the algorithm updates $D'$ in time proportional to number of the edges incident to $D$-scanned vertices.
Second, after an edge insertion,
all $D$-scanned vertices decrease their depth in $D'$ by at least one.
Since throughout a sequence of edge insertions the depth of a vertex can only decrease, each vertex can be $D$-scanned at most $n-1$ times, and thus
the algorithm examines at most $(n-1)$ times the adjacency list of each vertex.
This implies that
the algorithm 
has $O(mn)$ total update time, where $m$ is the number of edges in the graph after all insertions.

We now prove the following
lemma, which shows that after the insertion of an edge $(x,y)$, the ancestors and descendants of the vertices $v \notin D'(nca_D(x,y))$ do not change. We use this later later on in our incremental algorithm for maintaining the \newLNT{} forest of a flow graph.

\begin{lemma}
	\label{lem:out-of-D_y-nothing-changes}
	Let $u$ be a vertex that is not a descendant of $nca_D(x,y)$ in $D'$.
	Then $dom(u)=dom'(u)$ and $D'(u) = D(u)$. 
	Moreover, if $(d(u),u)$ was a bridge in $G_s$, then it remains a 
	bridge in $G'_s$.
\end{lemma}
\begin{proof}
The claim that 
$D'(u) = D(u)$ follows immediately from the fact that all $D$-scanned vertices (which were the only vertices $v$ for which $dom(v)\not=dom'(v)$) were in $D(nca_D(x,y))$ and therefore remain in $D'(nca_D(x,y)) = D(nca_D(x,y))$.
This also implies that $dom(u)=dom'(u)$, as for each $w\in dom(u)$ it holds that $w\notin D(nca_D(x,y))$.
Now let $(d(u),u)$ be a bridge in $G_s$, and assume by
contradiction that $(d(u),u)$ is no longer a bridge in $G'_s$.
This implies there is a path from $s$ to $u$ avoiding $(d(u),u)$ in $G'_s$.
Since $(d(u),u)$ was the only incoming edge to $D(u)=D'(u)$, the path in $G'$ must contain and edge $(w,z)$ such that $w\notin D(u), z\in D(u)$.
This contradicts the fact that the only new edge is $(x,y)$ and that either $x,y\in D(u)$ or $x,y\notin D(u)$, as otherwise $D(nca_D(x,y))$ is an ancestor of $u$ in $D$ (which we assume is not).
\proofend
\end{proof}

\section{Updating the \newLNT{} forest after an edge insertion}
\label{sec:canonical-loop-update}

Let $G$ be a directed graph and let $G_s$ be the flow graph of $G$ with an arbitrary start vertex $s$.
In this section we show how to maintain the \newLNT{} forest $L$ of $G_s$ under a sequence of edge insertions.
We assume that
$D$ and $L$ are rooted at $s$ and $c_s$, respectively.
For simplicity, we also assume that all vertices of $G$ are reachable from $s$, so $m \ge n-1$. If this is not true, then we can simply recompute $D$ and $L$ from scratch, in linear time, every time a vertex becomes reachable from $s$ after an edge insertion. Since there can be at most $n-1$ such events, the total running time for these recomputations is $O(mn)$.

Throughout the sequence of edge insertions, we maintain
as additional data structures only the dominator tree $D$ (with the incremental dominators algorithm in \cite{dyndom:2012}), the bridge decomposition and the auxiliary components of $G_s$ (with the algorithm in
\cite{GIN16:ICALP}).

\begin{algorithm}[t!]
	\LinesNumbered
	\DontPrintSemicolon
	Set $s$ to be the designated start vertex of $G$.\;
	
	Compute the dominator tree $D$ and the set of bridges $\mathit{Br}$ of the corresponding flow graph $G_s$. \;
	
	Compute the bridge decomposition $\mathcal{D}$ of $D$, and the auxiliary components of $G_s$.\;
	
	Compute the loop nesting forest $H$ of the $G_s$.\;
	Construct the \newLNT{} forest $L$ from $H$ by contacting each vertex $w$ into $c_w$.\;
	
	\caption{\textsf{Initialize}$(G,s)$}
	\label{alg:initialize}
\end{algorithm}

\begin{figure}[t!]
	\begin{center}
		\includegraphics[trim={0cm 0cm 1cm 0cm}, clip=true, width=1.0\textwidth]{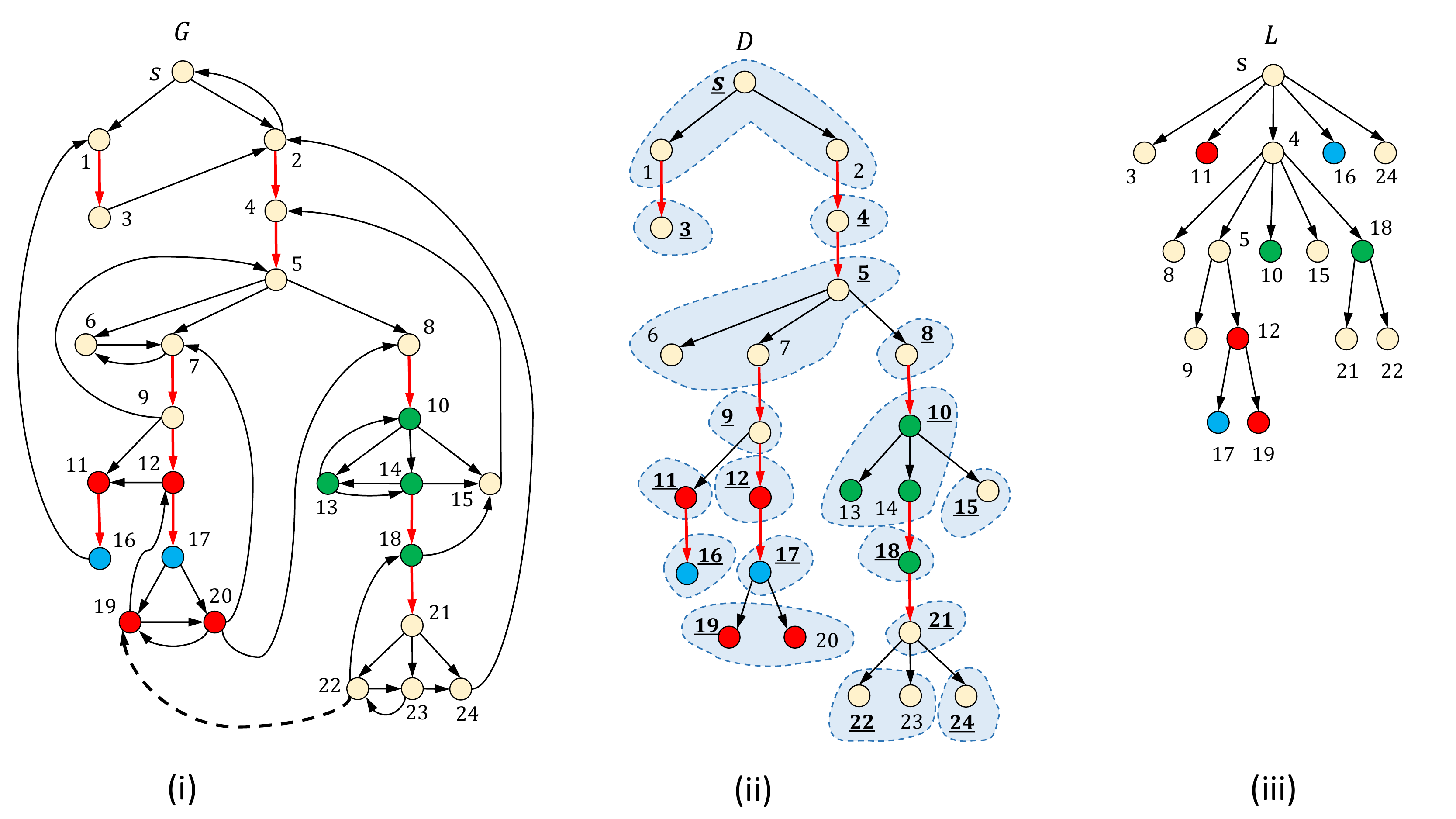}	
		\includegraphics[trim={0cm 0cm 1cm -1.2cm}, clip=true, width=1.0\textwidth]{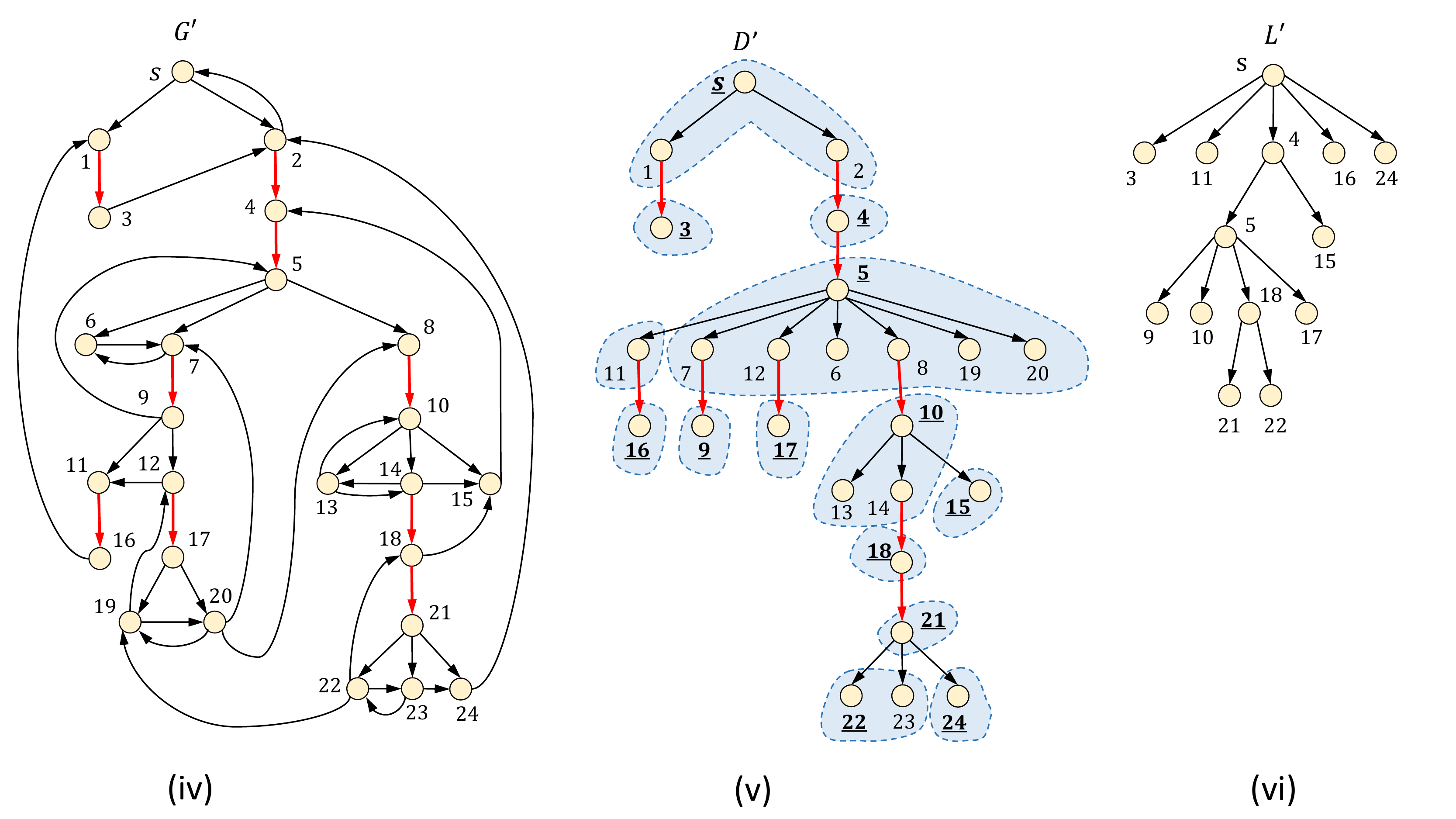}	
		\caption{A detailed example demonstrating the different types of vertices that we consider after an edge insertion.
			In (i), a digraph $G$ before the insertion of $(22,19)$ (dashed edge).
			In (ii), the dominator tree of $G$ where vertices in the same auxiliary components are grouped together.
			In (iii), the \newLNT{} tree of $G$.
			In (i)-(iii) we color the $D$-affected with red the $D$-scanned but not $D$-affected vertices with blue, and the $L$-affected but not $D$-scanned vertices with green.
			Finally, in (iv), (v), and (vi) we represent $G'$, $D'$ where vertices in the same auxiliary components are grouped together, and $L'$.
		}
		\label{figure:detailed-example}
	\end{center}
	\vspace{-.7cm}
\end{figure}

\smallskip
\noindent{\bf Initialization and 
	restarts.}
To initialize the algorithm, we compute the dominator tree $D$, bridge decomposition and auxiliary components, which can be done in linear time \cite{domin:ahlt,GIN16:ICALP}.
We also compute the \newLNT{} forest $L$ of $G_s$ in linear time, as suggested by Lemma \ref{lemma:loops-nestings-relation}. The pseudocode for the initialization is given in Algorithm \ref{alg:initialize}.

After the first initialization, in some special cases we initialize our algorithm again, in order to simplify the analysis.
We call this a \emph{restart}.
We restart our algorithm whenever a bridge $e=(u,v)$ of $G_s$ is canceled after the insertion of a new edge $(x,y)$ but we still have $d'(v)=u$, i.e., $(u,v)$ is no longer a strong bridge in $G$ but the parent of $u$ in the dominator tree $D$ does not change.
In this case, we say that the bridge $e=(u,v)$ is \emph{locally canceled}.
This is a difficult case to analyze:
the incremental dominators algorithm does not spend any time, since there are no $D$-affected vertices, while the bridge decomposition and the auxiliary components  of $G_s$ might change.
Fortunately, there are at most $O(n)$ locally canceled bridges throughout a sequence of edge insertions~\cite{GIN16:ICALP}.
Hence, we restart our algorithm at most $O(n)$ times. Consequently, the total time spent in restarts is $O(mn)$.

\smallskip
\noindent{\bf High-level overview of the update.}
Let $(x,y)$ be the new edge to be inserted.
Similarly to Section \ref{sec:incrementa-dominator-and-auxiliary-components},
for any function $f$, we use the notation $f'$ to denote the same function after the insertion of $(x,y)$, e.g., we denote by $\ell'(v)$, the parent of a \cvertex{} vertex $v$ in the \newLNT{} forest, after the insertion of $(x,y)$, and by $L'$ the resulting \newLNT{} forest.
Once again, we denote by \emph{$D$-scanned} the vertices that decrease their depth in the dominator tree $D$ after an edge insertion.
Moreover, we denote by \emph{$D$-affected} the vertices that change their parent in $D$ and by \emph{$L$-affected} the vertices for which it holds $\ell'(v')\not= c'_{\ell(v)}$, i.e.,
when the parent of $v'$
in $L'$ is not in the same auxiliary component as $\ell(v)$ (the parent of $v$ in $L$). 
After the insertion of a new edge $(x,y)$, if not involved in a restart, our algorithm performs the following updates, as shown in the
pseudocode of Algorithm \ref{alg:SCInsertEdge}:
\begin{itemize}
	\item[1)] Compute the new dominator tree $D'$, the corresponding bridge decomposition $\mathcal{D}'$, and the new auxiliary components. 
	\item[2)] Compute $\ell'(v)$ for the $D$-scanned \cvertex{} vertices $v\in D'_y$.
	\item[3)] Compute $\ell'(v)$ for the $D$-scanned \cvertex{} vertices $v\notin D'_y$.
	\item[4)] Compute $\ell'(v)$ for the $L$-affected canonical vertices $v$ that are not $D$-scanned.
\end{itemize}

As already mentioned,
the dominator tree $D$, the bridge decomposition and the auxiliary components of $G_s$ can be maintained during edge insertions within our claimed bounds 
\cite{dyndom:2012,GIN16:ICALP}. To complete the algorithm, it remains to show how to update efficiently the parent in the hyperloop nesting forest of the $D$-scanned and the $L$-affected vertices. This is non-trivial, and the low-level technical details of the method are spelled out in Sections \ref{sec:canonical-loop-update-d-scanned} and \ref{sec:canonical-loop-update-not-d-affected}, respectively.
Before giving the details of our algorithm, we show that we only need to consider the vertices $D'(r'_y)$.

\begin{algorithm}[t!]
	\LinesNumbered
	\DontPrintSemicolon
	
	Let $s$ be the designated start vertex of $G$, and let $e=(x,y)$.\;
	
	Update the dominator trees $D$. Let $S$ be the set of $D$-scanned vertices.\;
	
	Update the bridge decomposition $\mathcal{D}$, and the auxiliary components of $G_s$\;
	
	\eIf{a bridge is locally canceled in $G_s$}
	{
		Execute \textsf{Initialize}$(G,s)$.
	}
	{
		Execute 
		\textsf{Update-D-scanned}$(\mathcal{D},L,x,y,S)$ and \textsf{Update-L-affected}$(\mathcal{D},L,x,y,S)$.\;
	}
	
	\caption{\textsf{SCInsertEdge}$(G,e)$}
	\label{alg:SCInsertEdge}
\end{algorithm}

\begin{lemma}
	\label{lem:only-in-D(r_y)}
	No canonical vertex  $v\notin D'(r'_y)$ is $L$-affected.
\end{lemma}
\begin{proof}
	We have that $nca_D(x,y) \in D'(r'_y)$.
	Assume $\ell'(v') = w$.
	The fact that $v\notin D'(r'_y)$ implies that $r'_{w} \notin D'(r'_y)$.
	By Lemma~\ref{lem:out-of-D_y-nothing-changes}, $r'_w = r_w$ and $D'(r'_w) = D(r_w)$.
	Therefore, $G[D(r_w)] = G'[D'(r'_w)]\setminus (x,y)$.
	As $w$ and $v$ are not strongly connected in $G[D(r_w)] = G[D'(r'_w)]$ but they are strongly connected in $G'[D'(r'_w)] = G[D(r_w)]\setminus(x,y)$, if follows that either all paths from $w$ to $v$ or from $v$ to $w$ in $G'[D'(r'_w)]$ contain $(x,y)$.
	Thus, all paths from $w$ or from $v$ to $y$ in $G'[D'(r'_w)]$ contain $(x,y)$.
	As $w,v\notin D'(nca_D(x,y))$, by Lemma~\ref{lemma:partition-paths}, all paths from $w$ or $v$ to $y$ in $G'[D'(r'_w)]$ contain $nca_D(x,y)$.
	Note that there is path from $nca_D(x,y)$ to $y$ in $G'[D'(nca_D(x,y))]$ avoiding $(x,y)$ (as $y\in D(nca_D(x,y)) = D'(nca_D(x,y))$, and $G[D(nca_D(x,y))] = G'[D'(nca_D(x,y))]\setminus(x,y)$).
	A contradiction to the fact that all paths from $w$ or from $v$ to $y$ in $G'[D'(r'_w)]$ contain $(x,y)$.
	\proofend
\end{proof}

\subsection{Updating the $D$-scanned vertices}
\label{sec:canonical-loop-update-d-scanned}

Let $S$ be the set of $D$-scanned vertices containing also the $D$-affected vertices.
After the insertion of the edge $(x,y)$, all the $D$-affected vertices become children of $nca_D(x,y)$ in $D'$, by Lemma \ref{lemma:insert-affected}.
In this section we deal with the update of the parent in the hyperloop nesting forest $\ell'(v')$, for all \cvertex{} vertices $v\in S$.
From now on, in order to simplify the notation, we assume without loss of generality that $v=c_v$ for any vertex of interest $v$; we also denote $c'_v$ by $v'$.

After the insertion of the edge $(x,y)$ only a subset of the ancestors in $L'$ of an $L$-affected \cvertex{} vertex $v$ changes.
In particular, Lemma~\ref{lem:ancestor-remains-ancestor} shows that the ancestors $w$ of $v$ in $L$ such that  $w\notin D'(r'_y)$ remain ancestors of $v'$ in $L'$.
However, the insertion of $(x,y)$ might create a new path from $v$ to a \cvertex{} vertex $z$ such that $v\in D'(r'_z)$, containing only vertices in $D'(r'_z)$.
In such a case, $z'$ becomes an ancestor of $v'$ in $L'$.

\begin{lemma}
	\label{lem:ancestor-remains-ancestor}
	Let $v \in D(r'_y)$ be a \cvertex{} vertex in $G$. 		
	For each ancestor $z$ of $v$ in $L$, such that $z\notin D'(r'_y)$, the canonical vertex $z'$ remains ancestor of $v'$ in $L'$.
	Moreover, $D'(r'_{z'}) = D(r_z)$.
\end{lemma}
\begin{proof}
	By Lemma \ref{lem:out-of-D_y-nothing-changes} and the fact that $r'_{z'}$ is an ancestor of $D'(r'_y)$, it follows that $D'(r'_{z'}) = D(r_z)$ and $(d(r'_{z'}),r'_{z'})$ remains a bridge in $G'_s$.
	Then, the fact that
	${z'}$ is an ancestor of $v'$ in $L'$ follows from the fact that $v$ and $z$ are strongly connected in $G'[D'(r'_{z'})]$ and $D'(r'_{z'}) = D(r_z)$.
	\proofend
\end{proof}

The following lemma identifies the new parent in $L'$ of $y'$.

\begin{lemma}
	\label{lem:parent-of-y-in-h}
	It holds that $\ell'(y') = w'$, where $w$ is the nearest ancestor of $y$ in $L$ such that $w\notin D'(r'_y)$.
\end{lemma}
\begin{proof}
By Lemma \ref{lem:ancestor-remains-ancestor},
for all ancestors $z$ of $y$ in $L$, such that $z\notin D'(r'_y)$, vertex $z'$ remains ancestor of $y'$ in $L'$.
We show that all ancestors of $y'$ in $L'$ were ancestors of $y$ in $L$ before the edge insertion.
Assume by contradiction that there exists a \cvertex{} vertex $z'$ that is an ancestor of $y'$ in $L'$, but $z$ is not an ancestor of $y$ in $L$.
By the definition of \newLNT{} forest, $r'_{z}$ must be a proper ancestor of $r'_y$.
Lemma \ref{lem:out-of-D_y-nothing-changes} and the fact that $(x,y)$ is not locally canceled imply that  $dom'(r'_y) = dom(r_y)$, $r'_{z}$ is an ancestor of $r_y$ in $D$, and  $D'(r'_{z}) =
D(r_z)$.
Following our assumption we have that $z$ and $y$ are strongly connected in $G'[D'(r'_{z})]$ but not in $G[D(r_z)] = G'[D(r_z)] \setminus (x,y)$.
Therefore, there is no path from $x$ to $y$ in $G[D(r_z)]$.
By the fact that $y \in D({nca_D}(x,y))$, there exists a path from ${nca_D}(x,y)$ to $y$ in $G[D(nca_D(x,y))]$ avoiding $x$.
As $nca_D(x,y)\in D(r_z)$, there is a path from $z$ to $y$ through $nca_D(x,y)$ in $G[D(r_z)]$ avoiding $x$; a contradiction.
The lemma follows.
\proofend
\end{proof}

We now compute $\ell'(v')$ for each \cvertex{} vertex $v' \in S \cap D'_{y}$, that is, the $D$-affected vertices that are in the same tree of the canonical decomposition of $D'$ with $y$.
Note that all the new paths that are introduced by the insertion of the edge $(x,y)$ must contain $y$.
We use this observation to compute $\ell'(v')$, for $v' \in S \cap D'_{y}$, based on $\ell'(y')$ as shown in the following lemma.

\begin{lemma}
	\label{lem:new-parent-in-h-1}
	Let $v \in D'_{y}$ be $D$-scanned.
	If $v'$ and $y'$ are in different auxiliary components then $\ell'(v') = w'$, where $w$ is the nearest ancestor of $v$ in $L$ such that $w \notin D'(r'_{y})$.
\end{lemma}
\begin{proof}
By Lemma \ref{lem:ancestor-remains-ancestor}
for each ancestor $z$ of $v$ in $L$, such that $z\notin D'(r'_y)$, vertex $z'$ remains ancestor of $v'$ in $L'$.
Assume by contradiction that $\ell'(v') = z'$ with $level'(z') > level'(w')$, that is, $r'_{w}$ is a proper ancestor of $r'_{z}$ in $D'$.
Since $v\in D'_y$, it follows that $r'_{z} \notin D'(r'_y)$, and therefore $D'(r'_{z}) = D(r_z)$ by Lemma \ref{lem:out-of-D_y-nothing-changes}.
By our assumption, $v$ and $z$ are strongly connected in $G'[D(r_z)]$.
Lemma \ref{lemma:partition-paths} implies that all paths from $z$ to $v$ in $G'[D(r_{z})]$ contain $d'(r'_{v})=d'(r'_y)$.
Let $P_{zv} = P_{zd'(r'_y)} \cdot P_{d'(r'_y)v}$ be such a path, where $P_{zd'(r'_y)}$ is the subpath from $z$ to $d'(r'_y)$ and $P_{d'(r'_y)v}$ the subpath from $d'(r'_y)$ to $v$.
Since, $P_{zd'(r'_y)}$ does not contain $(x,y)$ and we also have that $v\in D(r'_y)$ it follows that $z$ has a path to $v$ in $G[D(r_z)]$.
As $v$ and $z$ are not strongly connected in $G[D(r_z)]$, all path from $v$ to $z$ in $G'[D(r_z)]$ contain $(x,y)$.
Moreover, due to the facts that $v' \not = y'$, any path from $v$ to $z$ contains a vertex out of $D'(r'_y)$, and therefore, also $r'_y$ by Lemma \ref{lemma:partition-paths}.
By the last two arguments, all paths from $v$ to $z$ contain both first $r'_y$ and then $(x,y)$.
That implies that all paths from $r'_y$ to $y$ in $G'_s$ contain the edge $(x,y)$, and as all paths from $s$ to $y$ in $G'_s$ contain $r'_y$ (by Lemma \ref{lemma:partition-paths}) it follows that all paths from $s$ to $y$ contain $(x,y)$.
This is sufficient for $(x,y)$ to be a bridge in $G'_s$, which is a contradiction as $y$ is reachable from $s$ in $G_s$.
Thus, our assumption about $z$ led us to a contradiction.
The lemma follows.
\proofend
\end{proof}

With the help of Lemma~\ref{lem:new-parent-in-h-1}, we can iterate over the vertices $v \in S \cap D'_{y}$ setting $\ell'(v')$ appropriately.
Recall that $S$ (the set of $D$-scanned vertices) is provided to us by the incremental dominators algorithm.
Next, we deal with the \cvertex{} vertices $v \in S \setminus D'_y$.
We begin with the computation of $\ell'(v')$, for the \cvertex{} vertices $v'$ in $S$ for which $level'(\ell'(v')) > level'(r'_y)$, that is, their new parent in $L$ is in $D'(r'_y) \setminus D'_{y}$.
Let $G_{scanned}$ be the graph induced by the $D$-scanned vertices, and let $H_{scanned}$ be the loop nesting forest rooted at $y$ of $G_{scanned}$.
(Note that $y$ reaches all vertices in $G_{scanned}$.)
By contracting every vertex $v$ into $c'_v$ in $H_{scanned}$, we obtain a forest $\tilde{H}$. Let $\tilde{h}(v)$ be the parent in $\tilde{H}$ of a \cvertex{} vertex $v \in S \setminus D'_y$.
As stated in Lemma~\ref{lem:new-parent-in-h-2}, the parent in $\tilde{H}$ of each \cvertex{} vertex $v'$ such that $\tilde{h}(v)\in S \setminus D'_{y}$ is the parent of $v'$ in $L'$.

\begin{figure}[t!]
	\begin{center}
		\includegraphics[trim={1cm 0cm 3cm 0cm}, clip=true, width=1.0\textwidth]{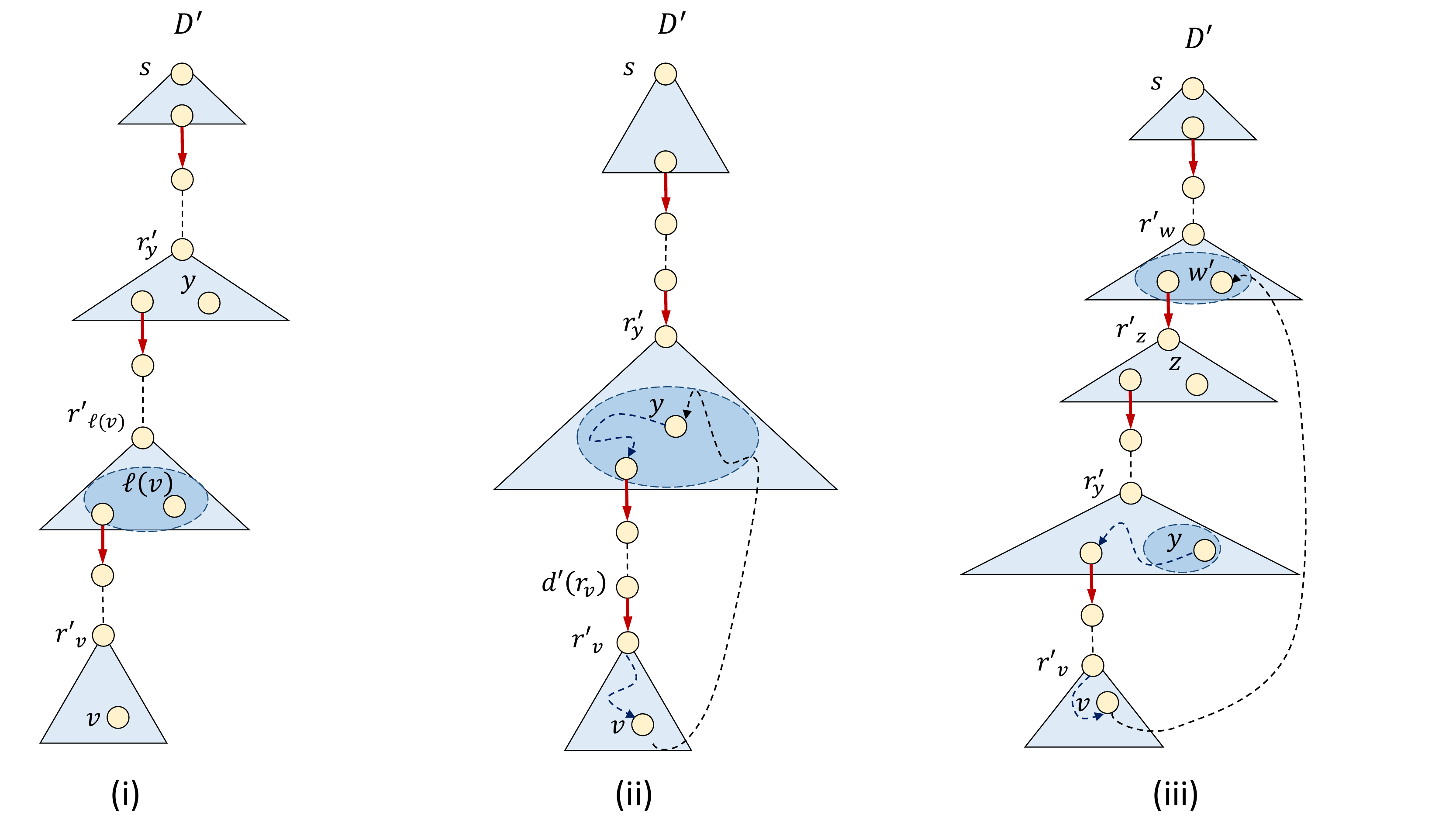}	
		\vspace{-.7cm}
		\caption{Instances of the updated dominator tree $D'$ after the insertion of $(x,y)$.
		In (i) $v$ is a $D$-scanned vertex where $v, \ell'(v)\in S\setminus D'_y$. For all such vertices, we apply Lemma \ref{lem:new-parent-in-h-2} to compute $\ell'(v)$.
		In (ii), $v$ has a path to $y$ in $G'[D'(r'_y)]$ and is $D$-affected, which is sufficient for $y'$ to be ancestor of $v'$ in $L'$.
		In (iii) vertex $w$ is an ancestor of $v$ in $L$, as $w$ and $v$ are strongly connected in $G[D(r_w)]$, and therefore $w'$ remains an ancestor of $v'$ in $L'$.
		The instances visualize the relations of the vertices of interest in the proofs of Lemmas \ref{lem:new-parent-in-h-2} and \ref{lem:new-parent-in-h-3}.}
		\label{figure:helping-proofs-1}
	\end{center}
		\vspace{-.7cm}
\end{figure}

\begin{lemma}
	\label{lem:new-parent-in-h-2}
	Let $\tilde{H}$ be the loop nesting forest rooted at $y$ of $G_{scanned}$ after contracting each vertex $v$ into $c'_v$.
	For every \cvertex{} vertex $v$ such that $\tilde{h}(v) \in S \setminus D'_{y}$, it holds that $\ell'(v') = \tilde{h}(v)$.
\end{lemma}
\begin{proof}
	Let $v\in S \setminus D'_{y}$.
	As all paths from $y$ to $v$ contain $r'_v$ and $v$ is $D$-scanned, all vertices in $D'(r'_v)$ are $D$-scanned.
	Hence, $G_{scanned}[D'(r'_v)] = G'[D'(r'_v)]$.
	Therefore, if $\ell'(v') \in S \setminus D'_{y}$ then $v$ and $\ell'(v')$ are strongly connected in $G'[D'(r'_{\ell'(v)})] = G_{scanned}[D'(r'_{\ell'(v)})]$, which implies that $\ell'(v')$ is an ancestor of $v'$ in $\tilde{H}$.
	(For a visualization of the relations of $v,y,\ell'(v)$, and $s$ see Figure \ref{figure:helping-proofs-1} (i).)
	To complete the proof we show that if a canonical vertex $z$ is the parent of $v'$ in $\tilde{H}$, where $v'\in S \setminus D'_{y}$, then $z$ is the parent of $v'$ in $L'$.
	Assume by contradiction that $\ell'(v') \not= \tilde{h}(v')$.
	Then, there is a vertex $w$ such that $level'(\ell'(v')) <level'(w) < level'(v)$, and $w$ and $v$ are strongly connected in $G'[D'(r'_w)] =G_{scanned}[D'(r'_{w})]$.
	Note that $w$ contradicts the definition of $\ell'(v')$, and thus the lemma follows.
	\proofend
\end{proof}

Finally, for the vertices $v \in S \setminus D'_{y}$ for which
$\ell'(v') \notin S \setminus D'_y$, we compute $\ell'(v')$ according to the following lemma.

\begin{lemma}\label{lem:new-parent-in-h-3}
	Let $v \notin D'_{y}$ be a $D$-scanned vertex such that $level'(\ell'(v')) \leq level'(r'_y)$.
	If $v$ has a path to $y$ in $G'[D'(r'_{y})]$, then $\ell'(v')=y'$ in $L'$.
	Otherwise,
	$\ell'(v') = w'$, where $w$ is the nearest ancestor of $v$ in $L$ such that $level'(w) \leq level'(r'_y)$.
	If there is no such vertex $\ell'(v') = \emptyset$
\end{lemma}
\begin{proof}
	First, we assume
	$v$ has a path to $y$ in $G'[D'(r'_{y})]$.
	Since $v$ is $D$-scanned, it means that $y$ has a path to $v$ in $ G'[D'(nca_{D'}(x,y))]$, and therefore, in $G'[D'(r'_{y})]$.
	Hence, $y$ and $ v $ are strongly connected in $G'[D'(r'_y)]$.
	By the definition of the \newLNT{} forest, $y'$ is an ancestor of $v'$ in $L'$, and since $level'(\ell'(v')) \leq level'(r'_y)$ it follows that $\ell'(v')=y'$.
	(For a visualization of the relations of $v,y,s$ in $D$ see Figure \ref{figure:helping-proofs-1} (ii).)	
	Now we assume $v$ does not contain a path to $y$ in $G'[D'(r'_{y})]$.
	First assume that $w\in D'_y$.
	In this case $v$ and $w$ were strongly connected in $G[D(r_{w})]$ before the insertion of $(x,y)$.
	Since $D(r_{w})\subseteq D'(r'_y)$, if follows that $v$ and $w$ are strongly connected in  $G'[D'(r'_{y})]$.
	By the definition of $L'$, $w'$ is an ancestor of $v'$ in $L'$, and by the assumption of the lemma that $level'(\ell'(v')) \leq level'(r'_y)$, and Lemma \ref{lemma:unique-level}, it follows that $\ell'(v') = w'$.
	
	Finally, for the rest of the proof we assume $w \notin D'(r'_y)$.
	By Lemma \ref{lem:ancestor-remains-ancestor}, $w'$ remains an ancestor of $v'$ in $L'$.
	(For a visualization of the relations of $v,y,w',z,s$ in $D$ see Figure \ref{figure:helping-proofs-1} (iii).)
	Assume, by contradiction, that $\ell'(v') = z'$ and $level'(w') < level'(z') < level'(r'_y)$ and $c_z$ is not an ancestor of $v$ in $L$.
	By Lemma \ref{lemma:partition-paths} all paths from $z$ to $v$ contain $r'_y$.
	As $v \in D(r'_y)$ there is a path from $r'_y$ to $v$ in $G[D(r'_{y})]$.
	Since  $D(r'_{y}) \subset D'(r'_{z})$, there is a path from $z$ to $v$ in $G'[D'(r'_{z})] \setminus (x,y)$ ($=G[D(r_{z})]$ by Lemma~\ref{lem:ancestor-remains-ancestor}).
	This fact, combined with our assumption that $c_z$ is not an ancestor of $v$ in $L$, implies that there is no path from $v$ to $z$ in $G[D(r_{z})]$.
	Since $v$ and $z$ are strongly connected in  $G'[D'(r'_{z})]$ but not in $G[D(r_{z})] = G'[D'(r'_{z})]\setminus (x,y)$ all paths from $v$ to $z$ contain $(x,y)$; let $P_{vz}$ be such a path from $v$ to $z$ and let $P_{vx}$ and $P_{yz}$ be its subpaths from $v$ to $x$ and from $y$ to $z$, respectively.
	Note that $P_{vx}$ and $P_{yz}$ exist also in $G[D(r_{z})] = G'[D'(r'_{z})]\setminus (x,y)$.
	Since we deal with the case where $v$ does not have a path to $y$ in $G'[D'(r'_{y})]$, the path $P_{vx}$ should contain a vertex $u\in D'(r'_z) \setminus D'(r'_{y})$ (as otherwise all vertices are in $D'(r'_{y})$ and the path $P_{vx}\cdot (x,y)$ is a path from $v$ to $y$ in $G'[D'(r'_{y})]$).
	By Lemma \ref{lemma:partition-paths}, $P_{vx}$ contains $r'_y$; let $P_{vr'_y}$ be the subpath of $P_{vx}$ from $v$ to $r'_y$.
	Since $y\in D(r'_y)$ there is a path $P_{r'_yy}$ from $r'_y$ to $y$ in $G[D(r'_{y})] \subset G[D(r'_z)]$ avoiding $(x,y)$.
	Observe that the paths $P_{vr'_y}$,$P_{r'_yy}$, and $P_{yz}$ all exist in $G[D(r'_z)]$.
	Collectively, we have that $P_{vr'_y}\cdot P_{r'_yy} \cdot P_{yz}$ is a path from $v$ to $z$ in $G[D(r'_z)] = G[D(r_z)]$.
	As we argued $z$ also has a path to $v$ in $G[D(r_z)]$, and therefore $z$ and $v$ are strongly connected in $G[D(r_z)]$.
	Since $r_z$ is an ancestor of $r_v$ in $D$, by the definition of $L$, $c_z$ is an ancestor of $v'$ in $L$.
	This is a contradiction to our assumption that $w$ is the nearest ancestor of $v$ in $L$ such that $w \notin D'(r'_y)$ and $level'(w) < level'(z) < level'(r'_y)$.
	The lemma follows.
	\proofend
\end{proof}

\begin{algorithm}[t!]
	\LinesNumbered
	\DontPrintSemicolon

	Set $\ell'(c'_y)=c'_w$, where $w$ is the nearest ancestor of $c_y$ in $L$ such that $w\notin D'(r'_y)$.
	
	Let $S$ be the set of $D$-scanned vertices, and $G_{scanned} = G'[V(S)]$.
	Compute the loop nesting forest $H_{scanned}$ of $G_{scanned}$ with start vertex $y$.
	Contract every vertex $v \in V(G_{scanned})$ into $c'_v$ in $H_{scanned}$, forming $\tilde{H}$. \;
	
	\ForEach{\cvertex{} vertex $v \in V(G_{scanned})$}
	{
		Let $v'=c'_v$ and $y'=c'_y$.\;
		\eIf{$v\in D'_y$}
		{
			\lIf{$v' \not= y'$}{$\ell'(v') = c'_w$, where $w$ is the nearest ancestor of $v$ in $L$ such that $w\notin D'(r'_y)$.}
		}
		{
			\lIf{$\tilde{h}(v') \in S \setminus D'_y$}{$\ell'(v')= \tilde{h}(v')$}
			\lElseIf{$v$ has a path to $y$ in $G'[D'(r'_y)]$}{$\ell'(v') = \ell'(y')$.\label{line:if-path-exists}}
			\lElse{$\ell'(v') = c'_w$, where $w$ is the nearest ancestor of $v$ in $L$ such that $level'(w)\leq level'(r'_y)$.}
		}
	}
	\caption{\textsf{Update-D-scanned}$(\mathcal{D},L,x,y,S)$}
	\label{alg:UpdateDAffected}
\end{algorithm}

Lemmas \ref{lem:parent-of-y-in-h}, \ref{lem:new-parent-in-h-1}, \ref{lem:new-parent-in-h-2}, and \ref{lem:new-parent-in-h-3} provide the tools to
update the parent in $L'$ of the $D$-scanned vertices $S$.
The pseudocode of this update is detailed
in Algorithm \ref{alg:UpdateDAffected}.
The challenging part is to determine whether $v$ has a path to $y$ in $G'[D'(r'_y)]$ (Line \ref{line:if-path-exists} of Algorithm \ref{alg:UpdateDAffected}).
We show how we can answer the queries of Line \ref{line:if-path-exists} of Algorithm \ref{alg:UpdateDAffected}, and thus how to update the parents in $L'$ of all the \cvertex{} $D$-scanned vertices in time linear in the 
number of $D$-scanned vertices and their adjacent edges.

\begin{lemma}
	For all $D$-scanned vertices $v$, we can compute $\ell'(v')$ in time $O(|V(S)|+|E(S)|)$.
	\label{lemma:scanned-h-computation}
\end{lemma}
\begin{proof}
	The loop nesting forest of the graph induced by the $D$-scanned vertices can be computed in linear time to the size of $S$, i.e., $O(|V(S)|+|E(S)|)$.
	The nearest ancestor $w$ of each $D$-scanned \cvertex{} vertex $v$ in $L$, such that $level'(w')\leq level'(r'_{y})$, can be computed in total $O(|V(S)|)$ as follows.
	We find all \cvertex{} vertices $v \in S$ such that $level'(\ell(v))  \leq level'(r'_{y})$ and we assign to each descendant of $v$ in $L$ the vertex $\ell(v)$.
	All the necessary tests of Lemmas \ref{lem:new-parent-in-h-1}, \ref{lem:new-parent-in-h-2}, and \ref{lem:new-parent-in-h-3} can be performed in constant time per vertex.
	
	Now we show that we can determine in time $O(|V(S)|+|E(S)|)$ all vertices $v \in S \setminus D'_{y}$ that have a path to $y$ in $G'[D'(r'_{y})]$, as required by Lemma \ref{lem:new-parent-in-h-3}.
	We mark all vertices in $S$ that have outgoing edges to vertices $v$ such that $c'_v = c'_y$.
	Then, the vertices in $S$ that reach $y$ in $G'[D'(r'_{y})]$ are the vertices that can reach a marked vertex in $G'[S]$.
	These vertices can be determined in time $O(|E(S)|)$, by executing backward traversals from the marked vertices without visiting the same vertices twice.
	Now we show that the above procedure is correct by showing that a vertex $v\in S \setminus D'_{y}$ has a path to $y$ in $G'[D'(r'_{y})]$ if and only if it has a path to a marked vertex in $G'[S]$.
	We start with the forward direction. If a vertex $v \in S \setminus D'_{y}$ has a path to $y$ in $G'[D'(r'_{y})]$ using only vertices in $S$, then clearly $v$ reaches $y$ in $G'[S]$.
	If on the other hand, there is a path from $v$ that uses vertices outside $S$, then let $w$ be the first vertex on that path such that $w\notin S$, and let $z$ be its predecessor on the path.
	Then $c'_w=c'_y$ since: (i) by the fact that $w\notin S$, it holds that either $r'_w \notin D'(r'_y)$ or $r'_w=r'_y$, where only the second case is interesting since in the first it holds that $w \notin G'[D'(r'_{y})]$, (ii) $w$ reaches $y$ in $G'[D'(r'_{y})]$ (by our assumption that a path from $v$ to $y$ exists in $G'[D'(r'_{y})]$ using $w$) and (iii) $y$ reaches $v$ since $v$ is $D$-scanned (which means $y$ has a path to $v$ in $G'[D'(r'_{y})]$).
	Therefore, $v$ has a path to the marked vertex $z\in S$.
	We continue with the reverse direction of the claim, that is, if $v$ does not have a path to $y$ in $G'[D'(r'_{y})]$ then there is not path to a marked vertex in $G'[S]$.
	This is true for the paths containing only vertices in $S$.
	Also there is no path from $v$ to a marked vertex $z$ in $G'[S]$, since otherwise, there is a path from $v$ to a vertex $w$ such that $c'_w=c'_y$ and therefore to $y$ in $G'[D'(r'_{y})]$ using the edge $(z,w)$, while we assumed that no such path exists.
	Thus, we showed how to compute all the vertices that have a path to $y$ in $G'[D'(r'_{y})]$ in time $O(|E(S)|)$, which concludes the lemma.
	\proofend
\end{proof}

\subsection{Updating the $L$-affected vertices that are not $D$-scanned}
\label{sec:canonical-loop-update-not-d-affected}

Now we consider updating the parent in the \newLNT{} forest of the vertices that are $L$-affected but not $D$-affected.
We start with the following
lemma that is used throughout the section. The lemma suggests that we can find all the $L$-affected vertices via a backward traversal from $y$ vising all vertices in $G'[D'(r'_y)]$ that have a path to $y$.
In the worst case we spend $O(m)$ time to execute the traversal, which our algorithm cannot afford.
Later, we show how to speed up this process by exploiting some key properties of the $L$-affected vertices.

\begin{lemma}
	\label{lem:path-through-x-y}
	For every $L$-affected vertex $v$ that is not $D$-scanned, every path from $v$ to $\ell'(v')$ in $G'[D'(r'_{\ell'(v')})]$ contains $(x,y)$.
	Moreover, $v$ has a path to $x$ in $G[D'(r'_{y})] = G'[D'(r'_{y})]\setminus (x,y)$.
\end{lemma}
\begin{proof}
	
	Assume by contradiction that there exists a $L$-affected vertex $v \notin S$ such that $v$ has a path $P$ to $\ell'(v')$ in $G'[D'(r'_{\ell'(v')})]\setminus(x,y)$.
	Since $v$ is not $D$-scanned the ancestors of $v$ in $D$ and $D'$ are the same and the edge $(d(r'_{\ell(v)}),r'_{\ell(v)})$ is a strong bridge in $G'_s$ such that $v\in D'(r'_{\ell(v)})$.
	Therefore, $c'_{\ell(v)}$ is an ancestor of $v'$ in $L'$.
	Notice that $level'(\ell'(v')) > level'(\ell(v))$ as $v$ is $L$-affected and $c'_{\ell(v)}$ remains an ancestor of $v'$ in $L'$ by Lemma~\ref{lem:ancestor-remains-ancestor}.
	If $\ell'(v')\notin D'(r'_y)$, by Lemma~\ref{lem:ancestor-remains-ancestor}, $\ell'(v')$ is an ancestor of $v$ in $L$, contradicting the fact that $level'(\ell'(v')) > level'(\ell(v))$.
	Now we assume $\ell'(v') \in D'(r'_y)$.
	All $D$-scanned vertices become children of $nca_D(x,y)$, and by Lemma~\ref{lem:out-of-D_y-nothing-changes} no vertex outside of $nca_{D'}(x,y)$ is $D$-scanned.
	Therefore, vertices can only move out of $D'(r'_{\ell'(v')})$ and hence $D'(r'_{\ell'(v')}) \subset D(r'_{\ell'(v')})$.
	Collectively we have that $\ell'(v')$ and $v$ are strongly connected in $G'[D'(r'_{\ell'(v')})]\setminus (x,y) = G[D(r'_{\ell'(v')})]$, and moreover $r'_{\ell'(v')}$ is a proper ancestor of $v$ in $D'$ (as it is in $D$).
	As $r'_{\ell'(v')}\in D'(r'_{\ell(v)})$, we have $r'_{\ell'(v')}\in D(r_{\ell(v)})$ as non of the ancestor of $v$ in $D$ are $D$-scanned.
	That means $r_\ell'(v')$ is a ancestor of $v$ in $L$ and a descendant of $r_{\ell(v)}$, which contradicts the fact that $level'(\ell'(v')) > level'(\ell(v))$.
	Thus, all paths from a $L$-affected vertex $v \notin S$  to $\ell'(v')$ contain $(x,y)$.

	Now we prove the second part of the lemma.
	Assume by contradiction that $v$ does not have a path to $x$ in $G[D'(r'_{y})]$.
	Then, all paths contain a vertex $w \notin D'(r'_{y})$, and by Lemma \ref{lemma:partition-paths} also vertices $(d'(r'_y),r'_y)$.
	Therefore, all paths in $G'[D'(r'_{\ell'(v')})]$ from $v$ to $x$ contain $r'_y$ (clearly avoiding $(x,y)$); let $P_{vr'_y}$ be the subpath from $v$ to $r'_y$ of any such path.
	By the fact that $y\in D'_{y}$ it follows that there is a path from $r'_y$ to $y$ avoiding $(x,y)$; let $P_{r'_yy}$ be such a path.
	Then, $P_{vr'_y} \cdot P_{r'_yy}$ is a path in $G'[D'(r'_{\ell'(v')})]$ from $v$ to $y$ avoiding $(x,y)$, which contradicts
	the fact that all paths from $v$ to $\ell'(v')$ contain $(x,y)$.
	\proofend
\end{proof}
		
		Notice that we only need to consider the vertices that are in $D'(r'_y)$ and are not $D$-scanned, by Lemma \ref{lem:only-in-D(r_y)}.
		
		\begin{figure}[t!]
			\begin{center}
				\includegraphics[trim={.5cm 1cm 1cm 0cm}, clip=true, width=.9\textwidth]{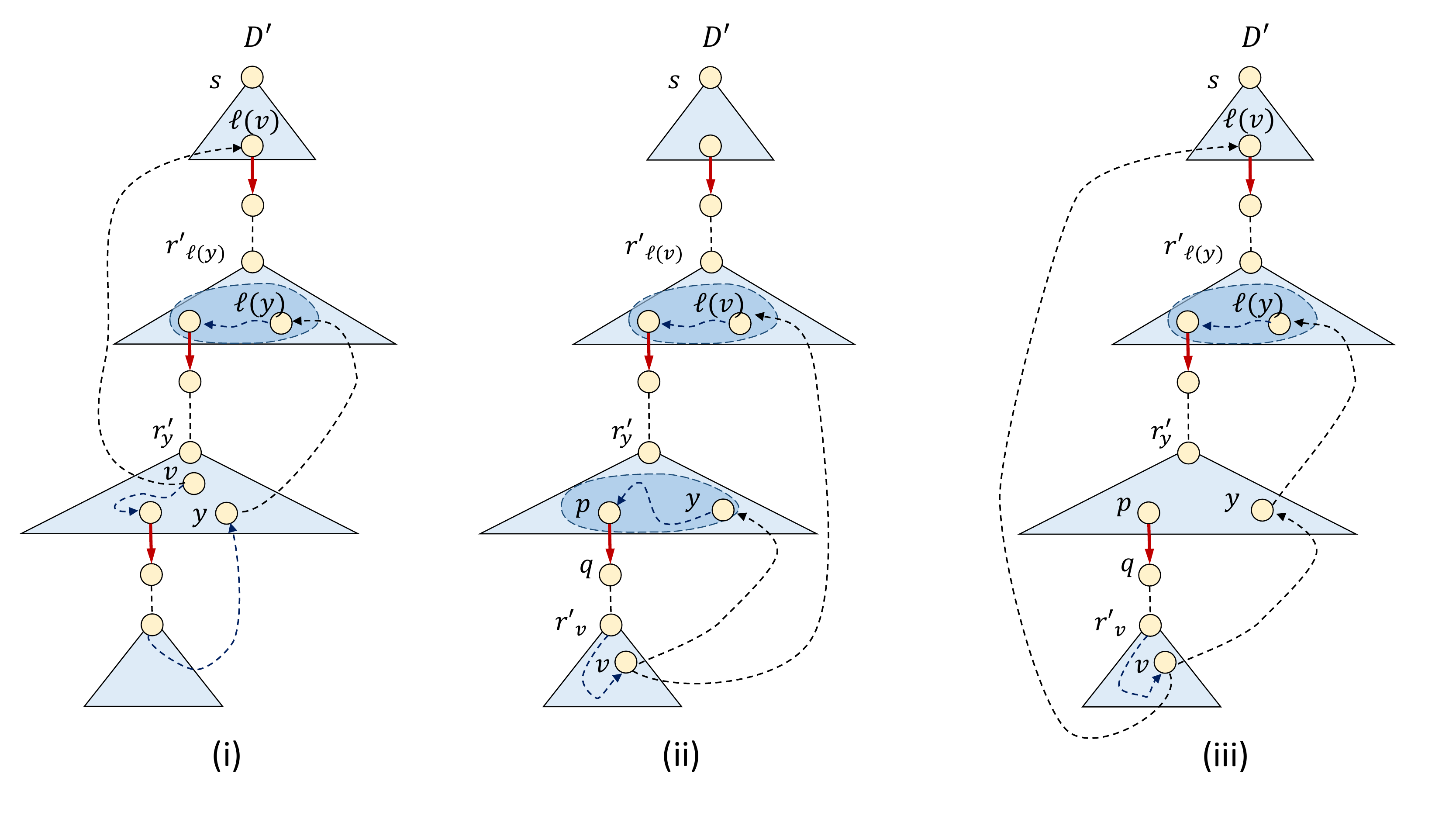}	
				\vspace{.3cm}
				\caption{A demonstration of the different case in Lemma~\ref{lemma:new-parent-in-H}.
					(i) Case (1) of Lemma~\ref{lemma:new-parent-in-H} where $v \in D'_y$ and  $level'(\ell(v)) < level'(\ell'(y'))$.
					Here we have $\ell'(v') = \ell'(y')$.
					(ii) Case (2.1) of Lemma~\ref{lemma:new-parent-in-H} where $v \notin D'_y, c'_p = y'$ and $level'(\ell(v)) < level'(c'_p)$.
					Now we have $\ell'(v') = c'_p$.
					(iii) Case (2.2) of Lemma~\ref{lemma:new-parent-in-H} where $v \notin D'_y,c'_p \not = y'$ and $level'(\ell(v)) < level'(\ell'(y'))$.
					In this case $\ell'(v') = \ell'(y')$.}
				\label{figure:helping-proofs-2}
			\end{center}
			\vspace{-.7cm}
		\end{figure}
		
		\begin{lemma}
			\label{lemma:new-parent-in-H}
			Let $(x,y)$ be the newly inserted edge.
			The \cvertex{} vertex $v'$ of a vertex $v \in D'(r'_y)$  that is not $D$-scanned and has a path to $x$ in $G'[D'(r'_{y})]$, changes its parent $\ell'(v')$ as follows:
			\vspace{-.25cm}
			\begin{mylist}{ (1)}
				\litem{(1)} Case $v \in D'_y$: if $level'(\ell(v)) < level'(\ell'(y'))$ or $\ell(v) = \emptyset$ then $\ell'(v') = \ell'(y')$. Otherwise, $\ell'(v') = c_{\ell(v)}$.
				\litem{(2)} Case $v \notin D'_y$: let $(p,q)$ the strong bridge such that $p \in D'_y$ and $q$ is an ancestor of $v$ in $D'$.
				\begin{mylist}{ (2.1)}
					\litem{(2.1)} Case $c'_p = y'$: if $level'(\ell(v)) < level'(c'_p)$ or $\ell(v) = \emptyset$, then $\ell'(v') = c'_p$. Otherwise, $\ell'(v') = c'_{\ell(v)}$.
					\litem{(2.2)} Case $c'_p \not = y'$: if $level'(\ell(v)) < level'(\ell'(y'))$  or $\ell(v) = \emptyset$, then $\ell'(v') = \ell'(y')$. Otherwise, $\ell'(v') = c'_{\ell(v)}$.
				\end{mylist}	
			\end{mylist}
			\vspace{-.2cm}
		\end{lemma}
		\begin{proof}
			
			First, consider the case where $v = D'_y$.
			Recall that we consider the update of $\ell'(v')$ for all vertices $v$ that are $L$-affected but not $D$-scanned when our algorithm is not involved in a restart.
			Since no bridge is locally canceled, and no ancestor of $v$ is $D$-affected, then for every bridge $(d(q),q)$ of $G_s$ for which $v\in D(q)$ we have that $v\in D'(q)$ and $(d(q),q) = (d'(q),q)$ is still a bridge in $G'_s$.
			Moreover, no new bridge can appear on the paths from $s$ to $v$.
			Hence $(d(r_v),r_v) = (d(r'_v),r'_v) = (d(r'_y),r'_y)$ is a bridge also in $G'_s$ and $r'_y = r_v$.
			Since $\ell(v)\notin D(r_v)$, it follows that $\ell(v) \notin D'(r'_y)$.
			Therefore, $z=\ell(v)$ is the \cvertex{} vertex with the largest level for which $v$ and $z$ are strongly connected in $G[D(r_{z})]=G[D'(r'_{z})]$ (the equality holds by Lemma \ref{lem:out-of-D_y-nothing-changes}); that is, without using the edge $(x,y)$.
			In the case where $\ell(v) = \emptyset$ there is no such vertex.
			If $v$ has a path to $x$ in $G'[D'(r'_y)]$, then $\ell'(y')$ and $v$ are strongly connected in $G'[D'(r'_{\ell'(y')})]$, as (i) $v$ has a path to $\ell'(y')$ in $G'[D'(r'_{\ell'(y')})]$ through $y$, and (ii) $\ell'(y')$ has a path to $v$ in $G'[D'(r'_{\ell'(y')})]$ through $r'_y$ (i.e., the very same subpath $P_1$ from $\ell'(y')$ to $r'_y$ as in the path from $\ell'(y')$ to $y$, followed by any path from $r'_y$ to $v$).
			Therefore, if $v$ has a path to $x$ in $G'[D'(r'_y)]$ then $\ell'(y')$ is an ancestor of $v'$ in $L'$.
			That includes the case where $\ell(v) = \emptyset$.
			Now we show that if $v$ is $L$-affected then $\ell'(v')$ is an ancestor of $y'$ in $L'$.
			By Lemma \ref{lem:path-through-x-y}, it follows that all paths in $G'[D'(r'_{\ell'(v')})]$ from $v'$ to $\ell'(y')$ contain $(x,y)$, and therefore $y$ has a path to $\ell'(y')$ in $G'[D'(r'_{\ell'(v')})]$ as well.
			Moreover, $\ell'(v')$ has a path to $y$ in $G'[D'(r'_{\ell'(v')})]$ through $r'_y$ (i.e., the very same subpath from $\ell'(v')$ to $r'_v$ as in the path from $\ell'(v')$ to $v$, followed by any path from $r'_y$ to $y$).
			Collectively, we showed that if $v$ has a path to $x$ in $G'[D'(r'_{\ell'(y')})]$, then $\ell'(y')$ is an ancestor of $v'$ in $L'$, and moreover if $v$ is $L$-affected then $\ell'(v')$ is an ancestor of $y'$ in $L'$.
			Thus, $\ell'(v') = \ell'(y')$ if $level'(\ell(v)) < level'(r'_{\ell'(y')})$ or $\ell(v) = \emptyset$ and $v$ reaches $x$ in $G'[D'(r'_{y})]$, and $\ell'(v') = c'_{\ell(v)}$ otherwise.
			
			Now we prove the case where $v \notin D'_y$.
			First, we show that if $level'(\ell(v)) \geq level'(y')$, then $\ell'(v')=\ell(v)$.
			There is no $D$-scanned vertex $t \in D'(r'_{\ell(v)})$ since otherwise it is an descendant of a $D$-affected vertex in $D'_y$ and thus, $v$ is also a $D$-scanned descendant of the same $D$-affected vertex.
			As all $D$-affected vertices become children of $nca_D(x,y)$ in $D$, for every vertex $z \in D'(r'_y)$ it holds that $D'(r'_z) \subseteq D(r_z)$.
			As we assume $level'(\ell(v)) \geq level'(y')$, all new ancestors of $v'$ in $L'$ have level greater than $level'(y)$.
			Moreover, in the case where $level'(\ell(v)) = level'(y')$ it holds that $v$ and $\ell(v)$ are strongly connected in $G'[D'(r'_y)]$ since $D'(r'_y) \supseteq D(r'_{\ell(v)})$.
			Hence, $c'_{\ell(v)}$ is an ancestor of $v'$ in $L'$.
			As mentioned before, $D'(r'_z) \subseteq D(r_z)$ for all ancestors $r'_z$ of $v$ in $D'$ such that $level'(r_z) > level'(y)$ (including $q$ from the statement of the lemma).
			Therefore,	if $\ell'(v') \not= c'_{\ell(v)}$ then $\ell'(v')$ and $v$ are strongly connected in $G[D(r_{\ell'(v')})]$, and hence, $\ell'(v')$ is an ancestor of $v$ in $L$ with $level'(\ell'(v')) > level'(\ell(v))$.
			This contradicts the definition of $\ell(v)$.
			Thus, also in this case it follows $\ell'(v)=\ell(v)$.
			
			Finally, we consider the case where $v\in D'(r'_v)\setminus D'_y$ and  $level'(\ell(v)) < level'(y')$ or $\ell(v) = \emptyset$.
			Let $(p,q)$ be the strong bridge such that $p \in D'_y$ and $q$ is an ancestor of $v$ in $D'$.
			As mentioned before, $D'(r'_z) \subseteq D(r_z)$ for all ancestors $r'_z$ of $v$ in $D'$ (including $q$) such that $level'(r'_z) > level'(y)$.
			Therefore, if $\ell'(v') \not= c_{\ell(v)}$ then $\ell'(v')$ cannot be a descendant of $q$ in $D'$ since $\ell'(v')$ and $v$ is strongly connected in $G[D(r'_{\ell'(v')})]$, and therefore, $\ell'(v')$ is an ancestor of $v$ in $L$ (which  contradicts the definition of $\ell(v)$).
			Notice that if $v$ did not have a path to $x$ before the insertion in $G[D'(r'_y)]$, by Lemma \ref{lem:path-through-x-y} $v$ is not $L$-affected.
			Next we assume that $v$ has a path to $x$ and $y$ in $G'[D'(r'_y)]$.	
			If $y$ and $p$ are strongly connected in $G'[D'(r'_y)]$ (i.e., $y' = c'_{p}$), then also $v$ and $c'_p$ are strongly connected in $G'[D'(r'_{y})]$, since $v$ has a path to $y$ and $p$ has a path to $v$ in $G'[D'(r'_{y})]$.
			If additionally $r'_{\ell'(v')}$ is an ancestor of $r'_{p}$ in $D'$ (that is, $level'(\ell(v)) < level'(\ell'(v')) \leq level'(p)$), by the definition of the \newLNT{} forest $\ell'(v') = c'_p$, including the case where $\ell(v) = \emptyset$.
			To prove the case where $y' \not = c'_p$ we can use the same argument as in the case where $v \in D'_y$.
			\proofend
		\end{proof}

		\begin{algorithm}[t!]
			\LinesNumbered
			\DontPrintSemicolon
			\ForEach{\cvertex{} vertex $v \in D'(r'_y), v\notin S,$ that has a path to $x$ in $G'[D'(r'_{y})]$  \label{line:has-path-to-x}}{
				\eIf{$v\in D'_y$}
				{
					\lIf{$level'(\ell(v)) < level'(\ell'(c'_y))$ or $\ell(v) = \emptyset$}{$\ell'(c'_v) = \ell'(c'_y)$}
				}
				{
					Let $(p,q)$ be the bridge such that $p \in D'_y$ and $q$ is an ancestor of $v$ in $D'$.
					
					\lIf{$c'_p = c'_y$ and $level'(\ell(v)) < level'(c'_p)$ or $\ell(v) = \emptyset$}{$\ell'(c'_v)= c'_p$}
					
					\lElseIf{$c'_p \not= c'_y$ and $level'(\ell(v)) < level'(\ell'(c'_y))$ or $\ell(v) = \emptyset$}{$\ell'(c'_v) = \ell'(c'_y)$}	
				}
			}
			\caption{\textsf{Update-L-affected}$(\mathcal{D},L,x,y,S)$}
			\label{alg:UpdateLAffected}
		\end{algorithm}
		
		Lemma \ref{lemma:new-parent-in-H} shows
		how to
		determine the new parent in $L'$ of each \cvertex{} vertex $v \in D'(r'_y)$ that is $L$-affected but not $D$-scanned.
		The pseudocode for this update is given
		in Algorithm \ref{alg:UpdateLAffected}.
		In the following we show how we can efficiently answer all the tests of Algorithm \ref{alg:UpdateLAffected}.
		The most challenging computation is to determine which vertices have a path to $x$ in $G'[D'(r'_{y})]$ as required in Line~\ref{line:has-path-to-x}.
		We show how to compute efficiently those vertices by executing a backward traversal: this runs in time proportional to the sum of the degrees of the $L$-affected vertices.
		We start with the following definition of \emph{loop cover} of a vertex, which
		we use to speed up our backward search.

		\begin{definition}
		Let  $w \in D'(r'_y)\setminus S$ be a \cvertex{} vertex, and let $\ell_{min}$ be the ancestor of $w$ in $L$ with the lowest level such that $\ell_{min}\in D'(r'_y)$.
		Moreover, let $(p,q)$ be the bridge such that $p\in D'_{\ell_{min}}$ and $q$ is an ancestor of $w$.
		We call $q$ the \emph{loop cover} $lcover(w)=q$ of $w$ in $D$.
		If $\ell(w) \not \in D'(r'_{y})$, then $lcover(w) = \emptyset$.
		\end{definition}

		We use the loop cover of vertices that are neither $D$-scanned nor $L$-affected in order to avoid unnecessary visits to vertices during the search for $L$-affected vertices.
		Whenever we visit a vertex $w \in D'(r'_{y})\setminus D'_y$ that is not $L$-affected, then we do not need to visit any of the vertices
in $D'(lcover(w))$. 
		Formally, we have the following lemma. 
		
		\begin{lemma}
			\label{lemma:successive-candidate}
			Let $w \in D'(r'_{y})\setminus D'_y$ be a \cvertex{} vertex that is not $D$-scanned and has a path to $x$ in $G'[D'(r'_{y})]$ and $\ell'(w') = c'_{\ell(w)}$.
			If $lcover(w)\not= \emptyset$, for every \cvertex{} vertex $v \in D'(lcover(w))$ such that $v$ has a path to $w$ in $G'[D'(lcover(w))]$, we have that $\ell'(v') = c'_{\ell(v)}$.
			If $lcover(w) = \emptyset$, for all vertices $v \in  D'(r'_y)$ that have a path to $w$ in $G'[D'(r'_y)]$, it holds that $\ell'(v') = c'_{\ell(v)}$.
		\end{lemma}
		\begin{proof}
			First, assume $lcover(w)\not= \emptyset$.
			As (i) $v$ has a path to $w$ in $G'[D'(lcover(w))]$, (ii) all paths from $\ell(lcover(w))$ to $w$ in $G'[D'(lcover(w))]$ contain $lcover(w)$, and (iii) $v\in D(lcover(w))$, it follows that $lcover(w)$ and $v$ are strongly connected in $G'[D'(lcover(w))]$.
			Hence, by definition of \newLNT{} forest, $c'_{\ell(lcover(w))}$ is an ancestor of $v'$ in $L'$ and $level'(\ell(lcover(w))) \geq level'(y)$.
			Note that Lemma~\ref{lemma:new-parent-in-H} implies that the insertion of $(x,y)$ might only introduce an ancestor $t$ of $v'$ in $L'$ such that $level'(t) \leq level'(y)$.
			Therefore, by the fact that $c'_{\ell(lcover(w))}$ is an ancestor of $v'$ in $L'$ and $level'(\ell(lcover(w))) \geq level'(y)$, it follows that $\ell'(v') = c'_{\ell(v)}$.
			
			No assume $lcover(w)=\emptyset$ and by contradiction that $v$ is $L$-affected.
			We have that $\ell'(w')=c'_{\ell(w)} \notin D'(r'_y)$.
			Since $D'(r'_y)\subset D'(r'_{\ell'(w')})$, $v$ has a path to $w$ in $G'[D'(r'_{\ell'(w')})]$.
			Moreover, there is a path from $\ell'(w')$ to $v$ in $G'[D'(\ell'(w'))]$ as all paths from $\ell'(w')$ to $w$ contain $r'_y$ and $v\in D'(r'_y)$.
			Hence, $\ell'(w')$ is an ancestor of $v'$ in $L'$.
			Therefore, $level'(w') < level'(\ell'(v')) \leq level'(r'_y)$ as $v$ is $L$-affected.
			By Lemma~\ref{lem:out-of-D_y-nothing-changes} we have that $r'_{\ell'(v')}$ is an ancestor of $v$ in $D$, and by Lemma~\ref{lemma:partition-paths} all paths from $\ell'(v')$ to $v$ contain $r'_y$ and can also avoid $(x,y)$ (as it is not a bridge in $G'_s$).
			By Lemma~\ref{lem:out-of-D_y-nothing-changes} we have that  $D'(r'_{\ell'(v')}) = D(r_{\ell(v)})$, and therefore, $\ell'(v')$ has a path to $v$ in $G[D(r_{\ell(v)})]$ (which avoids $(x,y)$).
			As $\ell'(v')$ is not an ancestor of $v$ before the insertion, by Lemma~\ref{lem:path-through-x-y}, all paths from $v'$ to $\ell'(v')$ in $G'[D'(r'_{\ell'(v')})]$ contain the edge $(x,y)$.
			Hence, $x,y,\ell'(v')$, and $v'$ are strongly connected in $G'[D'(r'_{\ell'(v')})]$.
			As $w$ has a path to $x$ in $G'[D'(r'_{\ell'(v')})]$ and $\ell'(v')$ has a path to $v'$ in $G'[D'(r'_{\ell'(v')})]$, it follows that $\ell'(v')$ and $w$ are strongly connected in $G'[D'(r'_{\ell'(v')})]$.
			This implies that $\ell'(v')$ is an ancestor of $w$ in $L'$ with $level'(\ell'(v'))>level'(\ell'(w'))$ which contradicts the definition of $\ell'(w')$.
			The lemma follows.
			\proofend
		\end{proof}

		\begin{lemma}
			If a bridge is not locally canceled by the edge insertion, the set $S'$ of $L$-affected vertices can be identified and $L'$ can be correctly updated in time $O(V(S')+E(S')+V(S)+E(S)+n)$.
			\label{lemma:h-affected-time}
		\end{lemma}
		\begin{proof}
			By Lemma~\ref{lemma:scanned-h-computation}, we can compute $\ell'(v')$ for each $v\in S$ in time $O(V(S)+E(S))$.
			For every vertex $v \in S' \setminus S$, we set its value $\ell'(v')$, according to Lemma~\ref{lemma:new-parent-in-H}.
			Notice that in all cases of Lemma~\ref{lemma:new-parent-in-H}, a vertex $v \in S' \setminus S$, should reach $x$ in $G'[D'(r'_y)]$.
			Therefore, a straightforward way to test whether a vertex in $D(r'_y)$ changes its parent in $L'$, and compute the new parent, is to start a backward traversal in $G'[D'(r'_y)]$ from $x$, and for each vertex $v \notin S$ that is visited by the traversal apply Lemma~\ref{lemma:new-parent-in-H}.
			This, takes $O(m+n)$ in the worst case.
			We next present a charging scheme to achieve the claimed bound.
			
			First, we can identify $lcover(v)$ for all \cvertex{} vertices $v$, after every edge insertion, in time $O(n)$ by traversing the forest $L$ from each \cvertex{} vertex $z \in D'(r'_y), \ell(z) \notin D'(r'_y)$, and setting for each descendant $v$ of $z$ in $L$ the $lcover(v) = q$ such that $(p,q)$ is the strong bridge for which $p\in D'_z$ and $v\in D'(q)$.
			Notice that $(p,q)$ is the $level'(q)$-th bridge in the path from $s$ to $v$ in $D'$.
			In order to identify for each descendant $v$ of $z$ in $L$ the vertex $q$, we proceed as follows.
			Initially, we set to each child of $z$ in $L$ the level of $q$ (which is $level'(z)+1$), that is, each $z$ knows $level'(lcover(v))$.
			Next, we begin a traversal from $s$ on the tree $D'$ keeping track of the bridges that exist on the path from $s$ to the current vertex $v$, and once we visit a vertex $v$ that is assigned a value $level'(lcover(v))$, we set $lcover(v)$ to be $q$ where $(p,q)$ is the $level'(lcover(v))$-th bridge on the path from $s$ to $v$ (we keep track of this information during the traversal).
			Therefore, the total time spent in this computation after all edge insertions is $O(mn)$, where $m$ is the number of edges after all insertions.
			
			Next we start a backward traversal from $x$ in $G'[D'(r'_y)]$; that is, we visit all vertices $v\in D'(r'_y)$ that have an edge to $x$, and consecutively to each visited vertex.
			During the traversal we act as follow.
			Whenever the traversal reaches an unvisited vertex $v$, we test whether $\ell'(v') \not = c'_{\ell(v)}$: if this is the case,
			we set $\ell'(v')$ according to Lemma~\ref{lemma:new-parent-in-H}, we iterate over the incoming edges of $v$ and recursively traverse each vertex $w$ that has an incoming edge to $v$.
			If $\ell'(v') = c'_{\ell(v)}$ and $v \notin D'_y$ we do not traverse any incoming edge to $v$, and instead, we continue the traversal from $lcover(v)$ ($=lcover'(v')$ since $\ell'(v') = c'_{\ell(v)}$) if $lcover(v) \not= \emptyset$, or otherwise we do not visit any vertices from $v$.
			It is correct to continue the traversal from $lcover(v)$ in the case where $\ell'(v') = c'_{\ell(v)}$ and $v \notin D'_y$ since by Lemma~\ref{lemma:successive-candidate} for every vertex $w \in D'(lcover(v))$, such that $w$ reaches $v$ in $G'[D'(lcover(v))]$ we have that $\ell'(c'_w) = c'_{\ell(c_w)}$.
			If, on the other hand, $\ell'(v') = c'_{\ell(v)}$  and  $v \in D'_y$, then we backtrack the traversal from $v$.
			This is correct as $level'(\ell(v)) \geq level'(\ell'(y'))$ follows from Lemma \ref{lemma:new-parent-in-H}, and therefore for all the vertices $z\in D'(r'_y)$ that have a path to $v$ in $G'[D'(r'_y)]$,
			$c'_{\ell(v)}$ is an ancestor of $c'_z$ in $L'$ (since $z$ and $\ell(v)$ are strongly connected in $G'[D'(r'_{\ell(v)})]$).
			Moreover, whenever the traversal reaches a vertex in $S$, we have time to traverse all the vertices in $S$ and their incoming edges; we continue the traversal without testing or updating its value $\ell'(v')$.
			By the above description it is clear that our traversal does not traverse the edges of a vertex that is neither $L$-affected nor $D$-scanned (that is, the vertices in $V\setminus S' \cup S$).
			Collectively, we traverse only $D$-scanned vertices, $L$-affected vertices, and vertices for which we test in constant time whether $\ell'(v') = c'_{\ell(v)}$.
			Thus, we spend time $O(V(S')+E(S')+V(S)+E(S)+n)$.
			
			Next we show that we correctly update the forest $L'$ after an edge insertion.
			Clearly, the vertices whose edges are traversed by the traversal contain the correct value in $L'$.
			Now we argue that all vertices $v$ such that $\ell'(v') \not = c'_{\ell(v)}$ correctly change their parent in $L'$.
			Since for all visited vertices we update correctly the value $\ell'(v')$, we need to show that the backwards traversal that we execute visits all vertices that change their value $\ell'(v')$.
			We already showed in Lemma~\ref{lem:new-parent-in-h-3} that we do this correctly for the vertices in $S$.
			Assume, by contradiction, that this is not true for some vertex $v \notin S$, and therefore, the algorithm fails to set $\ell'(v')$ according to Lemma~\ref{lemma:new-parent-in-H}.
			Let $z$ be the correct value $\ell'(v')$ that the algorithm fails to assign.
			Since $v$ has a path $P$ to $z$ in $G'[D'(r'_{z})]$ only after the insertion of $(x,y)$, $P$ goes through $(x,y)$.
			Let $P_{vx}$ be the subpath of $P$ from $v$ to $x$.
			Moreover, let $w$ be the first vertex on $ P_{vx} $ that is visited by the traversal, i.e., the traversal did not visited any vertices that appear before $w$ on $ P_{vx} $ (recall that we assume the traversal did not visited $v$).
			We know that $\ell'(c'_w) = c'_{\ell(c_w)}$ since otherwise the traversal visits the predecessor of $w$ on $P_{vx}$.
			Assume first that $w\in D'(r'_{y}) \setminus D'_y$.
			Let $(p,q)$ be the first strong bridge on $D'[r'_{\ell'(c'_w)}, w]$.
			Vertex $v$ is not a descendant of $q$ in $D'$, since otherwise by Lemma~\ref{lemma:successive-candidate} it holds $\ell'(v') = c'_{\ell(v)}$, which contradicts our assumption.
			Thus, $v$ is not a descendant of $q$ in $D'$, and therefore, all paths from $v$ to $w$ go through $(p,q)$, by the properties of the dominator tree.
			That means $p \in P_{vx}$ and $p$ appears earlier than $w$ on $P_{vx}$.
			According to the rules of the traversal, if $w$ is visited and $\ell'(c'_w) = c'_{\ell(c_w)}$, then the traversal continues from $p$, so $p$ must have visited during the traversal.
			This contradicts the choice of $w$ as the earliest vertex on $P_{vx}$ that is visited.
			If $w \in D'_y$, then by the assumption that $\ell'(c'_w) = c'_{\ell(c_w)}$ and by Lemma \ref{lemma:new-parent-in-H}, we have that $level'(y) > level'(\ell'(c'_w)) \geq level'(\ell'(y'))$.
			Since $v$ has a path to $w$ in $G'[D'(r'_{y'})] \subset G'[D'(r'_{\ell'(c'_w)})]$ and $\ell'(c'_w)$ has path to $v$ in $G'[D'(r'_{\ell'(c'_w)})]$ (through $r'_y$; such a path exists since $\ell'(c'_w)$ has a path to $w$ in $G'[D'(r'_{\ell'(c'_w)})]$), it follows that $v$ and $\ell'(c'_w)$ are strongly connected in $G'[D'(r'_{\ell'(c'_w)})]$.
			Therefore,  $level'(\ell'(v')) > level'(\ell'(y'))$, which contradicts the assumption that $\ell'(v') \not= c'_{\ell(v)}$ according to Lemma \ref{lemma:new-parent-in-H}.
			Thus, all vertices $v$ for which $\ell'(v') \not= c'_{\ell(v)}$ are visited by the traversal, and the lemma follows.
			\proofend
		\end{proof}
		
		Finally, we bound the total time spend over any sequence of $m$ insertions.
		
		\begin{lemma}
			\label{lem:final-running-time-strongly-connected}
			After any sequence of edge insertions in a flow graph $G_s$, any \cvertex{} vertex $v$ changes its parent $\ell'(c'_v)$ in $L'$ at most $t$ times, where $t < n$ is the number of bridges dominating $v$ in $D$ before any insertion.
			Moreover, $L$-affected vertices can be identified and correctly updated in a total of $O(mn)$ time for all edge insertions, where $m$ is the number of edges after all edge insertions.
		\end{lemma}
		\begin{proof}
			We first bound the number of times that a vertex $v$ can be $L$-affected but not $D$-scanned, that is $v\in S' \setminus S$.
			Note that the number of bridges in a flow graph $G_s$ is at most $n-1$.
			Therefore, the number of distinct bridges that appear on the path $D[s,v]$ for each vertex $v$ throughout the course of the algorithms is at most $n-1$.
			Our strategy is to show that for each strong bridge on $D[s,v]$ for any vertex $v$ a vertex is $L$-affected and not $D$-scanned at most once.
			We first claim that once a vertex $w$ becomes an ancestor of $v$ in $L$, such that $v\in D(r_w) \setminus D_w$, $w'$ remains an ancestor of $c'_v$, after any edges insertion, as long as $v'\in D(r'_w) \setminus D'_w$.
			For every vertex $z$ on the path from $w$ to $v$ in $G[D(r_w)]$ we have that $z \in D'(r'_w)$ as otherwise there is a path from a vertex $z$ to $v$ avoiding $r'_w$, which  means a path from $s$ to $v$ avoiding $r'_w$, a contradiction to the fact that $v\in D'(r'_w)$.
			Hence, there is a path from $w$ to $v$ in $G'[D'(r'_w)]$.
			We use the same argument to show that there is a path from $v$ to $w$ in $G'[D'(r'_w)]$.
			For every vertex $z$ on the path from $v$ to $w$ in $G[D(r_w)]$ remains in $D'(r'_w)$ as otherwise there is a path from a vertex $z$ to $w$ avoiding $r'_w$, which means a path from $s$ to $w$ avoiding $r'_w$, a contradiction to the fact that $w\in D'(r'_w)$.
			Therefore, there is a path from $v$ to $w$ in $G'[D'(r'_w)]$.
			Thus, $w$ and $v$ are strongly connected in $G'[D'(r'_w)]$, which means $w'$ is an ancestor of $v'$ in $L'$.
			
			Every time $v$ is $L$-affected but not $D$-scanned we have that $level'(\ell'(v')) > level'(\ell(v))$, as implied by Lemma~\ref{lemma:new-parent-in-H}.
			Therefore, $v$ is assigned an ancestor $\ell'(v')$ such that there is no ancestor $w$ of $v$ where $c'_w=\ell'(v')$ (as otherwise $v$ is be $L$-affected).
			Moreover, as we shown above, as long as $(d'(r'_{\ell'(v')}),r'_{\ell'(v')})$ is a bridge such that $v\in D'(r'_{\ell'(v')}) \setminus D_{\ell'(v')}$, vertex $\ell'(v')$ remains ancestor of $v$ in $L'$.
			Hence, since at most $n-1$ bridges can appear on the path $D[s,v]$, it follows that $v$ can be at most $O(n)$ times $L$-affected but not $D$-scanned.

			Now we bound the overall running time spent on identifying and updating the $L$-affected vertices.
			First note that if an $L$-affected vertex is also $D$-scanned, then by Lemma~\ref{lemma:scanned-h-computation} we can update their parent in $L'$ in time $O(V(S)+E(S))$ where $S$ is the set of $D$-scanned vertices.
			We charge this time to the algorithm for updating the dominator tree, which spends $O(V(S)+E(S))$ time after each edge insertion.
			Thus, the overall time spent on updating the parent of a vertices that are $D$-scanned is $O(mn)$.
			Now, let $S'$ be the set of vertices that are $L$-affected.
			By Lemma~\ref{lemma:h-affected-time}, $S'$ can be identified in time $O(V(S')+E(S')+V(S)+E(S)+n)$.
			We again charge the time $O(V(S)+E(S))$ to the algorithm for updating the dominator tree, which sums to $O(mn)$ after all insertions.
			As we showed above, a vertex $u$ can be in $S' \setminus S$ at most $n-1$ times.
			Hence, the time $O(V(S')+E(S'))$ per insertion considers each vertex at most $n-1$ times, and therefore, it takes time $O(mn)$ time.
			Finally, the time $O(n)$ spent after every edge insertion sums to $O(mn)$ overall, since we have at most $m$ insertions.
			The bound follows.
			\proofend
		\end{proof}
		
\begin{theorem}
\label{theorem:hyperloop-total-time}
Let $G_s$ be a flow graph with $n$ vertices.
We can maintain the hyperloop nesting forest $L$ of $G_s$ through a sequence of edge insertions in $O(mn)$ total time, where $m$ is the number of edges after all insertions.
\end{theorem}

		\section{Answering queries in optimal time}
		\label{sec:queries-edges}

		The data structure from \cite{2C:GIP:arXiv} computes the strong bridges of $G$ plus four trees: the dominator tree $D$ and the loop nesting tree $H$ of the flow graph $G_s$, and the dominator tree $D^R$ and the loop nesting tree $H^R$ of the reverse flow graph $G^R_s$.
		This information is sufficient to answer in optimal time all the following types of queries:
		\vspace{-.25cm}
		\begin{mylist}{(iii)}
			\litem{(i)} Report in $O(1)$ time the total number of SCCs in $G\setminus e$, for a query edge $e$ in $G$.
			\litem{(ii)} Report in $O(1)$ time the size of the largest and of the smallest SCCs in $G\setminus e$, for a query edge $e$ in $G$.
			\litem{(iii)} Report in $O(n)$ worst-case time all the SCCs of $G\setminus e$, for a query edge $e$.
			\litem{(iv)} Test in $O(1)$ time if two query vertices $u$ and $v$ are strongly connected in $G\setminus e$, for a query edge $e$.
			\litem{(v)} For query vertices $u$ and $v$ that are strongly connected in $G$,
			report all edges $e$ such that $u$ and $v$ are not strongly connected in $G\setminus e$, in optimal worst-case time, i.e., in time $O(k+1)$, where $k$ is the number of separating edges.
		\end{mylist}
		\vspace{-.1cm}
		
		We note that queries of type {(i)} and {(ii)} require additional $O(n)$ preprocessing time, based solely on the same four trees.
		In particular, the crux of the method is the following theorem, which shows that the information relevant for our queries can indeed be extracted from the strong bridge of $G$ and the four trees $D$, $D^R$, $H$ and $H^R$:

		\begin{theorem}[\cite{2C:GIP:arXiv}]
			\label{cor:scc}
			Let $G=(V,E)$ be a strongly connected digraph, $s$ be an arbitrary start vertex in $G$, and let $e=(u,v)$ be a strong bridge of $G$. Let $C$ be a SCC of $G \setminus e$. Then one of the following cases holds:
			\vspace{-.25cm}
			\begin{mylist}{(a)}
				\litem {(a)} If $e$ is a bridge in $G_s$ but not in $G_s^R$ then either $C \subseteq D(v)$ or $C = V \setminus D(v)$.
				\litem {(b)} If $e$ is a bridge in $G_s^R$ but not in $G_s$ then either $C \subseteq D^R(u)$ or $C = V \setminus D^R(u)$.
				\litem {(c)} If $e$ is a common bridge of $G_s$ and $G_s^R$ then either $C \subseteq D(v) \setminus D^R(u)$, or $C \subseteq D^R(u) \setminus D(v)$, or $C \subseteq D(v) \cap D^R(u)$, or $C = V \setminus \big( D(v) \cup D^R(u) \big)$.
			\end{mylist}
			\vspace{-.25cm}
			Moreover, if  $C \subseteq D(v)$ \textup{(}resp., $C \subseteq D^R(u)$\textup{)} then $C=H(w)$ \textup{(}resp., $C=H^R(w)$\textup{)} where $w$ is a vertex in $D(v)$ \textup{(}resp., $D^R(u)$\textup{)} such that $h(w) \not \in D(v)$ \textup{(}resp.,  $h^R(w) \not \in D^R(u)$\textup{)}.
		\end{theorem}

		Now we show that exactly the same information can be extracted if we replace the loop nesting trees $H$ and $H^R$ with two new trees $\hat{H}$ and $\hat{H}^R$, which (differently from loop nesting trees) can be maintained efficiently throughout any sequence of edge insertions. As a result,
		the strong bridges of $G$ plus $D$, $D^R$, $\hat{H}$ and $\hat{H}^R$ allow us to answer all our queries in optimal time throughout any sequence of edge insertions.

		We next define the new trees
		$\hat{H}$ and $\hat{H}^R$.
		Without loss of generality, we restrict our attention to $\hat{H}$, as $\hat{H}^R$ is defined in the reverse graph $G^R$ in a completely analogous fashion.
		We construct $\hat{H}$ starting from the \newLNT{} tree $L$, as follows.
		For every vertex $u$ such that $c_u \not = u$ we set $\hat{h}(u) = c_u$, and for every vertex $u$ where $c_u= u, u \not = s$ we set $\hat{h}(u) = \ell(u)$.
		Note that, once  $L$ is available, the tree $\hat{H}$ can be computed in $O(n)$ time.

		As suggested by Theorem~\ref{cor:scc}, every SCC $C$ in $G\setminus (u,v)$ is either a subtree of $H$ rooted at a vertex $w \in D(v)$ such that $h(w) \notin D(v)$, or a subtree of $H^R$ rooted at a vertex $z \in D^R(u)$ such that $h^R(z) \notin D^R(u)$, or $C = V \setminus D(v) \cup D^R(u)$.
		As a consequence, in order to show that we can safely replace $H$ by $\hat{H}$ and $H^R$ by $\hat{H}^R$,
		we only need to prove the following lemma, which holds symmetrically also for $G_s^R$, $D^R$, $H^R$ and $\hat{H}^R$.
		First, we start with an intermediate technical lemma that is used in the proof or Lemma \ref{lemma:replace}.

		\begin{lemma}
			\label{lemma:nca-canonical-loop-nesting}
			Let $w$ be the nearest common ancestor of two vertices $u$ and $v$ in $H$, and let $z$ be the nearest common ancestor of $c_u$ and $c_v$ in $L$.
			Then, $c_w = z$.
		\end{lemma}
		\begin{proof}
			It suffices to show that, for every ancestor $u$ of a vertex $v$ in $H$,
			$c_u$ is an ancestor of $c_v$ in $L$.
			First, note that if $c_u$ = $c_v$ then they are in the same auxiliary component and thus the above statement trivially holds.
			Now let $c_u \not = c_v$.
			By Lemma~\ref{lemma:loops-nestings-relation}, $v$ and $h_v$ map to the same vertex in $L$. 
			Moreover, $c_{h(h_v)}$ is the parent of $c_{\ell}$ in $L$.
			By repeatedly applying the same argument, it follows that for every ancestor $u$ of $v$ in $H$,
			$c_u$ is an ancestor of $v$ in $L$.
			\proofend
		\end{proof}

		\begin{lemma}
			Let $(u,v)$ be a strong bridge in $G_s$. For every set $H(w)$ where $w \in D(v)$ and $h(w) \notin D(v)$ there is a vertex $z \in D(v)$ and $\hat{h}(z) \notin D(v)$ such that $\hat{H}(z) = H(w)$.
			Additionally,
			for every set $\hat{H}(z)$ where $z \in D(v)$ and $\hat{h}(z) \notin D(v)$ there is a vertex $w \in D(v)$ and $h(w) \notin D(v)$ such that $H(w) = \hat{H}(z)$.
			\label{lemma:replace}
		\end{lemma}
		\begin{proof}
			We prove the lemma by proving
			the following two statements:
			\vspace{-.25cm}
			\begin{mylist}{(ii)}
				\litem{(i)} Let vertex $w \in D(v)$ such that $h(w) \notin D(v)$. Then, for every vertex $p \in H(w)$ there is a vertex $z \in D(v)$ such that $p\in \hat{H}(z)$.
				\litem{(ii)} Let vertex $z \in D(v)$ such that $\hat{h}(z) \notin D(v)$. Then, for every vertex $t \in \hat{H}(z)$ there is a vertex $w \in D(v)$ such that $t\in H(w)$.
			\end{mylist}
			\vspace{-.25cm}
			We start with statement {(i)} for $z = c_w$.
			By Lemma \ref{lemma:loops-nestings-relation}, for each vertex $p \in H(w) \cap D_v$ we have that $c_p = c_w$. By definition of $\hat{H}$, it follows that $p$ is a child of $c_w$ in $\hat{H}$ (thus, $p\in \hat{H}(c_w)$).
			For any other vertex $p \in H(w)$,
			the nearest common ancestor of $p$ and $w$ in $H$ is $w$.
			By Lemma \ref{lemma:nca-canonical-loop-nesting}, 
			the nearest common ancestor of $c_w$ and $c_p$ is $c_w$.
			Therefore, $c_p$ is a descendant of $c_w$ in $\hat{H}$ and $p$ is a child of $c_p$ in $\hat{H}$ (or $p = c_p$).
			Thus, $p\in \hat{H}(c_w)$.
			This proves that for every vertex $p\in H(w)$,
			$p\in \hat{H}(c_w)$.
			
			We not turn to statement (ii).
			Since $\hat{h}(z) \notin D(v)$ it holds that $z$ and $\hat{h}(z)$ are not in the same auxiliary component.
			Therefore, by definition of $\hat{H}$,
			$c_{z} = z$.
			We prove statement (ii) for $w = h_{z}$.
			Since $z \in H(w)$, so do vertices $t$ such that $c_{t} = z$ (including $w$), since they are children of $z$ in $\hat{H}$.
			Now let $t$ be a vertex in $\hat{H}(z)$.
			By definition of $\hat{H}$,
			$t$ is a child of $c_{t}$ in $\hat{H}$ (or $c_{t}$).
			Hence,
			the nearest common ancestor of $c_{t}$ and $c_{z}$ in $\hat{H}$, and therefore in $L$, is $z$.
			Thus, by Lemma \ref{lemma:nca-canonical-loop-nesting},
			for the nearest common ancestor $q$ of $p$ and $z$ in $H$, we have that $c_q = z$.
			This implies that $q$ is a child of $z$ in $H$ and thus $p \in H(w)$.
			\proofend
		\end{proof}
		
		In summary, our algorithm works as follows. Given a strongly connected digraph $G$ subject to edge insertions, we maintain in a total of $O(mn)$ time the strong bridges of $G$ \cite{GIN16:ICALP}, the dominator trees $D$ and $D^R$ \cite{dyndom:2012}, and the \newLNT{} trees $L$ and $L^R$, by Theorem
		\ref{theorem:hyperloop-total-time}.
		After each edge insertion, we construct in $O(n)$ time the trees $\hat{H}$ and $\hat{H}^R$ from $L$ and $L^R$, respectively. Since there can be at most $m$ edge insertions,
		where $m$ is the final number of edges after all edge insertions, the total time spent on all those computations is $O(mn)$.
		By Lemma \ref{lemma:replace}, after each update we can answer all our queries in optimal time.

		\begin{corollary}
		\label{corollary:edge-queries}
			We can maintain a strongly connected digraph $G$ through any sequence of edge insertions in a total of $O(mn)$ time, where $m$ is the number of edges after all insertion, so as to  answer the following queries in optimal time after each insertion:
			\begin{mylist}{(iii)}
				\vspace{-.25cm}
				\litem{(i)} Report in $O(1)$ time the total number of SCCs in $G\setminus e$, for a query edge $e$ in $G$.
				\litem{(ii)} Report in $O(1)$ time the size of the largest and of the smallest SCCs in $G\setminus e$, for a query edge $e$ in $G$.
				\litem{(iii)} Report in $O(n)$ worst-case time all the SCCs of $G\setminus e$, for a query edge $e$.
				\litem{(iv)} Test in $O(1)$ time if two query vertices $u$ and $v$ are strongly connected in $G\setminus e$, for a query edge $e$.
				\litem{(v)} For query vertices $u$ and $v$ that are strongly connected in $G$, report all edges $e$ such that $u$ and $v$ are not strongly connected in $G\setminus e$, in optimal worst-case time, i.e., in time $O(k+1)$, where $k$ is the number of separating edges.
			\end{mylist}
			\vspace{-.1cm}
		\end{corollary} 

\subsection{Extension to general graphs}
\label{sec:extension-to-general-graphs}

In this subsection we extend Corollary~\ref{corollary:edge-queries} to general (not necessarily strongly connected) graphs within the same time bounds. 
The main idea is to maintain the necessary structures for each SCC of the input graph $G$, rooted at suitable vertices. As edges are inserted into $G$, several SCCs may merge, so we need to update our structures accordingly. 
In \cite{GIN16:ICALP} it is shown how to maintain the dominator tree, the bridge decomposition, and the auxiliary components, all rooted at the same start vertex $s$, of each SCC of a general graph under any sequence of edge insertions in total time $O(mn)$.
Since we always maintain the \newLNT{} tree rooted at the same start vertex as the dominator tree, 
we will be using the same start vertices as the algorithm in \cite{GIN16:ICALP}.
Next we briefly review these choices of start vertices.

The algorithm in \cite{GIN16:ICALP} maintains incrementally the SCCs of the graph using the algorithm from \cite{Bender:IncCycleDetection:TALG}.
In each SCC the algorithm maintains an instance of the data structure that we developed for strongly connected graphs.
Whenever an edge insertion causes two or more SCCs to merge, a new data structure instance is created for the newly merged SCC as follows.
Let $C_1, C_2, \ldots, C_j$ be the SCCs that are merged into $C$ after the insertion of an edge. 
We choose the start vertex of $C$ to be the start vertex of the largest SCC $C_i$ and we restart the algorithm in $C$.
We refer to the component $C_i$ as the \emph{principal component of $C$}.

We now show that the choices of start vertices allow our algorithm from Section \ref{sec:canonical-loop-update} to run in a total of $O(mn)$ time, when executed in each SCC independently.
The total time required to maintain instances of the algorithm in the SCCs of the input graph $G$ is $O(mn)$.
The total number of new SCCs that can be created is at most $n-1$. 
In the insertion of an edge results into a merge between two SCCs, the algorithm restarts in the newly merged component: each time we restart the algorithm, it takes $O(m)$ time to initialize the data structures.
We remark that these restarts are different from the restarts described in Section \ref{sec:canonical-loop-update}, that are executed whenever a strong bridge is locally canceled.
We refer to this new type of restarts as \emph{top-level restarts}.
In summary, the total time required to maintain the SCCs and the time spent for the top-level restarts of the algorithm is $O(mn)$.

By Lemma \ref{lemma:h-affected-time}, after each edge insertion, that does not result to a restart of the algorithm, the \newLNT{} tree is updated in time $O(V(S')+E(S')+V(S)+E(S)+n')$ where $S'$ and $S$ are the set of $L$-affected and $D$-scanned vertices after the edge insertion, and $n'$ is the number of vertices inside the SCCs of both $x$ and $y$.
To bound the running time, we show that each vertex $v$ can be $L$-affected at most $O(n)$ times and $D$-scanned at most $O(n)$ times, through the whole execution of the algorithm. 
This holds  despite the fact that a vertex can be part of many SCCs, as they merge while the graph undergoes edge insertions.
Let $C$ be the current SCC containing $v$, and let $C'$ be the new SCC that contains $C$ after a merge caused by an edge insertion.
If $C$ is the principal subcomponent of $C'$, then the depth of $v$ may only decrease as we keep the same start vertex in $C'$ as in $C$. 
Otherwise, the depth of $v$ may increase.
The \emph{effective depth} of $v$ after merging $C$ into $C'$ is defined as follows: it is zero, if $C$ is the principal subcomponent of $C'$, and is equal to the depth of $v$ in the dominator tree of $G_s[C']$ otherwise.
To bound the total amount of work needed to maintain the \newLNT{} tree of all SCCs, we compute the sum of the effective depths of $v$ in all the SCCs that $v$ is contained through the execution of algorithm. We refer to this sum as the \emph{total effective depth of $v$}.

\begin{lemma}[\cite{GIN16:ICALP}]
	\label{lemma:effective-depth}
	The total effective depth of any vertex $v$ is $O(n)$.
\end{lemma}

Therefore, each vertex can be $D$-scanned at most $O(n)$ times.
Moreover, by Lemma \ref{lem:final-running-time-strongly-connected},
a vertex $v$ can be $L$-affected at most $t$ times, where $t$ is the number of strong bridges that dominate over $v$ before any edge insertion takes place. Clearly, $t \leq n'$, where $n'$ is the number of vertices inside the SCC containing $v$.
Note that, whenever some SCCs merge into one SCC $C'$, if $v$ belongs to the principal component of $C'$, then the number of bridges dominating $v$ in $D$ may only decrease. Otherwise, the number of strong bridges dominating $v$ may increase up to $|V(C')|$, and therefore $v$ may be $L$-affected $|V(C')|$ times again in the future.
By summing these numbers each time some SCCs merge, the number of times that a vertex can be $L$-affected equals the effective depth of the vertex.
Thus, by Lemma \ref{lemma:effective-depth}, each vertex can become $L$-affected at most $O(n)$ times.

In summary, we have shown that each vertex can be $D$-scanned and $L$-affected at most $O(n)$ times.
By Lemma \ref{lemma:h-affected-time}, it immediately follows that our algorithm runs in total time $O(mn)$ under any sequence of edge insertions, except for the edge insertions that result to restarts of the algorithm.
Each of the restarts takes $O(m)$ time, as we shown previously.
The following lemma from \cite{GIN16:ICALP} shows that, throughout any sequence of edge insertions, at most $O(n)$ strong bridges can appear in the graph, which results to at most $O(n)$ restarts of the algorithm.

\begin{lemma}[\cite{GIN16:ICALP}]
	During any sequence of edge insertions 
	on a general graph, at most $2(n-1)$ strong bridges can appear.
\end{lemma}

Thus, we have the following theorem.

\begin{theorem}
Let $G$ be a general graph with $n$ vertices.
Both the dominator trees $D$ and $D^R$, the \newLNT{} trees $L$ and $L^R$ of each SCC $C$ of $G$, all rooted at the same arbitrary start vertex $s$, can be maintained in a total of $O(mn)$ time under any sequence of edge insertions, where $m$ is the number of edges after all insertions.
\end{theorem}

\section{Answering queries under vertex failures}
\label{sec:extension-to-vertices}

In this section we extend the queries from Section \ref{sec:queries-edges}, with respect to vertex failures instead of edge failures. More specifically, given a digraph $G$, we show how we can answer in asymptotically optimal (worst-case) time the following type of queries under vertex failures:
	\vspace{-.25cm}
\begin{mylist}{(iii)}
	\litem{(i)} Report in $O(1)$ time the total number of SCCs in $G\setminus v$, for a query vertex $v \in V$.
	\litem{(ii)} Report in $O(1)$ time the size of the largest and of the smallest SCCs in $G\setminus v$, for a query vertex $v\in V$.
	\litem{(iii)} Report in $O(n)$ time all the SCCs of $G\setminus v$, for a query vertex $v\in V$.
	\litem{(iv)} Test in $O(1)$ time if two query vertices $u$ and $w$ are strongly connected in $G\setminus v$, for a query vertex $v$.
	\litem{(v)} For query vertices $u$ and $w$ that are strongly connected in $G$,
	report all vertices $v$ such that $u$ and $w$ are not strongly connected in $G\setminus v$, in time $O(k+1)$, where $k$ is the number of separating vertices.
\end{mylist}
\vspace{-.1cm}

Let $G=(V,E)$ be a strongly connected digraph and $s\in V$ be an arbitrary start vertex.
For any vertex $u$ in $G$, we let $\widetilde{D}(u)$ \textup{(}resp., $\widetilde{D}^R(u)$\textup{)} denote the set of proper descendants of $u$ in $D$ \textup{(}resp., $D^R$\textup{)}, i.e., $\widetilde{D}(u) = D(u) \setminus u$ \textup{(}resp., $\widetilde{D}^R(u) = D^R(u) \setminus u$\textup{)}.
Clearly, $\widetilde{D}(u)\not= \emptyset$ \textup{(}resp., $\widetilde{D}^R(u) \not= \emptyset$\textup{)} if and only if either $u$ is a nontrivial dominator of $G_s$ \textup{(}resp., $G_s^R$\textup{)} or $u=s$.
Analogously, let $D^R$ and $H^R$ be the dominator tree and the loop nesting tree of $G^R$, respectively.

In order to answer the above queries we use the framework developed in \cite{2C:GIP:arXiv}, that characterizes the SCCs of $G\setminus v$, for any $v \not= s$ that is a strong articulation point, according to the following lemma.

\begin{theorem} [\cite{2C:GIP:arXiv}]
	\label{theorem:vertices-scc}
	Let $u$ be a strong articulation point of $G$, and let $s$ be an arbitrary vertex in $G$.
	Let $C$ be a SCC of $G \setminus u$. Then one of the following cases holds:
	\vspace{-.25cm}
	\begin{mylist}{(a)}
		\litem{(a)} If $u$ is a nontrivial dominator in $G_s$ but not in $G_s^R$ then either $C \subseteq \widetilde{D}(u)$ or $C = V \setminus D(u)$.
		\litem{(b)} If $u$ is a nontrivial dominator in $G_s^R$ but not in $G_s$ then either $C \subseteq \widetilde{D}^R(u)$ or $C = V \setminus D^R(u)$.
		\litem{(c)} If $u$ is a common nontrivial dominator of $G_s$ and $G_s^R$ then either $C \subseteq \widetilde{D}(u) \setminus \widetilde{D}^R(u)$, or $C \subseteq \widetilde{D}^R(u) \setminus \widetilde{D}(u)$, or $C \subseteq \widetilde{D}(u) \cap \widetilde{D}^R(u)$, or $C = V \setminus \big( D(u) \cup D^R(u) \big)$.
		\litem{(d)} If $u = s$ then $C \subseteq \widetilde{D}(u)$.
	\end{mylist}
	\vspace{-.25cm}
	Moreover, if  $C \subseteq \widetilde{D}(u)$ \textup{(}resp., $C \subseteq \widetilde{D}^R(u)$\textup{)} then $C=H(w)$ \textup{(}resp., $C=H^R(w)$\textup{)} where $w$ is a vertex in $\widetilde{D}(u)$ \textup{(}resp., $\widetilde{D}^R(u)$\textup{)} such that $h(w) \not \in \widetilde{D}(u)$ \textup{(}resp.,  $h^R(w) \not \in \widetilde{D}^R(u)$\textup{)}.
\end{theorem}

\begin{figure}[t!]
	\begin{center}
		\includegraphics[trim={0cm 1cm 1cm 4cm}, clip=true, width=1.0\textwidth]{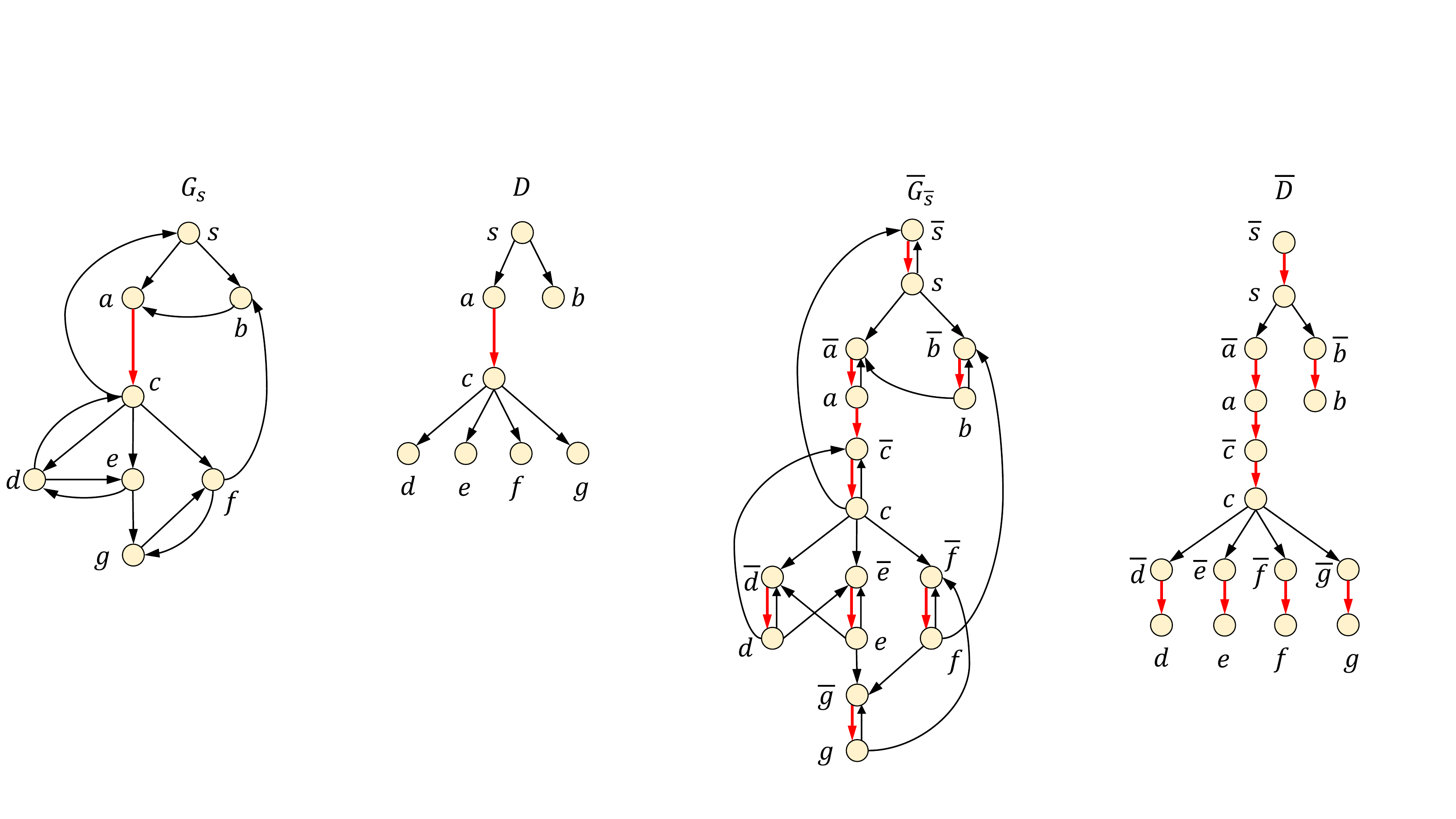}	
		\caption{A strongly connected directed flow graph $G_s$, the dominator tree $D$ of $G_s$, the graph $\overline{G}_{\overline{s}}$, and the dominator tree $\overline{D}$ of $\overline{G}_{\overline{s}}$. The bridges of $G_s$ and $\overline{G}_{\overline{s}}$ are shown in red.}
		\label{figure:overline-graph}
	\end{center}
\vspace{-.7cm}
\end{figure}

Next we show a reduction from answering queries under vertex failures to answering queries under edges failures, which allows us to exploit the incremental algorithm from Section \ref{sec:queries-edges}.
For the reduction we built on Lemma \ref{theorem:vertices-scc} and on 
other properties that were shown in \cite{2C:GIP:arXiv}.
We construct a new graph $\overline{G}=(\overline{V},\overline{E})$ that results from $G$ by applying the following transformation.
For each strong articulation point $x$ of $G$ (i.e., nontrivial dominator in $G_s$ or $G^R_s$), we add an auxiliary vertex $\overline{x} \in \overline{V}$ and add the auxiliary edges $(\overline{x},x)$ and $(x,\overline{x})$. 
Moreover, we replace each edge $(u, x)$ entering $x$ in $G$ with an edge $(u,\overline{x})$ entering $\overline{x}$ in $\overline{G}$.
Note that this transformation maintains the strong connectivity of $\overline{G}$. 
Call the vertices of $V\subset \overline{V}$ ordinary.
The resulting graph $\overline{G}$ has $n$ auxiliary vertices and $2n$ auxiliary edges.
Hence, $|\overline{V}| = 2n$ and $|\overline{E}| = m+2n$. 
Finally, we choose vertex $\overline{s}$ as the start vertex of $\overline{G}$.
See Figure \ref{figure:overline-graph}.

Let $\overline{D}$ and $\overline{D}^R$ be the dominator trees of $\overline{G}_{\overline{s}}$ and $\overline{G}^R_{\overline{s}}$, respectively.
The following lemma states the correspondence between the strong articulation points in $G$ (except for $s$) and the strong bridges in $\overline{G}$. See Figure \ref{figure:overline-graph}.

\begin{lemma}
	Let $x \not =s$ be a strong articulation point of digraph $G$. Then the following properties hold:
	\vspace{-.25cm}
	\begin{mylist}{(iii)}
		\litem{(i)} Let $x$ be a nontrivial dominator of $G_s$ (resp., $G^R_s$).
		Then the auxiliary edge $(\overline{x},x)$ is a strong bridge of $\overline{G}$ and a bridge
		of $\overline{G}_{\overline{s}}$ \textup{(}resp., $\overline{G}^R_{\overline{s}}$\textup{)}.
		\litem{(ii)} For any ordinary vertex $u \in V \setminus x$, we have that $u \in D(x)$ \textup{(}resp., $u\in D^R(x)$\textup{)} if and only if $u \in \overline{D}(x)$ \textup{(}resp., $u \in \overline{D}^R(\overline{x})$\textup{)}.
		\litem {(iii)} All vertices in a SCC of $G\setminus x$ are the ordinary vertices in a SCC of $\overline{G} \setminus (\overline{x},x)$.
	\end{mylist}
	\vspace{-.25cm}
	\label{lemma:vertex-to-edge-reduction}
\end{lemma}
\begin{proof}
	We prove (i) and (ii) only for the case where $x$ is a nontrivial dominator in $G_s$ since the case where $x$ is a nontrivial dominator in $G^R_s$ is completely analogous.
	Note that $x$ has indegree one in $\overline{G}$, and hence $(\overline{x},x)$ is a strong bridge of $\overline{G}$.
	By the fact that $\overline{G}$ is strongly connected, there is at least one path from $\overline{s}$ to $x$. 
	All such paths must contain $(\overline{x},x)$, so $(\overline{x},x)$ is a bridge in $\overline{G}_{\overline{s}}$.
	This implies (i).
	
	In order to prove (ii), we show that $\overline{s}$ has a path to an ordinary vertex $v$ avoiding $(\overline{x},x)$ in $\overline{G}$ if and only if $s$ has a path to $v$ avoiding $x$ in $G$.
	We start with one direction.
	Let $W$ be the set of vertices such that all paths from $\overline{s}$ to vertices in $W$ contain $(\overline{x},x)$.
	Clearly there is no edge $(w,z) \not=(\overline{x},x)$ such that $w\notin W$, $z\in W$, as that means there is a path from $\overline{s}$ to $z \in W$ avoiding $(\overline{x},x)$, which contradicts the fact that $z \in W$.
	Therefore, by the construction of $\overline{G}$ it follows that there is no edge $(w,z)$ in $G$, such that $w \notin W,  w \not= x$ and $z\in W$.
	Thus all paths from $s \notin W$ to ordinary vertices in $W$ contain $x$.
	We now prove the opposite direction of our claim: namely, if $x$ appears in all paths from $s$ to a vertex $v$ in $G$ then $(\overline{x},x)$ appears in all paths from $\overline{s}$ to $v$ in $\overline{G}$.
	Let $W$ be the set of vertices such that all paths from $s$ to vertices in $W$ contain $x$.
	Clearly there is no edge $(w,z)$ such that $w\notin W, w \not= x$, $z\in W$, as that implies that there is a path from $s$ to $z \in W$ avoiding $x$, which contradicts the fact that $z \in W$. 
	Therefore, by the construction of $\overline{G}$ it follows that there is no edge $(w,\overline{z})$, such that $w\notin W, w \not=x$ and $z\in W$.
	Thus all paths from $\overline{s}$ to vertices in $W$ contain $(\overline{x},x)$.
	Thus, (i) and (ii) follow.
	
	We finally prove (iii). 
	Assume that two vertices $u$ and $v$, $u \not= v \not= x$, are in the same SCC in $G \setminus x$ but not in the same SCC in $\overline{G} \setminus (\overline{x},x)$.
	The fact that $u$ and $v$ are not strongly connected in $\overline{G}\setminus (\overline{x},x)$ implies that either all paths from $u$ to $v$ contain $(\overline{x},x)$ or all paths from $v$ to $u$ contain $(\overline{x},x)$.
	Without loss of generality, assume that all paths from $u$ to $v$ in $\overline{G}$ contain the strong bridge $(\overline{x},x)$.
	Since $u$ and $v$ are in the same SCC in $G\setminus x$, there is a path $\pi$ from $u$ to $v$ in $G$ that avoids all edges incoming to $x$ and all edges outgoing to $x$.
	That means in $\overline{G} \setminus (\overline{x},x)$ the corresponding path $\overline{\pi}$ of $\pi$ avoids all incoming edges to $\overline{x}$ and all outgoing edges from $x$.
	This contradicts the fact that all paths from $u$ to $v$ in $\overline{G}$ contain $(\overline{x},x)$.
	Now suppose that in a SCC in $\overline{G} \setminus (\overline{x},x)$, there are two ordinary vertices $u$ and $v$, $u \not= v \not= x$, that are in different SCC in $G \setminus x$.
	Therefore, there is a path $\overline{\pi}$ in $\overline{G}$ that avoids $(\overline{x},x)$, and thus vertices $\overline{x}$ and $x$ by construction.
	Then the corresponding path $\pi$ of $\overline{\pi}$ in $G$ avoids $x$, a contradiction.
	\proofend
\end{proof}

Now we combine Lemmas \ref{theorem:vertices-scc} and \ref{lemma:vertex-to-edge-reduction} to prove the following lemma.

\begin{lemma}
\label{lem:reduction-from-vertices-to-edges}
Let $G$ be a strongly connected graph with start vertex $s$.
For each strong articulation point $v$ it holds that:
	\vspace{-.25cm}
\begin{mylist}{(iii)}
	\litem{(i)} The total number of SCCs in $G\setminus v$, is equal to the total number of SCCs in $\overline{G}\setminus (\overline{v},v)$ minus $1$.
	\litem{(ii)} The size of the largest SCC in $G\setminus v$, is equal to half of the size of the largest SCC in $\overline{G}\setminus (\overline{v},v)$.
	The size of the smallest SCC in $G\setminus v$, equal to half of the size of the smallest SCC in $\overline{G}\setminus (\overline{v},v)$, excluding the singleton SCCs $\overline{v}$ and $v$.
	\litem{(iii)} The SCCs of $G\setminus v$, are the sets of ordinary vertices of each SCC of $\overline{G}\setminus (\overline{v},v)$ except of the singleton SCCs $v$ and $\overline{v}$.
	\litem{(iv)} Two query vertices $u$ and $w$ are strongly connected in $G\setminus v$, if and only if $u$ and $w$ are strongly connected in $\overline{G}\setminus (\overline{v},v)$.
	\litem{(v)} All vertices $v$ such that $u$ and $w$ are not strongly connected in $G\setminus v$, are the ordinary endpoints of strong bridges $e$ \textup{(}excluding $u$ and $w$\textup{)} such that $u$ and $w$ are not strongly connected in $\overline{G}\setminus e$.
\end{mylist}
\end{lemma}
\begin{proof}
We start with (i).
Notice that in $\overline{G}\setminus (\overline{x},x)$, the vertices $\overline{x}$ and $x$ are singleton SCCs since $(\overline{x},x)$ is the only outgoing and the only incoming edge of $\overline{x}$ and $x$, respectively.
First, notice that each ordinary vertex $v\not=x$ is in the same SCC together with $\overline{v}$ in $\overline{G}$ since $\overline{G}$ contains $(\overline{v},v)$ and $(v,\overline{v})$.
By Lemma \ref{lemma:vertex-to-edge-reduction}(iii), all the vertices in a SCC in $G\setminus x$ are the ordinary vertices of a SCC in $\overline{G}\setminus (\overline{x},x)$.
Therefore, (i) follows.

The cases (ii), (iii), and (iv) follows immediately from Lemma \ref{lemma:vertex-to-edge-reduction}(iii) and the fact that all ordinary vertices are in the same SCC with their auxiliary vertex.

Finally, we consider case (v).
By Lemma \ref{lemma:vertex-to-edge-reduction}(i) and Lemma \ref{lemma:vertex-to-edge-reduction}(ii) it follows for each vertex $x$ such that $u$ and $w$ are not strongly connected in $G\setminus x$, there exists the strong bridge $(\overline{x},x)$ such that $u$ and $w$ are not strongly connected in $\overline{G} \setminus (\overline{x},x)$.
Let now $e = (p,\overline{q})$ be a strong bridge in $\overline{G}$ separating $w$ and $u$, such that $p\not= w \not=u$.
We show that $(\overline{p},p)$ is also a strong bridge in $\overline{G}$ separating $w$ and $u$.
That implies, by Lemma \ref{lemma:vertex-to-edge-reduction}(iii), $p$ separates $w$ and $u$ in $G$.
Without loss of generality, assume all paths from $w$ to $u$ contain $(p,\overline{q})$.
By the fact that the only incoming edge to $p$ is $(\overline{p},p)$ it follows that all paths from $w$ to $u$ in $\overline{G}$ contain $(\overline{p},p)$.
Thus, $(\overline{p},p)$ is a strong bridge in $\overline{G}$ separating $w$ and $u$, and as we mentioned above that implies $p$ is a separating vertex for $w$ and $u$.
Case (v) follows.
\proofend
\end{proof}

Lemma \ref{lem:reduction-from-vertices-to-edges} allows us to answer strong connectivity queries under vertex failures 
using the strong connectivity queries under edge failures from Section \ref{sec:queries-edges}. Similarly to Section \ref{sec:queries-edges}, we only need to maintain incrementally the dominator trees $\overline{D}$ and $\overline{D}^R$, and the \newLNT{} trees $\overline{L}$ and $\overline{L}^R$ of $\overline{G}$ and $\overline{G}^R$, respectively.
The most challenging query is to report the size of the smallest SCC in $G\setminus v$, for a query vertex $v$.
This happens due to the fact that the size of the smallest SCCs in $\overline{G}\setminus (\overline{v},v)$ is always one, namely, 
the singleton SCCs $\overline{v}$ and $v$.
To resolve this problem we apply a minor modification of the algorithm in \cite{2C:GIP:arXiv}, in order to ignore the singleton SCC $v$ when computing the minimum SCC in $G[D(v)]$, and the singleton SCC $\overline{v}$ when computing the minimum SCC in $G^R[D^R(v)]$. In summary, we have the following theorem.

\begin{theorem}
	\label{thrm:main-result}
	We can maintain a digraph $G$ throughout any sequence of edge insertions in a total of $O(mn)$ time, where $m$ is the number of edges after all insertions, and after each edge insertion  the following queries can be answered in asymptotically optimal (worst-case) time:
			\vspace{-.25cm}
	\begin{mylist}{(iii)}
		\litem{(i)} Report in $O(1)$ time the total number of SCCs in $G\setminus v$, for a query vertex $v$ in $G$.

		\litem{(ii)} Report in $O(1)$ time the size of the largest and of the smallest SCCs in $G\setminus v$, for a query vertex $v$ in $G$.
		\litem{(iii)} Report in $O(n)$ time all the SCCs of $G\setminus v$, for a query vertex $v$.
		\litem{(iv)} Test in $O(1)$ time if two query vertices $u$ and $w$ are strongly connected in $G\setminus v$, for a query vertex $v$.
		\litem{(v)} For query vertices $u$ and $w$ that are strongly connected in $G$, report all vertices $v$ such that $u$ and $w$ are not strongly connected in $G\setminus v$, in $O(k+1)$ time, where $k$ is the number of separating vertices.
	\end{mylist}
	\vspace{-.1cm}
\end{theorem} 

\section{Maintaining the $2$-vertex-connected components of a digraph}
\label{sec:2-vertex-connected-components}

In this section we present an algorithm for maintaining incrementally the $2$-vertex-connected components of a digraph $G$ in a total $O(mn)$ of time.
In our algorithm we will use some properties from \cite{2C:GIP:arXiv}. However, we need to extend those properties, by substituting them with strong connectivity queries under edge failures, in order to provide necessary and sufficient conditions for two vertices being $2$-vertex-connected.

We compute the $2$-vertex-connected components of a digraph by computing as an intermediary the following relation.
Two vertices $x$ and $y$ are said to be \emph{vertex-resilient} if the removal of any vertex different from $x$ and $y$
leaves $x$ and $y$ in the same strongly connected component. A \emph{vertex-resilient component}
of a digraph $G=(V,E)$ is defined as a maximal subset $B \subseteq V$ such that $x$ and $y$ are vertex-resilient for all $x, y \in B$. 
As a (degenerate) special case, a vertex-resilient component might consist of a singleton vertex only, in which case we have a \emph{trivial vertex-resilient component}.
We are interested in computing only the nontrivial vertex-resilient components, and since there is no danger of ambiguity, we will call them simply vertex-resilient components.
The following lemma states that we can compute the $2$-vertex-connected components starting from the $2$-edge-connected components and the vertex-resilient components.
Moreover, this computation can be done in $O(n)$ time~\cite{2VCB}.

\begin{lemma} \emph{(\cite{2VCB})}
\label{cor:2vc-resilient}
For any two distinct vertices $x$ and $y$, $x$ and $y$ are $2$-vertex-connected if and only if $x$ and $y$ are both vertex-resilient and $2$-edge-connected.
\end{lemma}

Because, of Lemma \ref{cor:2vc-resilient}, an incremental algorithm for maintaining the $2$-vertex-connected components of a digraph can be immediately obtained from incremental algorithms for maintaining the vertex-resilient and $2$-edge-connected components. Since we know how to maintain incrementally the $2$-edge-connected components of a digraph in a total of $O(mn)$ time \cite{GIN16:ICALP}, in the remainder of this section we will focus on the incremental maintenance of the vertex-resilient components. 

Once again, let $s$ be an arbitrary start vertex in $G$, and let $D$ (resp., $D^R$) be the dominator tree of the flow graph $G_s$ (resp., $G_s^R$).
For any vertex $u$, we denote by $C(u)$ (resp. $C^R(u)$) the set of children of $u$ in $D$ (resp. $D^R$). 
For any pair of vertices $u$ and $v$ we identify a set of vertices $C(u,v)$ defined as follows.
Set $C(u,v)$ contains all vertices in $C(u) \cap C^R(v)$. Also, if $u=v$ or $u \in C^R(v)$ then we include $u$ in $C(u,v)$, and if $v \in C(u)$ then we include $v$ in $C(u,v)$.
These sets can be computed in $O(n)$ time~\cite{LuigiGILP15}.

\begin{corollary} [\cite{2C:GIP:arXiv}]
\label{corollary:vertex-resilient-necessary}
Let $G=(V,E)$ be a strongly connected digraph, and let $s \in V$ be an arbitrary start vertex.
Any two vertices $x$ and $y$ are vertex-resilient only if they are located in a common set $C(u,v)$.
\end{corollary}

Let $x$ and $y$ be two distinct vertices in $G$. 
We say that a strong articulation point $u$ \emph{separates $x$ and $y$} (or equivalently that $u$ is a \emph{separating vertex for $x$ and $y$}) if all paths from $x$ to $y$ or all paths from $y$ to $x$ contain vertex $u$ (i.e., $x$ and $y$ belong to different strongly connected components of $G\setminus u$).
Clearly, $x$ and $y$ are vertex-resilient if and only if there exists no separating vertex for them.
The next two lemmas from \cite{2C:GIP:arXiv} imply that only $O(1)$ vertices need to be tested in order to determine whether there exists a separating vertex for 
$x$ and $y$.

\begin{lemma} [\cite{2C:GIP:arXiv}]
\label{vertices-separators-ancestors}
Let $u$ be nontrivial dominator in either $D$ or $D^R$ that is a separating vertex for vertices $x$ and $y$. Then $u$ must appear in at least one of the paths $D[s,x]$, $D[s,y]$, $D^R[s,x]$, and $D^R[s,y]$.
Let $u$ be a strong articulation point that is a separating vertex for vertices $x$ and $y$. Then $u$ must appear in at least one of the paths $D[s,x]$, $D[s,y]$, $D^R[s,x]$, and $D^R[s,y]$.
\end{lemma}

\begin{lemma} [\cite{2C:GIP:arXiv}]
\label{SAP-relation}
Let $x$ and $y$ be vertices that are either siblings or $x$ is the parent of $y$ in $D$, and let $w=d(x)$.
A strong articulation point $u$ that is not a descendant of $w$ in $D$ is a separating vertex for $x$ and $y$ only if $w$ is a separating vertex for $x$ and $y$.
\end{lemma}

The following corollary summarizes Lemmas \ref{vertices-separators-ancestors} and \ref{SAP-relation}.

\begin{corollary}
\label{cor:2-vertex-connecte-pair}
Two vertices $x$ and $y$ are vertex-resilient if and only if there are some vertices $u,v$ such that $x,y \in C(u,v)$, and $x$ and $y$ are strongly connected in:
\vspace{-.25cm}
\begin{mylist}{(a)}
	\litem{(a)} $G\setminus u$ 
	\litem{(b)} (if $u\not=s$) in $G\setminus d(u)$ 
	\litem{(c)} $G\setminus v$ 
	\litem{(d)} (if $v\not=s$) in $G\setminus d^R(v)$ 
\end{mylist}
\end{corollary}

Note that, by Corollary \ref{cor:2-vertex-connecte-pair}, each vertex-resilient component is contained in a $C(u,v)$ set. 
Thus, the set $C(u,v)$ defines an initial set of ``coarse'' blocks that are supersets of the vertex-resilient components. 
To find the real vertex-resilient components, our algorithm will refine those coarse blocks with the help of the auxiliary components.
The sets  $C(u,v)$ can be represented by a block forest $\mathcal{F}$ of size $O(n)$ as shown in \cite{2C:GIP:arXiv}.
The block forest is a bipartite undirected acyclic graph that contains a node for each vertex $v\in V$ and a node for each vertex-resilient component of the graph; it contains an edge ${\alpha,\beta}$ if $\alpha$ is in the vertex-resilient component $\beta$.
In order to refine the initial block forest, we will use the following operation from \cite{2VCB}.
Let $\mathcal{B}$ be a set of blocks, let $\mathcal{S}$ be a partition of a set $U \subseteq V$, and let $x$ be a vertex not in $U$.

\begin{description}
\item[\emph{refine}$(\mathcal{B}, \mathcal{S}, x)$:] For each block $B \in \mathcal{B}$, substitute $B$ by the sets  $B \cap (S \cup \{ x \} )$ of size at least two, for all $S \in \mathcal{S}$.
\end{description}

As shown in \cite{2VCB}, this operation can be executed in time that is linear in the total number of elements in all sets of $\mathcal{B}$ and $\mathcal{S}$.
The algorithm needs to locate the blocks that contain a specific vertex, and, conversely, the vertices that are contained in a specific block. 
Note that $\mathcal{F}$ is bipartite, so the adjacency list of a vertex $v$ stores the blocks that contain $v$, and the adjacency list of a block node $B$ stores the vertices in $B$.
Initially $F$ contains one block for each $C(u,v)$ set. 
Those blocks are 
refined by means of $\mathit{refine}$ operations.
The executions of $\mathit{refine}$ operations update the block forest $F$, while maintaining the linear execution time.

Our algorithm, called \textsf{Inc2VCC}, computes after every edge insertion the vertex-resilient components of the digraph in $O(n)$ time, using Corollary \ref{cor:2-vertex-connecte-pair}.
The pseudocode is shown in Algorithm \ref{fig:Inc2VCC}.
We first prove the correctness of our algorithm. Next, we show that it runs in $O(n)$ time.

\begin{algorithm}[t!]
	\LinesNumbered
	\DontPrintSemicolon
	\KwIn{Strongly connected digraph $G=(V,E)$}
	\KwOut{The vertex-resilient components of $G$}
	
	\textbf{Dominator trees and initialize block forest:}\;
	Let $D$ and $D^R$ be the dominator tree of $G_s$ and $G^R_s$, respectively.\;
	Compute the sets $C(u,v)$.\label{line:child-parent}\;
	Initialize the block forest $F$ to contain one block for each set $C(u,v)$ with at least two vertices.\label{line:common-child-parent}\;
	
	\textbf{Forward direction:}\;
	\ForEach{$u \in V$ in a bottom-up order of $D$}
	{
		Find the set of blocks $\mathcal{B}$ that contain at least two vertices in $C(u)$.\;
		Compute the collection of vertex subsets $\mathcal{S} = \{S \cap C(u) : \textit{S is an SCC in } G\setminus u\}$.\label{line:gather-components}\;
		Execute $refine(\mathcal{B}, \mathcal{S}, u)$.\label{line:refine-1}\;
		\If{$u \not= s$}
		{
			\ForEach{$B \in \mathcal{B}$  such that $u \in B$ \label{line:query-1}}
			{
				Choose an arbitrary vertex $v \not=u$ in $B$.\;
				\If{$u$ and $v$ are strongly connected in $G\setminus d(u)$}{
					set $B = B \setminus u$ \label{line:remove-1}\;
					\lIf{$|B|=1$}{delete $B$ from $F$}
				}
			}
		}
	}
	\textbf{Reverse direction:}\;
	\ForEach{$u \in V$ in a bottom-up order of $D^R$}
	{
		Find the set of blocks $\mathcal{B}$ that contain at least two vertices in $C^R(u)$.\;
		Compute the collection of vertex subsets $\mathcal{S} = \{ S \cap C ^R(u) : S~ \textit{is an SCC in } G\setminus u\}$.\label{line:gather-components-2}\;
		Execute $refine(\mathcal{B}, \mathcal{S}, u)$.\label{line:refine-2}\;
		\If{$u \not= s$}
		{
			\ForEach{$B \in \mathcal{B}$  such that $u \in B$ \label{line:query-2}}
			{
				Choose an arbitrary vertex $v \not=u$ in $B$.\;
				\If{$u$ and $v$ are strongly connected in $G^R\setminus d^R(u)$}{
					set $B = B \setminus u$\label{line:remove-2}\;
					\lIf{$|B|=1$}{delete $B$ from $F$}
				}
			}
		}
	}
	\caption{\textsf{Inc2VCC}}
	\label{fig:Inc2VCC}
\end{algorithm}

\begin{lemma}
Algorithm \textsf{Inc2VCC} is correct.
\end{lemma}
\begin{proof}
We will prove that two vertices $x$ and $y$ are in the same block of the maintained block forest $\mathcal{F}$ after the end of the algorithm if and only if they satisfy Corollary \ref{cor:2-vertex-connecte-pair}.
Since two vertices are vertex-resilient if and only if they satisfy Corollary \ref{cor:2-vertex-connecte-pair}, it follows that the blocks in $\mathcal{F}$ are exactly the vertex-resilient components of the graph.

We first prove one direction of the claim: assume that $x$ and $y$ satisfy Corollary \ref{cor:2-vertex-connecte-pair}.
The vertices $x$ and $y$ will not be separated in Line \ref{line:common-child-parent} since there exists a pair of vertices $u,v$ such that $x,y\in C(u,v)\cup u$ by Corollary \ref{cor:2-vertex-connecte-pair}.
Now we first assume that $x$ and $y$ are siblings in $D$.
The vertices $x,y$ will not be separated in Line \ref{line:refine-1} since they are strongly connected in $G\setminus u$ by Corollary \ref{cor:2-vertex-connecte-pair} and $x,y\in C(u)\cup u$.
Moreover, the removal carried out in Line \ref{line:remove-1} cannot remove any of $x$ or $y$ from their block.
Now assume, without loss of generality, that $x=d(y)=u$.
The vertices $x,y$ will not be separated in Line \ref{line:refine-1} since $u,y\in C(u)$ and $u$ will be included in the same block $C$ as $x$ by the definition of the $\mathit{refine}$ operation.
Notice that all vertices in $C$ remain strongly connected in $G\setminus d(u)$ by Lemma \ref{SAP-relation}.
Therefore, $u$ is strongly connected in $G\setminus d(u)$ with all vertices in $C$, since $u$ and $x$ are strongly connected by Corollary \ref{cor:2-vertex-connecte-pair} (recall that we assumed $x$ and $y$ satisfy Corollary \ref{cor:2-vertex-connecte-pair}).
The same arguments can be used to show that $x$ and $y$ are not separated in Lines \ref{line:refine-2} and \ref{line:remove-2}.
Hence, if two vertices satisfy Corollary \ref{cor:2-vertex-connecte-pair}, they are contained in the same block of the block forest after the end of the algorithm.

Now we prove the opposite direction of the claim. Namely, we  
assume that at the end of the algorithm $x$ and $y$ are in the same block of the block forest and we wish to prove that $x$ and $y$ satisfy Corollary \ref{cor:2-vertex-connecte-pair}.
Since $x$ and $y$ are not separated in Line \ref{line:common-child-parent} there exists a pair of vertices $u,v$ such that $x,y\in C(u,v)$.
Now assume, without loss of generality, that $x$ and $y$ are siblings in $D$ (both children of some vertex $u$).
Then they are strongly connected in $G\setminus d(x)$ since they are not separated in Line \ref{line:refine-1}.
Now assume, without loss of generality, that $x=d(y)=u$.
Recall that all vertices in the same block $C$ as $y$ in the block forest, right after the execution of Line \ref{line:refine-1}, are strongly connected in $G \setminus x$.
By Lemma \ref{SAP-relation}, they are also strongly connected in $G \setminus d(x)$.
Since $x$ remains in the same block as $y$ in Line \ref{line:remove-1}, then $x$ is strongly connected with all vertices in the component of $y$ in $G\setminus d(x)$.
Therefore, $x$ and $y$ satisfy conditions (a) and (b) in Corollary \ref{cor:2-vertex-connecte-pair}.
The same arguments can be use to show that $x$ and $y$ satisfy (c) and (d) from Corollary \ref{cor:2-vertex-connecte-pair}, and 
the lemma follows.
\proofend
\end{proof}

In order to collect efficiently the sets of vertices in $C(u)$ (resp., $C^R(u)$), for some vertex $u$, that are strongly connected in $G\setminus u$, as required in Line \ref{line:gather-components} (resp., Line \ref{line:gather-components-2}) of Algorithm \textsf{Inc2VCC}, we use an additional data structure.
Recall the construction of the graph $\overline{G}=(\overline{V},\overline{E})$ from Section \ref{sec:extension-to-vertices}. 
For each strong articulation point $x$ of $G$, we add an auxiliary vertex $\overline{x} \in \overline{V}$ and add the auxiliary edges $(\overline{x},x)$ and $(x,\overline{x})$. 
Moreover, we replace each edge $(u, x)$ entering $x$ in $G$ with an edge $(u,\overline{x})$ entering $\overline{x}$ in $\overline{G}$.
We maintain incrementally the flow graph $\overline{G}_s$, with start vertex $s$, its dominator tree $\overline{D}$, its bridge decomposition $\overline{\mathcal{D}}$, and the auxiliary components of $\overline{G}_s$.
This will allow us to execute efficiently Lines \ref{line:gather-components} and \ref{line:gather-components-2} of Algorithm \textsf{Inc2VCC}, as suggested by the following lemma.

\begin{lemma}
	\label{lem:vertex-resilience-condition}
	Let $x$ and $y$ be two vertices that are siblings in $D$, and let $w=d(x)=d(y)$.
	Then, $x$ and $y$ are strongly connected in $G\setminus d(x)$ if and only if 
	$r_{\overline{x}} = r_{\overline{y}}$ and $c_{\overline{x}} = c_{\overline{y}}$, where $r_{\overline{x}}$ and $r_{\overline{y}}$ are the roots of the bridge decomposition of $\overline{G}_s$ containing $x$ and $y$, respectively, and $c_{\overline{x}}$ and $c_{\overline{y}}$ are the auxiliary components of $\overline{G}$ containing $x$ and $y$, respectively.
\end{lemma}
\begin{proof}
	We begin with the forward direction of the claim: namely, we show that if $x$ and $y$ are strongly connected in $G\setminus d(x)$ then $r_{\overline{x}} = r_{\overline{y}}$ and $c_{\overline{x}} = c_{\overline{y}}$.
	Notice, that since $\overline{x}$ and $\overline{y}$ are auxiliary vertices, 
	$\overline{d}(\overline{x})$ (resp., $\overline{d}(\overline{y})$) is an ordinary vertex and therefore $e'=(\overline{d}(\overline{d}(\overline{x})),\overline{d}(\overline{x}))$ (resp., $e'=(\overline{d}(\overline{d}(\overline{y})),\overline{d}(\overline{y}))$) is the only 
	edge entering $\overline{d}(\overline{x})$  (resp., $\overline{d}(\overline{y})$) and moreover it is a bridge in $G$.
	Now assume that $r_{\overline{x}} \not= r_{\overline{y}}$.
	Then 
	either $(\overline{d}(\overline{x}),\overline{x})$ 
	or $(\overline{d}(\overline{y}),\overline{y})$ is a strong bridge (or both).
	Without loss of generality, assume that $(\overline{d}(\overline{x}),\overline{x})$ is a strong bridge and $y$ is not a descendant of $\overline{x}$ in $D$.
	By definition of strong bridges, all paths from a vertex $v \notin \overline{D}(\overline{x})$ (which includes $y$) to $\overline{x}$ contain the strong bridge $(\overline{d}(\overline{x}),\overline{x})$ and therefore the vertex $\overline{d}(\overline{x})$.
	By Lemma \ref{lemma:vertex-to-edge-reduction}(ii),
	$\overline{d}(\overline{x}) = d(x)$.
	The fact that all paths from $y$ to $\overline{x}$ in $\overline{G}$ contain the vertex $\overline{d}(\overline{x})$ contradicts our initial assumption that $x$ and $y$ are strongly connected in $G\setminus d(x)$, as indicated by Lemma  \ref{lemma:vertex-to-edge-reduction}(ii).
	Hence, if $x$ and $y$ are strongly connected in $G\setminus d(x)$ then $r_{\overline{x}} = r_{\overline{y}}$.
	To complete the proof of this case, assume by contradiction that $x$ and $y$ are strongly connected in $G\setminus d(x)$ but $c_{\overline{x}} \not= c_{\overline{y}}$.
	Since neither $(\overline{d}(\overline{x}),\overline{x})$ nor $(\overline{d}(\overline{y}),\overline{y})$ are strong bridges and $e'=(\overline{d}(\overline{d}(\overline{x})),\overline{d}(\overline{x}))$ is a strong bridge in $G$, it follows that $r_{\overline{x}} = r_{\overline{x}} = \overline{d}(\overline{x})$.
	The fact that $c_{\overline{x}} \not= c_{\overline{y}}$ implies that $\overline{x}$ and $\overline{y}$ are not strongly connected in $\overline{G}[\overline{D}(\overline{d}(\overline{x}))]$, and therefore, they are not strongly connected in $\overline{G}\setminus e'$.
	By Lemma  \ref{lemma:vertex-to-edge-reduction}(iii), all vertices of a strongly connected component in $G\setminus d(x)$ are ordinary vertices of a strongly connected component in $\overline{G}\setminus e'$.
	This contradicts our assumption that $x$ and $y$ are strongly connected in $G\setminus d(x)$.
	Thus, if $x$ and $y$ are strongly connected in $G\setminus d(x)$ then it must be $r_{\overline{x}} = r_{\overline{y}}$ and $c_{\overline{x}} \not= c_{\overline{y}}$.
	
	Now we prove the other direction of the claim: namely, if $r_{\overline{x}} = r_{\overline{y}}$ and $c_{\overline{x}} \not= c_{\overline{y}}$ then $x$ and $y$ are strongly connected in $G\setminus d(x)$.
	Since $r_{\overline{x}} = r_{\overline{y}}$, it follows that neither $(\overline{d}(\overline{x}),\overline{x})$ nor $(\overline{d}(\overline{y}),\overline{y})$ is a strong bridges
	Moreover, by the fact that $c_{\overline{x}} \not= c_{\overline{y}}$ if follows that $\overline{x}$ and $\overline{y}$ are strongly connected in $\overline{G}\setminus e'$.
	By Lemma \ref{lemma:vertex-to-edge-reduction}, $x$ and $y$ must be strongly connected.
\proofend
\end{proof}

We are now ready to analyze the total running time of Algorithm \textsf{Inc2VCC}.

\begin{lemma}
	Algorithm \textsf{Inc2VCC} runs in $O(n)$ time.
\end{lemma}
\begin{proof}
The computation of all non-empty sets $C(u,v)$ in Line \ref{line:common-child-parent} takes $O(n)$ time as it is shown in~\cite{LuigiGILP15}.
Each vertex $v$ is contained in at most $4$ different initial blocks, i.e., the blocks $C(u,v),C(d(u),v),C(u,d(v)),C(d(u),d(v))$.
This implies that the initial block forest contains $O(n)$ edges, and therefore its initialization in Line \ref{line:child-parent} takes $O(n)$ time.
To find  all subsets of vertices of $C(u)$ that are strongly connected in $G\setminus u$ in Line \ref{line:gather-components}, we use Lemma \ref{lem:vertex-resilience-condition}, which implies that these blocks are the auxiliary components of $\overline{G}_s$ containing vertices $c_{\overline{x}}$, where $x\in C(u)$.
Those auxiliary components are maintained by running an instance of the incremental algorithm in \cite{GIN16:ICALP} on $\overline{G}$ in a total of time $O(mn)$, and they can be easily collected in time $O(|C(u)|)$.
As a result, we spend a total of $O(mn)$ time to maintain the auxiliary components of $\overline{G}$, and $O(n)$ time for all queries throughout one execution of Algorithm \textsf{Inc2VCC}.

The initial block forest in Line \ref{line:common-child-parent} is a forest since it is defined with respect to the sets $C(u,v)=\{C(u) \cup u\}\cap \{C^R(v) \cup v\}$, where both sets $C(u)\cup u$ and $C^R(v)\cup v$ form trees.
We now show that this is preserved by a $refine(\mathcal{B}, \mathcal{S}, u)$ operation.
Consider what happens to a block $B \in \mathcal{B}$.
This block is represented by a node $b$ in $\mathcal{F}$.
$B$ is either partitioned into several disjoint blocks, in which case the tree containing $b$ becomes a forest.
Otherwise, $B$ is replaced by sets $B_1,\dots,B_l$ where all sets share $u$.
In this case, the new adjacency lists of the nodes $b_1,\dots, b_l$ representing the sets $B_1,\dots, B_l$ share only the node corresponding to $u$.
Therefore, the tree containing $b$ remains a tree.
Since $\mathcal{F}$ is a forest, the sum of the cardinalities of the blocks $\mathcal{B}$ from Line \ref{line:refine-1} (resp., Line \ref{line:refine-2}) is at most $2|C(u)|$  (resp., $2|C^R(u)|$). To see this, just root each subtree of $\mathcal{F}$ that contains some of the blocks in $\mathcal{B}$ and their vertices: each node has one parent and $|\mathcal{B}|\leq |C(u)|$.
Hence, the $refine$ operations in Lines \ref{line:refine-1} and \ref{line:refine-2} are executed in time $O(|C(u)|)$ and $O(|C^R(u)|)$, respectively.
Thus, all refine operations take overall $O(n)$ time.
Finally, notice that we can answer each query in Lines \ref{line:query-1} and \ref{line:query-2} in constant time by a type (iii) query from Section \ref{sec:extension-to-vertices}.
This gives the lemma.
\proofend
\end{proof}

\begin{lemma}
	\label{lem:vertex-resilient-running-time}
We can maintain the vertex-resilient components of a directed graph $G$ through any sequence of edge insertions in a total of $O(mn)$ time, where $m$ is the number of edges after all edge insertions, and linear space.
\end{lemma}

\begin{theorem}
	We can maintain the $2$-vertex-connected components of a directed graph $G$ through any sequence of edge insertions in a total of $O(mn)$ time, where $m$ is the number of edges after all edge insertions, and linear space.
\end{theorem}
\begin{proof}
	The theorem follows from Lemmas \ref{cor:2vc-resilient} and \ref{lem:vertex-resilient-running-time} and from the fact that we can maintain incrementally the $2$-edge-connected components of a directed graph in a total of $O(mn)$ time \cite{GIN16:ICALP}.
	\proofend
\end{proof}

\section{Conditional lower bounds}
In this section we present conditional lower bounds for both the partially dynamic and the fully dynamic setting.
In particular, one of our lower bounds implies that a polynomial improvement over our bounds would have interesting consequences, as such an improvement would disprove widely believed conjectures.
More specifically, we show that there is no algorithm that can maintain either incrementally of decrementally a data structure allowing queries of the form ``are $u$ and $v$ strongly connected in $G\setminus e$?", where $u,v\in V, e\in E$, and has total update time $ O ((m n)^{1-\epsilon}) $ (for some constant $ \epsilon > 0 $) and sub-polynomial query time unless the online matrix-vector multiplication (OMv) conjecture~\cite{HenzingerKNS15} is false.
Hence, under the OMv conjecture the total running time of our algorithm, for this particular query, is asymptotically optimal.

In the fully dynamic version we show that, unless the strong exponential time hypothesis (SETH) is false, there is no algorithm maintaining a graph and being able to answer any of the queries that we consider, within the same asymptotic query time, with amortized update time $O(m^{1-\epsilon})$.
Finally, we show that in the incremental/decremental model, there is no algorithm that can maintain a data structure answering any of the queries we consider in this paper, in the same asymptotic query time, with worst-case update time $O(m^{1-\epsilon})$, for any $\epsilon>0$, unless the SETH is false.
Therefore, in these two cases, recomputing the data structure from~\cite{2C:GIP:arXiv} from scratch after every update achieves the best possible asymptotic update time.

\subsection{$\Omega(mn)$ conditional lower bound for the total update time in the partially dynamic model} 
In this section show a conditional lower bound for the total update time of a partially dynamic algorithm that either incrementally or decrementally maintains a data structure that can answer the queries ``are $ u $ and $ v $ strongly connected in $ G\setminus e $?'', where $u,v\in V$, $e\in E$.
Here, we show that there is neither incremental nor decremental algorithm for maintaining a data structure for answering the aforementioned that has total update time $ O ((m n)^{1-\epsilon}) $ (for some constant $ \epsilon > 0 $) and sub-polynomial query time unless the \textsf{OMv} Conjecture~\cite{HenzingerKNS15} fails.
This bound matches the total update time of our algorithms.
For our reduction we use a modification of the construction that was used in \cite{decdom17} to prove conditional lower bounds for partially dynamic algorithm for updating the dominator tree.
In what follows, we prove the following statement.

\begin{theorem}
	For any constant $ \delta \in (0, 1/2] $ and any $ n $ and $ m = \Theta (n^{1 / (1-\delta)}) $, there is no algorithm for maintaining a data structure under edge deletions/insertions allowing queries of the form ``are $ u $ and $ v $ strongly connected in $G\setminus e$'', where $u,v\in V$, $e\in E$, that uses polynomial preprocessing time, total update time $ u (m, n) = (m n)^{1-\epsilon} $ and query time $ q (m) = m^{\delta-\epsilon} $ for some constant $ \epsilon > 0 $, unless the \textsf{OMv} conjecture fails.
\end{theorem}

Under this conditional lower bound, the running time of our algorithm is optimal up to sub-polynomial factors.
We give the reduction for the incremental version of the problem.
The hardness of the decremental problem follows via an analogous reduction.

\smallskip
\noindent{\bf Hardness assumption.}
As in \cite{decdom17}, we consider the following \textsf{$\gamma$-OuMv} problem (for a fixed $ \gamma > 0 $) and parameters $ n_1 $, $ n_2 $, and $ n_3 $ such that $ n_1 = \lfloor n_2^\gamma \rfloor $:
We are first given a Boolean $ n_1 \times n_2 $ matrix $ M $ to preprocess. 
After the preprocessing, we are given a sequence of pairs of $n_1$-dimensional Boolean vectors $ (u^{(1)}, v^{(1)}), \dots, (u^{(n_3)}, v^{(n_3)}) $ one by one.
For each $ 1 \leq t \leq n_3 $, we have to return the result of the Boolean vector-matrix-vector multiplication $ (u^{(t)})^\intercal M v^{(t)} $ before we are allowed to see the next pair of vectors $ (u^{(t+1)}, v^{(t+1)}) $.
It has been shown~\cite{HenzingerKNS15} that under the \textsf{OMv} Conjecture as stated above, there is no algorithm for this problem that has polynomial preprocessing time and for processing all vectors spends total time $ n_1^{1-\epsilon_1} n_2^{1-\epsilon_2} n_3^{1-\epsilon_3} $ such that all $ \epsilon_i $ are $ \geq 0 $ and at least one $ \epsilon_i $ is a constant $ > 0 $.

\smallskip
\noindent{\bf Reduction.}
We now give the reduction from the \textsf{$\gamma$-OuMv} problem with $ \gamma = \delta / (1-\delta)$ to the incremental maintainance of a data structure that supports the queries ``are $ u $ and $ v $ strongly connected in $ G\setminus e $?'', where $u,v\in V$, $e\in E$.
In the following we denote by $ v_i $ the $i$-th entry of a vector $ v $ and by $ M_{i,j} $ the entry at row~$ i $ and column~$ j $ of a matrix~$ M $.

Consider an instance of the \textsf{$\gamma$-OuMv} problem with parameters $ n_1 = m^{1-\delta} $, $ n_2 = m^{\delta} $, and $ n_3 = m^{1-\delta} $.
We preprocess the matrix $ M $ by constructing a graph $ G^{(0)} $ with the set of vertices
\begin{equation*}
V = \{ s, x_0, x_1, \dots, x_{n_3}, y_1, \dots, y_{n_1}, z_1, \dots, z_{n_2} \}
\end{equation*}
and the following edges:
\begin{itemize}
	\item an edge $ (s, x_{n_3}) $, and, for every $ 1 \leq t \leq n_3 $, an edge $ (x_{t},x_{t-1}) $ (i.e., a path from $ s $ to $ x_0 $).
	\item for every $ 1 \leq j \leq n_2 $, an edge $ (x_0, z_j) $.
	\item for every $ 1 \leq i \leq n_1 $ and every $ 1 \leq j \leq n_2 $, an edge $ (y_i, z_j) $ if and only if $ M_{i,j} = 1 $ (i.e.\ a bipartite graph between $ \{ y_1, \ldots, y_{n_1} \} $ and $ \{ z_1, \ldots, z_{n_2} \} $ encoding the entries of matrix $ M $).
	\item an edge $ (z_j, s) $, for every $ 1 \leq j \leq n_2 $ (this makes the graph strongly connected).
\end{itemize}

Whenever the algorithm is given the next pair of vectors $ (u^{(t)}, v^{(t)}) $, we first create a graph $ G^{(t)} $ by performing the following edge insertions in $ G^{(t-1)} $:
If $t>1$, we first insert from $x_{t-1}$ an edge $ (x_{t-1}, y_i) $, for all $ 1 \leq i \leq n_1 $.
Then, for every $ i $ such that $ u^{(t)}_i = 1 $ we add the edge from $ x_t $ to $ y_i $.
Having created $ G^{(t)} $, we now, for every $ j $ such that $ v^{(t)}_j = 1 $, check whether $ s $ and $ z_j $  are strongly connected in $G^{(t)} \setminus (x_{t},x_{t-1})$.
If this is the case for at least one $ j $ we return that $ (u^{(t)})^\intercal M v^{(t)} $ is $ 1 $, otherwise we return $ 0 $.

\smallskip
\noindent{\bf Correctness.}
The correctness of our reduction follows from the following lemma.

\begin{lemma}\label{lem:technical lemma lower bound}
	For every $ 1 \leq t \leq n $, the $j$-th entry of $ (u^{(t)})^\intercal M $ is $ 1 $ if and only if $s$ and $ z_j $ are strongly connected in $ G^{(t)} \setminus (x_{t}, x_{t-1})$.
\end{lemma}

\begin{proof}
	First, notice that there is always a path from $z_j$ to $s$ since there is the edge $(z_j,s)$.
	If the $j$-th entry of $ (u^{(t)})^\intercal M $ is $ 1 $, then there is an $ i $ such that $ u^{(t)}_i = 1 $ and $ M_{i,j} = 1 $.
	Thus, $ G^{(t)} $ contains the edges $ (x_t, y_i) $ and $ (y_i, z_j) $ and consequently a cycle containing $ s $ and $ z_j $, namely $ \langle s, x_{n_3}, \ldots, x_t, y_i, z_j \rangle $.
	Therefore, $s$ and $z_j$ are strongly connected in $ G^{(t)} \setminus (x_t, x_{t-1})$.
	
	If the $j$-th entry of $ (u^{(t)})^\intercal M $ is $ 0 $, then there is no $ i $ such that $ u^{(t)}_i = 1 $ and $ M_{i,j} = 1 $.
	This implies that there is no path (of length $ 2 $) from $ x_t $ to $ z_j $ avoiding $ (x_t, x_{t-1}) $ (via some vertex $ y_i $).
	Moreover, notice that there is not edge $(x_l,y_i)$ for $t \leq l \leq n_3$ and $1\leq i \leq n_1$.
	Thus, every path from $ s $ to $ z_j $ contains $(x_t, x_{t-1})$.
	Thus, $ s $ cannot be the strongly connected with $ z_j $ in $ G^{(t)} \setminus (x_t,x_{t-1})$.
	\proofend
\end{proof}

Note that $ (u^{(t)})^\intercal M v^{(t)} $ is $ 1 $ if and only if there is a $ j $ such that both the $j$-th entry of $ u^{(t)} M $ as well as the $j$-th entry of $ v^{(t)} $ are $ 1 $.
Furthermore, $ s $ and  $ z_j $  are strongly connected in $ G^{(t)} \setminus (x_t,x_{t-1})$ if and only if $ s $ and $ z_j $ are strongly connected in the maintained graph minus the edge $(x_t,x_{t-1})$.
Therefore the lemma establishes the correctness of the reduction.

\smallskip
\noindent{\bf Complexity.}
The final graph $ G^{(t)} $ has $ n := \Theta (n_1 + n_2 + n_3) = \Theta (m^\delta + m^{1-\delta}) = \Theta (m^{1-\delta}) $ vertices and $ \Theta (n_1 n_2 + n_2 n_3) = \Theta (m) $ edges.
The total number of queries is $ O (n_1 n_3) = m^{2 (1-\delta)} $.
Suppose the total update time of the incremental algorithms that supports the required queries is $ O (u(m, n)) = (m n)^{1-\epsilon} $ and its query time is $ O (q(m)) = m^{\delta - \epsilon} $.
Using the reduction above, we can thus solve the \textsf{$\gamma$-OuMv} problem for the parameters $ n_1, n_2, n_3 $ with polynomial preprocessing time and total update time
\begin{equation*}
O (u (m, n) + m^{2 (1-\delta)} q (m)) = O (u (m, m^{1-\delta}) + m^{2 (1-\delta)} q (m)) = O (m^{2-\delta - \epsilon}) \, .
\end{equation*}
Since $ n_1 n_2 n_3 = m^{2-\delta} $, this means we would get an algorithm for the \textsf{$\gamma$-OuMv} problem with polynomial preprocessing time and total update time $ n_1^{1-\epsilon_1} n_2^{1-\epsilon_2} n_3^{1-\epsilon_3} $ where at least one $ \epsilon_i $ is a constant $ > 0 $.
This contradicts the \textsf{OMv} Conjecture.

\subsection{$\Omega(m)$ conditional lower bound for the amortized update time in the fully dynamic setting}

Now we prove a conditional lower bounds for the fully dynamic setting and for the partially dynamic setting with worst-case update time guarantees. 
More specifically, we show that for those two models, under the assumption of the strong exponential time hypothesis, the trivial algorithm that recomputes from scratch the solution using the static algorithm is asymptotically optimal up to sub-polynomial factors.
We base our bounds on the following conjecture that was first 
stated in \cite{impagliazzo1999complexity,impagliazzo1998problems}.

\begin{conjecture}[Strong Exponential Time Hypothesis (SETH)]
\label{con:SETH}
For every $\epsilon > 0$, there exists a $k$, such that SAT on $k$-CNF formulas on n variables cannot be solved in $O(2^{(1-\epsilon)n} poly(n))$ time.
\end{conjecture}

As an intermediate step in our reductions, we use the following result from~\cite{AW14}.

\begin{theorem}[\cite{AW14}]
	\label{lem:SC2-lower-bound}
	Let $G$ be a digraph with $n$ vertices that undergoes $m$ edge updates from an initially empty graph (until the graph is empty, in the case of a decremental algorithm). If for some $\epsilon > 0$ and $t \in N$, there exists either
	\begin{itemize}
		\item a fully dynamic algorithm with preprocessing time $O(n^t)$, amortized update time $O(m^{1-\epsilon})$, and amortized query time $O(m^{1-\epsilon})$, or
		\item an incremental or decremental algorithm with preprocessing time $O(n^t)$, worst-case update time $O(m^{1-\epsilon})$, and worst-case query time $O(m^{1-\epsilon})$
	\end{itemize}
that can answer either the query ``are there more that two or more SCCs in $G$?" or the query ``what is the size of the largest SCC in $G$?", then Conjecture~\ref{con:SETH} is false.
\end{theorem}

Now we show the reduction from the problems of Lemma~\ref{lem:SC2-lower-bound} to the problem of maintaining a data structure answering any of the queries that we consider in Theorem~\ref{thrm:main-result}.
We break our result into two parts.
In the first part, we prove that the result for three of the totally five types of queries without any assumptions.
In the second part, where we assume that the number of edges is in the graph is superlinear to the number of vertices.
This assumption about the density of the graph is necessary as our reduction in this spends $O(n)$ additional time per query.
We start with the first case where we make no assumptions about the density of the graph.

\begin{lemma}
\label{lem:lower-bound-fully-1}
Let $G=(V,E)$ be a digraph with $n$ vertices that undergoes $m$ edge updates from an initially empty graph (until the graph is empty, in the case of a decremental algorithm), $\epsilon>0$.
Assume that there exists a dynamic algorithm that can answer any of the following queries, either incrementally, decrementally, or fully dynamically:
\begin{itemize}	
	\item[(i)] \label{query1} Report in $O(m^{1-\epsilon})$ time the total number of SCCs in $G\setminus e$ (resp., $G\setminus v$), for any query edge $e$  (resp., vertex $v$) in $G$.
	\item[(ii)] \label{query2} Report in $O(m^{1-\epsilon})$ time the size of the largest SCC in $G\setminus e$ (resp., $G\setminus v$), for any query edge $e$  (resp., vertex $v$) in $G$.
	\item[(iii)] \label{query3} Report in $O(m^{1-\epsilon})$ time all the SCCs of $G\setminus e$ (resp., $G\setminus v$), for any query edge $e$ (resp., vertex $v$).
\end{itemize}
Then, there exists a dynamic algorithm, in the same model (incremental, decremental, or fully dynamic), that can answer either the query ``are there more that two SCCs in the graph?" or ``what is the size of the largest SCC in the graph?'' with the same preprocessing time, the same update time, and can answer queries in $O(m^{1-\epsilon})$ time.
\end{lemma}
\begin{proof}
It is safe to assume that $O(n)\in O(m)$ as if $m\leq n-3$ we can answer immediately that $G$ has more than two SCCs, as each SCC $C$ has at least $|C|-1$ edges.
We construct the following graph $G'=(V',E')$ from $G=(V,E)$.
$G'$ consists of the edges and vertices of $G$ and additionally two vertices $s_1,s_2$, the edge $(s_1,s_2)$, for each $v\in V$ an edge $(v,s_1)$ and an edge $(s_2,v)$.
Notice that $|V'|=|V|+2 \in O(|V|)$ and $|E'|=|E|+2\cdot n + 1 \in O(|E|)$. Our overall strategy for proving the lemma is to argue that at any time during the course of the dynamic algorithm we can answer the queries ``are there more that two SCCs in $G$?" or ``what is the size of the largest SCC in $G$?'' by answering any one of the queries in the statement of the lemma.
That is, we answer those queries independently of the type of updates (i.e., incremental, decremental, or fully dynamic algorithm), and of the type of query bound (i.e., either worst-case per update, or amortized over all updates).
Specifically, we show that by answering the queries of the lemma withing the stated bounds, then we can answer the queries ``are there more that two SCCs in $G$?" or ``what is the size of the largest SCC in $G$?'' in time $O(m^{1-\epsilon})$.

First, observe that two vertices $u$ and $v$ are strongly connected in $G'\setminus(s_1,s_2)$ or in $G'\setminus s_1$ if and only if they are strongly connected in $G$.
This is true as any new path from $u$ to $v$ or from $v$ to $u$, that uses any of the new edges that we inserted, contains vertex $s_1$ and edge $(s_1,s_2)$.
Therefore, by deleting either the vertex $s_1$ or the edge $(s_1,s_2)$ we destroy all new paths from $u$ to $v$ and from $v$ to $u$.

We begin with type (i) queries.
Assume that we can report the number of SCCs in $G'\setminus (s_1,s_2)$ (resp., in $G'\setminus s_1$) in time $O(m^{1-\epsilon})$.
Let this number be $c$.
Note that the SCCs among vertices in $V$ remain unchanged, and $s_1$ and $s_2$ form singleton SCCs in $G'\setminus (s_1,s_2)$ (resp., $s_2$ forms a singleton SCC in $G'\setminus s_1$).
It follows that the number of SCCs in $G$ is $c-2$ (resp., $c-1$).
Hence we can answer whether $G$ has more that two SCCs in time $O(m^{1-\epsilon})$.
The same argument applies for type (ii) queries that report all SCCs in $G'\setminus (s_1,s_2)$ (resp., in $G'\setminus s_1$) in time $O(m^{1-\epsilon})$.
The number of SCCs in $G$ can be extracted in additional $O(n)\in O(m^{1-\epsilon})$ time (recall that we assume $O(n)\in O(m)$).

Now we consider type (iii) queries.
Since the SCCs in $G'\setminus (s_1,s_2)$ (resp., in $G'\setminus s_1$) correspond to the SCCs of $G$ plus the singleton SCCs $s_1$ and $s_2$ (resp., the singleton SCC $s_2$), the largest SCC in $G'\setminus (s_1,s_2)$ (resp., in $G'\setminus s_1$) has the same size with the largest SCC in $G$.
Thus, if we can answer the size of the largest SCC in $G'\setminus (s_1,s_2)$ (resp., in $G'\setminus s_1$) in time $O(m^{1-\epsilon})$, then we can answer the size of the largest SCC in $G$ in $O(m^{1-\epsilon})$.
\proofend
\end{proof}

\begin{theorem}
	Let $G$ be a digraph with $n$ vertices that undergoes $m$ edge updates from an initially empty graph (until the graph is empty, in the case of a decremental algorithm).
	If for some $\epsilon > 0$, there exists an algorithm that can answer the following queries:
	\begin{itemize}	
		\item[(i)] \label{query1} Report in $O(m^{1-\epsilon})$ time the total number of SCCs in $G\setminus e$ (resp., $G\setminus v$), for any query edge $e$  (resp., vertex $v$) in $G$.
		\item[(ii)] \label{query2} Report in $O(m^{1-\epsilon})$ time the size of the largest SCC in $G\setminus e$ (resp., $G\setminus v$), for any query edge $e$  (resp., vertex $v$) in $G$.
		\item[(iii)] \label{query3} Report in $O(m^{1-\epsilon})$ time all the SCCs of $G\setminus e$ (resp., $G\setminus v$), for any query edge $e$ (resp., vertex $v$).
	\end{itemize}
	while maintaining $G$ fully dynamically with amortized update time $O(m^{1-\epsilon})$ and amortized query time after polynomial time preprocessing, or partially dynamically with worst-case update time $O(m^{1-\epsilon})$ and worst-case query time after polynomial time preprocessing, then Conjecture~\ref{con:SETH} is false.
\end{theorem}

We now proceed to prove similar results of the last two query types, in graphs where the number of edges is superlinear to the number of vertices.

\begin{lemma}
	Let $G=(V,E)$ be a digraph with $n$  vertices that undergoes $m>n^{1+\delta}$ edge updates from an initially empty graph (until the graph is empty, in the case of a decremental algorithm), $\delta>\epsilon/(1-\epsilon)$, $1/2>\epsilon>0$.
	Assume that there exists a dynamic algorithm that can answer any of the following queries, either incrementally, decrementally, or fully dynamically:
	\begin{enumerate}
		\item[(iv)] \label{query4} Test in $O(m^{1-\epsilon}/n)$ time whether two query vertices $u$ and $v$ are strongly connected in $G\setminus e$ (resp., $G\setminus v$), for any query edge $e$ (resp., vertex $v$).
		\item[(v)] \label{query5} For any two query vertices $u$ and $v$ that are strongly connected in $G$, test whether there exists an edge $e$ (resp., vertex $v$) such that $u$ and $v$ are not strongly connected in $G\setminus e$ (resp., $G\setminus v$) in time $O(m^{1-\epsilon}/n)$.
	\end{enumerate}
	Then, there exists a dynamic algorithm, in the same model (incremental, decremental, or fully dynamic), that can answer either the query ``are there more that two SCCs in the graph?" with the same preprocessing time, the same update time, and can answer queries in $O(m^{1-\epsilon})$ time.
\end{lemma}
\begin{proof}
	We use the same construction as in the proof of Lemma~\ref{lem:lower-bound-fully-1},
	That is, we construct the following graph $G'=(V',E')$ from $G=(V,E)$.
	$G'$ consists of the edges and vertices of $G$ and additionally two vertices $s_1,s_2$, the edge $(s_1,s_2)$, for each $v\in V$ an edge $(v,s_1)$ and an edge $(s_2,v)$.
	Notice that $|V'|=|V|+2 \in O(|V|)$ and $|E'|=|E|+2\cdot n + 1 \in O(|E|)$. Our overall strategy for proving the lemma is to argue that at any time during the course of the dynamic algorithm we can answer the queries ``are there more that two SCCs in $G$?" by answering any one of the queries in the statement of the lemma.
	That is, we answer those queries independently of the type of updates (i.e., incremental, decremental, or fully dynamic algorithm), and of the type of query bound (i.e., either worst-case per update, or amortized over all updates).
	Specifically, we show that by answering the queries of the lemma withing the stated bounds, then we can answer the queries ``are there more that two SCCs in $G$?'' in time $O(m^{1-\epsilon}) \in \omega(n^{1+O(1)})$.
	
	For type (v) queries we claim that for any two $u,v\in V$ there exists a separating edge (resp., vertex) in $G'$ if and only in $u$ and $v$ are not strongly connected in $G$.
	We now prove this claim. 
	First, assume that $u$ and $v$ are strongly connected in $G$.
	$G'$ contain the paths $\langle u, s_1,s_2,v\rangle$ and $\langle v, s_1,s_2,u\rangle$, and those paths exists if we delete any edge from $E$.
	Similarly, if we delete any edge from $E'\setminus E$, the vertices $u$ and $v$ will remain strongly connected as they are strongly connected in $(V',E)$.
	Second, assume that $u$ and $v$ are not strongly connected in $G$.
	Then, either there is no path from $u$ to $v$ or there is no path from $v$ to $u$ in $G$.
	Assume, w.l.o.g., that there is no path from $u$ to $v$ in $G$.
	Since all new path from $u$ to $v$ contain the edge $(s_1,s_2)$, then $(s_1,s_2)$  (resp., the vertex $s_1$ and the vertex $s_2$) is a separating edge (resp., vertex) for $u$ and $v$.
	Our claim follows.
	
	We follow a similar approach with type (iv) queries.
	Assume that we can test for any two query vertices $u$ and $v$ that are strongly connected in $G$, whether there exists an edge $e$ (resp., vertex $v$) such that $u$ and $v$ are not strongly connected in $G\setminus e$ (resp., $G\setminus v$) in time $O(\max\{m^{1-\epsilon}/n,1\})$.
	We pick an arbitrary vertex $x$ of $G$ and we test whether there exists a separating edge (resp., a separating vertex) for $x$ and $v$, for every $v\in V\setminus x$.
	This requires $O(m^{1-\epsilon})$ time in total.
	Let $X$ be the set of vertices for which there exists no separating edge (resp., vertex) with $x$.
	Then we pick an arbitrary vertex $y$ from $V\setminus X$ and we test whether there exists a separating edge (resp., vertex) for $y$ and $v$, for every $v\in V \setminus \{X\cup y\}$.
	This requires additional $O(m^{1-\epsilon})$ time.
	Let $Y$ be the set of vertices for which there exists no separating edge (resp., vertex) with $y$.
	If $|X| + |Y| = |V|$, then $G$ has two SCCs, namely the SCC formed by the vertices in $X$ and the SCC formed by the vertices in $Y$.
	If, on the other hand $|X|+|Y|<|V|$, then there exists at least one vertex in $V$ that is neither strongly connected to $x$ nor strongly connected to $y$ in $G$, and hence, $G$ contains at least three SCCs.
	\proofend
\end{proof}

\begin{theorem}
	Let $G$ be a digraph with $n$ vertices that undergoes $m>n^{1+\delta}$ edge updates from an initially empty graph (until the graph is empty, in the case of a decremental algorithm), $\delta>\epsilon/(1-\epsilon)$.
	If for some $\epsilon > 0$ there exists an algorithm that can answer the following queries:
	\begin{itemize}	
		\item[(iv)] \label{query4} Test in $O(\max\{m^{1-\epsilon}/n,1\})$ time whether two query vertices $u$ and $v$ are strongly connected in $G\setminus e$ (resp., $G\setminus v$), for any query edge $e$ (resp., vertex $v$).
		\item[(v)] \label{query5} For any two query vertices $u$ and $v$ that are strongly connected in $G$, test whether there exists an edge $e$ (resp., vertex $v$) such that $u$ and $v$ are not strongly connected in $G\setminus e$ (resp., $G\setminus v$) in time $O(\max\{m^{1-\epsilon}/n,1\})$.
	\end{itemize}
	while maintaining $G$ fully dynamically with amortized update time $O(m^{1-\epsilon})$ and amortized query time, or partially dynamically with worst-case update time $O(m^{1-\epsilon})$ and worst-case query time, after arbitrary polynomial time preprocessing, then Conjecture~\ref{con:SETH} is false.
\end{theorem}

\bibliographystyle{plain}
\bibliography{ltg}

\newpage

\begin{appendix}
\thispagestyle{empty}
\begin{table}[!ht]
\begin{small}
\tabulinesep=0.5mm
\setlength{\tabcolsep}{4.0pt} 
\begin{tabu}[t]{|p{3cm}|p{11cm}|} 
\hline
  \textbf{Symbol} & \textbf{Description}\\
 \hline
 $G^R$ & The graph resulting from $G$ after reversing the direction of all edges.\\
 \hline
 $G_s$ (resp., $G^R_s$)& A flow graph $G$ (resp., $G^R$), with start vertex $s$, where all vertices are reachable from $s$ (resp., reach $s$). \\
 \hline
 $V(G)$ (resp., $E(G)$) & The set of vertices (resp., edges) of $G$.\\
 \hline
 $G\setminus v, v \in V$ & $G$ after deleting $v$ together with and all its incident edges.\\
 \hline
 $G\setminus e, e \in E$ & The graph resulting from $G$ after deleting edge $e$.\\
 \hline
 $G[S], S \subseteq V$ & The subgraph of $G$ induced by the vertices in $S$.\\
 \hline
 $T(u)$, where $T$ a tree and $u\in V(T)$ & The set of vertices in the subtree rooted of $T$ at $u$.\\
 \hline
 $T[u,v]$, where $T$ a tree and $u,v\in V(T)$ & The path between $u$ and $v$ in $T$ (including $u$ and $v$).\\
 \hline
 $nca_T(u,v)$, where $T$ a tree and $u,v\in V(T)$  & The nearest common ancestor of $u$ and $v$ in $T$.\\
 \hline
 $D$ (resp., $D^R$) of $G_s$ (resp., of $G^R_s$) & The dominator tree $D$ (resp., $D^R$) of $G_s$ (resp., $G^R_s$).\\
 \hline
 $d(u)$ (resp., $d^R(u)$) & The parent of $u$ in $D$ (resp., $D^R$).\\
 \hline
 $dom(u)$ (resp., $dom^R(u)$)& The set of ancestors of $u$ in $D$ (resp., $D^R$).\\
 \hline
 $H$ (resp., $H^R$) of $G_s$ (resp., of $G^R_s$) & The loop nesting tree $H$ (resp., $H^R$) of $G_s$ (resp., $G^R_s$)\\
 \hline
 $h(u)$ (resp., $h^R(u)$) & The parent of $u$ in  $H$ (resp., $H^R$).\\
 \hline
 $h_u$ (resp., $h^R_u$) & The nearest ancestor of $u$ in the loop nesting tree $H$ (resp., $H^R$) such that $h_u\in D(r_u)$ and $h(h_u)\notin D(r_u)$ (resp., $h^R_u\in D^R(r^R_u)$ and $h^R(h^r_u)\notin D^R(r^R_u)$).\\
 \hline
 Bridge decomposition $\mathcal{D}$ (resp., $\mathcal{D}^R$) & The resulting forest after deleting from $D$ (resp., $D^R$) all bridges of $G_s$ (resp., $G^R_s$).\\
 \hline
 $r_u$ (resp., $r^R_u$) & The root of the tree in $\mathcal{D}$ (resp., $\mathcal{D}^R$) containing $u$.\\
 \hline
 $c_u$ (resp., $c^R_u$) & The canonical vertex of the auxiliary component of $u$ in $G_s$  (resp., $G^R_s$).\\
 \hline
 $L$ (resp., $L^R$) of $G_s$ (resp., of $G^R_s$) & The \newLNT{} tree $L$ (resp., $L^R$) of $G_s$ (resp., $G^R_s$).\\
 \hline
 $\ell(u)$ (resp., $\ell^R(u)$) & The parent of $u$ in the \newLNT{} tree $L$ (resp., $L^R$).\\
 \hline
 $\hat{H}$ (resp., $\hat{H}^R$) of $G_s$ (resp., of $G^R_s$) & The tree resulting from $L$ (resp., $\hat{L}^R$) of $G_s$ (resp., $G^R_s$) after making each vertex $u\not=c_u$ (resp., $u\not=c^R_u$) child of $c_u$ (resp., $c^R_u$).\\
 \hline
 $\hat{h}(u)$ (resp., $\hat{h}^R(u)$) & The parent of $u$ in $\hat{H}$ (resp., $\hat{H}^R$).\\
 \hline
 $level(u)$ (resp.,  $level^R(u))$ & The number of strong bridges on $D[s,u]$ (resp., $D^R[s,u]$) where $s$ is the start vertex of the flow graph on which $D$ (resp., $D^R$) is defined.\\
 \hline
 $f'$, where $f$ any relation & The relation $f$ after the insertion of an edge (the inserted edge is usually specified, or we refer to any edge insertion).\\
 \hline
 $D$-scanned vertices $S$ & The vertices that increase their depth in $D$ after an edge insertion.\\
 \hline
 $G_{scanned}$ & The graph induced by the vertices in $S$ (the set of $D$-scanned vertices).\\
 \hline
 $\tilde{H}$ & The loop nesting tree of $G_{scanned}$ rooted at $y$ after the deletion of $(x,y)$\\
 \hline
 $\tilde{h}(u)$ & the parent of $u$ in $\tilde{H}$.\\
 \hline
 $D$-affected & The vertices that change parent in $D$ after an edge insertion.\\
 \hline
 $L$-affected $S'$ & The set of vertices $S$ that change parent in $L$ after an edge insertion.\\
 \hline
 $\widetilde{D}(u)$ (resp., $\widetilde{D}^R(u)$)& $D(u)\setminus u$ (resp., $D^R(u)\setminus u$).\\
 \hline
 $\overline{G}=(\overline{V},\overline{E})$ & For each strong articulation point $x$ of $G$, we add an auxiliary vertex $\overline{x} \in \overline{V}$ and add the auxiliary edges $(\overline{x},x)$ and $(x,\overline{x})$. 
 Moreover, we replace each edge $(u, x)$ entering $x$ in $G$ with an edge $(u,\overline{x})$ entering $\overline{x}$ in $\overline{G}$.\\
 \hline
 $\mathcal{F}$ & The block forest, as defined in \cite{2VCB}. A data structure for representing overlapping sets of vertices.\\
 \hline
 \end{tabu}
\end{small}
\caption{The notation that we use throughout the paper. We exclude notation that is used briefly for further definitions (e.g., $loop(u)$ for a vertex $u$, which is used to define the loop nesting tree of the graph).}
\label{tab:notation}
\end{table}
\end{appendix}

\end{document}